\RequirePackage{snapshot}
\documentclass{vldb}

\usepackage{booktabs} \usepackage{xcolor}
\usepackage{todonotes}

\usepackage{amsmath}
\usepackage{amssymb}

\usepackage{amsthm}
\usepackage{bm}
\usepackage{mathtools}
\usepackage{array}
\usepackage{tabularx}
\usepackage{multirow}
\usepackage{color}
\usepackage{listings}
\usepackage{colortbl}
\usepackage{hyperref}
\usepackage{xspace}
\usepackage{environ}
\usepackage{etoolbox}
\usepackage{subcaption}
\usepackage{rotating}
\usepackage{times}
\usepackage{cite}
\usepackage[normalem]{ulem}
\usepackage{enumitem}
\usepackage{balance}

\usepackage[labelfont=bf]{caption}

\usepackage{tikz}

\usetikzlibrary{arrows,decorations.pathmorphing,positioning}
\usetikzlibrary{shapes.geometric}
\usetikzlibrary{calc}

\usepackage{footnote}
\makesavenoteenv{tabular}
\makesavenoteenv{figure}
\usepackage{tablefootnote}

\usepackage{adjustbox}
\usepackage{xcolor,colortbl}

\newbool{CompileTechReport}
\newcommand{\detailedproof}[2]{\ifbool{CompileTechReport}{
\begin{proof}
#2
\end{proof}
}{\noindent\textit{Proof Sketch:}#1\qed}\smallskip}
\newcommand{\ifnottechreport}[1]{\ifbool{CompileTechReport}{}{#1}}
\newcommand{\iftechreport}[1]{\ifbool{CompileTechReport}{#1}{}}

\booltrue{CompileTechReport}

\newrobustcmd{\revdel}[1]{}
\newrobustcmd{\reva}[1]{#1}
\newrobustcmd{\revb}[1]{#1}
\newrobustcmd{\revc}[1]{#1}
\newrobustcmd{\revm}[1]{#1}

\newrobustcmd{\crtodo}[1]{}

\usepackage{hhline}

\newtheorem{theo}{Theorem}
\numberwithin{theo}{section}
\newtheorem{lem}[theo]{Lemma}

\newtheorem{coll}[theo]{Corollary}
\newtheorem{exam}{Example}
\numberwithin{exam}{section}
\newtheorem{defi}{Definition}
\numberwithin{defi}{section}

\newcommand{\myeop}{\hspace*{\fill}\mbox{$\Box$}}

\newcommand{\proofsketch}[1]{\ifnottechreport{\noindent{\it Proof Sketch.}$\,${#1}\qed\smallskip}}

\newcommand{\thead}[1]{{\textbf{#1}}}
\newcommand{\tthead}[1]{{\textbf{\texttt{#1}}}}

\colorlet{LightBlue}{blue!30!}
\colorlet{LightRed}{red!20!}
\newcommand{\annotheadcell}[1]{{\cellcolor{LightBlue}{\underline{#1}}}}
\newcommand{\annotcell}[1]{{\cellcolor{LightBlue}{#1}}}

\newcommand{\bugcell}[1]{{\cellcolor{LightRed}{#1}}}

\DeclareRobustCommand{\BGDel}[2]{}

\definecolor{black}{rgb}{0,0,0}
\definecolor{grey}{rgb}{0.8,0.8,0.8}
\definecolor{red}{rgb}{1,0,0}
\definecolor{green}{rgb}{0,1,0}
\definecolor{darkgreen}{rgb}{0,0.5,0}
\definecolor{darkpurple}{rgb}{0.5,0,0.5}
\definecolor{darkdarkpurple}{rgb}{0.3,0,0.3}
\definecolor{blue}{rgb}{0,0,1}
\definecolor{shadegreen}{rgb}{0.95,1,0.95}
\definecolor{shadeblue}{rgb}{0.95,0.95,1}
\definecolor{shadered}{rgb}{1,0.85,0.85}
\definecolor{shadegrey}{rgb}{0.85,0.85,0.85}
\definecolor{oddRowGrey}{rgb}{0.80,0.80,0.80}
\definecolor{evenRowGrey}{rgb}{0.85,0.85,0.85}

\newcommand{\mathtext}[1]{\thickspace\text{#1}\thickspace}
\newcommand{\mathtab}{\thickspace\thickspace\thickspace}

\newcommand{\myproofpar}[1]{\noindent\underline{{#1}:}}

\newcommand{\abbrAGB}{AG}
\newcommand{\abbrBDB}{BD}
\newcommand{\SQLrel}{SQL period relation}
\newcommand{\SQLrels}{SQL period relations}
\newcommand{\SKrel}{snapshot $\semK$-relation}
\newcommand{\SKrels}{snapshot $\semK$-relations}

\newcommand{\SKdbs}{snapshot $\semK$-databases}
\newcommand{\CapSKdbs}{Snapshot $\semK$-databases}
\newcommand{\periodKrel}{period $\semK$-relation}
\newcommand{\periodKrels}{period $\semK$-relations}

\newcommand{\rel}{R}
\newcommand{\relSchema}{{\bf R}}
\newcommand{\arity}{arity}
\newcommand{\attr}{A}
\newcommand{\db}{D}
\newcommand{\dbSchema}{{\bf D}}
\newcommand{\dbDom}{\mathcal{DB}}
\newcommand{\relDom}{\mathcal{R}}
\newcommand{\dbDomOf}[1]{\dbDom_{\dbSchema}}
\newcommand{\valDom}{\mathcal{U}}
\newcommand{\aSchema}{SCH}
\newcommand{\schemaOf}[1]{\aSchema(#1)}
\newcommand{\query}{Q}

\newcommand{\tuple}{t}

\newcommand{\uDom}{\mathcal{U}}
\newcommand{\semK}{K}
\newcommand{\semN}{{\mathbb{N}}}
\newcommand{\semB}{\mathbb{B}}

\newcommand{\semPlus}[1]{+_{#1}}
\newcommand{\addK}{\semPlus{\semK}}
\newcommand{\semMult}[1]{\cdot_{#1}}
\newcommand{\multK}{\semMult{\semK}}
\newcommand{\semMonus}[1]{-_{#1}}
\newcommand{\monK}{\semMonus{\semK}}
\newcommand{\semZero}[1]{0_{#1}}
\newcommand{\zeroK}{\semZero{\semK}}
\newcommand{\semOne}[1]{1_{#1}}
\newcommand{\oneK}{\semOne{\semK}}
\newcommand{\homo}{h}
\newcommand{\naturalOrder}{\preceq}
\newcommand{\naturalStrictOrder}{\prec}
\newcommand{\naturalgeq}{\succeq}

\newcommand{\kDBDom}[2]{\dbDom_{#1}}
\newcommand{\kRelDom}[2]{\relDom_{#1,#2}}

\newcommand{\semTimeNI}{\semK_{\anyTE}}
\newcommand{\semTimeNIof}[1]{{#1}_{\anyTE}}

\newcommand{\zeroNI}{0_{\semTimeNI}}
\newcommand{\oneNI}{1_{\semTimeNI}}
\newcommand{\multNI}{\cdot_{\semTimeNI}}
\newcommand{\addNI}{+_{\semTimeNI}}
\newcommand{\monNI}{-_{\semTimeNI}}
\newcommand{\multP}{\cdot_{\semK_{\mathcal{P}}}}
\newcommand{\addP}{+_{\semK_{\mathcal{P}}}}
\newcommand{\monP}{-_{\semK_{\mathcal{P}}}}

\newcommand{\semTimeNIN}{\semTimeNIof{\semN}}
\newcommand{\addNIN}{+_{\semTimeNIN}}

\newcommand{\monNIN}{-_{\semTimeNIN}}
\newcommand{\addPN}{+_{\semN_{\mathcal{P}}}}

\newcommand{\domTimeRel}[1]{\dbDom_{#1_{\anyTE}}}

\newcommand{\timeDomain}{\mathbb{T}}
\newcommand{\timeSlice}[1]{\tau_{#1}}

\newcommand{\tLeq}{\leq_{\timeDomain}}
\newcommand{\tLe}{<_{\timeDomain}}
\newcommand{\tPoint}{T}

\newcommand{\tMin}{\tPoint_{min}}
\newcommand{\tMax}{\tPoint_{max}}

\newcommand{\anyTE}{\mathcal{T}}

\newcommand{\cTEDom}[1]{\mathbb{TE}_{{#1}}}
\newcommand{\nTEDom}[1]{\mathbb{TEC}_{{#1}}}

\newcommand{\interval}{I}
\newcommand{\iBegin}[1]{{#1}^+}
\newcommand{\iEnd}[1]{{#1}^-}

\newcommand{\adjacent}[2]{adj(#1,#2)}
\newcommand{\tBegin}{\tPoint_b}
\newcommand{\tEnd}{\tPoint_e}

\newcommand{\intervalDom}{\mathbb{I}}

\newcommand{\intervalEq}{\sim}
\newcommand{\coalesce}{\mathcal{C}}
\newcommand{\CPs}[1]{CP(#1)}

\newcommand{\CPIs}[1]{CPI(#1)}
\newcommand{\kCoalesce}[1]{\coalesce_{{#1}}}
\newcommand{\normalize}{\mathcal{N}}

\newcommand{\tSlice}[1]{\tau_{#1}}

\newcommand{\reprDomain}{\mathcal{E}}
\newcommand{\repr}{\textsc{Enc}}
\newcommand{\reprInv}[1]{{{\repr_{#1}}^{-1}}}

\newcommand{\reprRewr}{\textsc{Rewr}}

\newcommand{\reprN}{\textsc{PeriodEnc}}
\newcommand{\attrBegin}{\attr_\mathit{begin}}
\newcommand{\attrEnd}{\attr_\mathit{end}}

\newcommand{\ra}{\mathcal{RA}}
\newcommand{\raPlus}{\ra^{+}}
\newcommand{\raAgg}{\ra^{agg}}
\newcommand{\qClass}{\mathcal{C}}
\newcommand{\projection}{\Pi}
\newcommand{\selection}{\sigma}
\newcommand{\aAggregation}{\gamma}
\newcommand{\aggregation}[2]{{}_{#1}\aAggregation_{#2}}
\newcommand{\union}{\cup}

\newcommand{\rename}{\rho}

\newcommand{\join}{\bowtie}

\newcommand{\coalesceOp}{\mathcal{C}}
\newcommand{\normalizeOp}{\mathcal{N}}

\newcommand{\TimeOut}{\textcolor{red}{\textbf{TO (2h)}}}
\newcommand{\OOTS}{\textcolor{red}{\textbf{OOTS}}}
\newcommand{\ONA}{\textcolor{red}{\textbf{N/A}}}

\newcommand{\qAggEx}{\ensuremath\query_\mathit{onduty}}
\newcommand{\qDiffEx}{\ensuremath\query_\mathit{skillreq}}

\newcommand{\parttitle}[1]{\smallskip\textbf{#1.}}

\newcommand{\card}[1]{| {#1} |}

\ifnottechreport{\title{Snapshot Semantics for Temporal Multiset Relations}}
\iftechreport{\title{Snapshot Semantics for Temporal Multiset Relations\\(Extended Version)}}

\numberofauthors{1}
\author{Anton Dign\"os$^{1}$, Boris Glavic$^{2}$, Xing Niu$^{2}$, Michael B\"ohlen$^{3}$, Johann Gamper$^{1}$\\[2mm]
\affaddr{Free University of Bozen-Bolzano$^{1}$}
\hspace{5mm} \affaddr{Illinois Institute of Technology$^{2}$}
\hspace{5mm} \affaddr{University of Zurich$^{3}$}\\[2mm]
\{dignoes,gamper\}@inf.unibz.it
\hspace{5mm} \{bglavic@, xniu7@hawk.\}iit.edu
\hspace{5mm} boehlen@ifi.uzh.ch
}

\vldbTitle{} \vldbAuthors{} \vldbDOI{https://doi.org/10.14778/3311880.3311882}
\vldbVolume{12}
\vldbNumber{6}
\vldbYear{2019}

\begin{document}

\sloppy
\maketitle

\definecolor{lstpurple}{rgb}{0.5,0,0.5}
\definecolor{lstred}{rgb}{1,0,0}
\definecolor{lstreddark}{rgb}{0.7,0,0}
\definecolor{lstredl}{rgb}{0.64,0.08,0.08}
\definecolor{lstmildblue}{rgb}{0.66,0.72,0.78}
\definecolor{lstblue}{rgb}{0,0,1}
\definecolor{lstmildgreen}{rgb}{0.42,0.53,0.39}
\definecolor{lstgreen}{rgb}{0,0.5,0}
\definecolor{lstorangedark}{rgb}{0.6,0.3,0}
\definecolor{lstorange}{rgb}{0.75,0.52,0.005}
\definecolor{lstorangelight}{rgb}{0.89,0.81,0.67}
\definecolor{lstbeige}{rgb}{0.90,0.86,0.45}

\DeclareFontShape{OT1}{cmtt}{bx}{n}{<5><6><7><8><9><10><10.95><12><14.4><17.28><20.74><24.88>cmttb10}{}

\lstdefinestyle{psql}
{
tabsize=2,
basicstyle=\footnotesize\upshape\ttfamily,
language=SQL,
morekeywords={PROVENANCE,BASERELATION,INFLUENCE,COPY,ON,TRANSPROV,TRANSSQL,TRANSXML,CONTRIBUTION,COMPLETE,TRANSITIVE,NONTRANSITIVE,EXPLAIN,SQLTEXT,GRAPH,IS,ANNOT,THIS,XSLT,MAPPROV,cxpath,OF,TRANSACTION,SERIALIZABLE,COMMITTED,INSERT,INTO,WITH,SCN,UPDATED,WINDOW},
extendedchars=false,
keywordstyle=\bfseries,
mathescape=true,
escapechar=@,
sensitive=true
}

\lstdefinestyle{psqlcolor}
{
tabsize=2,
basicstyle=\footnotesize\upshape\ttfamily,
language=SQL,
morekeywords={PROVENANCE,BASERELATION,INFLUENCE,COPY,ON,TRANSPROV,TRANSSQL,TRANSXML,CONTRIBUTION,COMPLETE,TRANSITIVE,NONTRANSITIVE,EXPLAIN,SQLTEXT,GRAPH,IS,ANNOT,THIS,XSLT,MAPPROV,cxpath,OF,TRANSACTION,SERIALIZABLE,COMMITTED,INSERT,INTO,WITH,SCN,UPDATED,FOLLOWING,RANGE,UNBOUNDED,PRECEDING,OVER,PARTITION,WINDOW},
extendedchars=false,
keywordstyle=\bfseries\color{lstpurple},
deletekeywords={count,min,max,avg,sum,lag,first_value,last_value},
keywords=[2]{count,min,max,avg,sum,lag,first_value,last_value,lead,row_number},
keywordstyle=[2]\color{lstblue},
stringstyle=\color{lstreddark},
commentstyle=\color{lstgreen},
mathescape=true,
escapechar=@,
sensitive=true
}

\lstdefinestyle{datalog}
{
basicstyle=\footnotesize\upshape\ttfamily,
language=prolog
}

\lstdefinestyle{pseudocode}
{
  tabsize=3,
  basicstyle=\small,
  language=c,
  morekeywords={if,else,foreach,case,return,in,or},
  extendedchars=true,
  mathescape=true,
  literate={:=}{{$\gets$}}1 {<=}{{$\leq$}}1 {!=}{{$\neq$}}1 {append}{{$\listconcat$}}1 {calP}{{$\cal P$}}{2},
  keywordstyle=\color{lstpurple},
  escapechar=&,
  numbers=left,
  numberstyle=\color{lstgreen}\small\bfseries,
  stepnumber=1,
  numbersep=5pt,
}

\lstdefinestyle{xmlstyle}
{
  tabsize=3,
  basicstyle=\small,
  language=xml,
  extendedchars=true,
  mathescape=true,
  escapechar=£,
  tagstyle=\color{keywordpurple},
  usekeywordsintag=true,
  morekeywords={alias,name,id},
  keywordstyle=\color{lstred}
}
 \lstset{style=psqlcolor}

\begin{abstract}
 \emph{Snapshot semantics} is widely used for evaluating queries over
  temporal data: temporal relations are seen as sequences of snapshot relations,
  and queries are evaluated at each snapshot.   In this work, we demonstrate that current approaches for snapshot semantics
  over interval-timestamped multiset relations are subject to two bugs regarding snapshot aggregation and bag difference. We introduce a novel temporal data model based on $\semK$-relations that overcomes
  these bugs and prove it to correctly encode snapshot semantics. Furthermore, we present an efficient implementation of our model as a database middleware and demonstrate experimentally that our approach is competitive with native implementations and significantly outperforms such implementations on queries that involve aggregation.
\end{abstract}

\urlstyle{tt}

\section{Introduction}
\label{sec:introduction}

Recently, there
is renewed interest in temporal databases fueled by the fact that abundant
storage has made long term archival of historical data feasible. This
has led to the incorporation of temporal features into the SQL:2011
standard~\cite{KulkarniM12}
which
defines an encoding of temporal data associating each tuple with a validity period. We refer to such relations as  \textit{SQL period
  relations}. Note that SQL period relations use multiset semantics. Period relations are supported by
many DBMSs, e.g., PostgreSQL~\cite{PGSQLRT},
Teradata~\cite{teradata1510}, Oracle~\cite{ORA16}, IBM
DB2~\cite{SaraccoNG2012}, and MS SQLServer~\cite{MSSQL16}. However,
none of these systems, with the partial exception of Teradata,
supports \textit{snapshot
  semantics}, an important class of
temporal queries.
Given a temporal database, a \revb{non-temporal} query $\query$ interpreted
under snapshot semantics returns a temporal relation that assigns to each point
in time the result of evaluating $\query$ over the snapshot of the database at
this point in time.  This fundamental property of snapshot semantics is known as
\emph{snapshot-reducibility}~\cite{LorentzosM97,SJ95}.  A specific type of
snapshot semantics is the so-called \textit{sequenced
  semantics}~\cite{DBLP:reference/db/BohlenJ09} which in addition to
snapshot-reducibility enforces another property called \textit{change preservation} that
determines how time points are grouped into intervals in a snapshot query
result.

\begin{figure}  \scriptsize
  \begin{subfigure}[b]{1\linewidth}
    \centering
    \begin{tabular}{|ll|l|}
      \multicolumn{3}{@{}l}{\textbf{works}} \\
      \cline{1-3}
      \thead{name} &\thead{skill} & \thead{period}\\
      \cline{1-3}
      Ann & SP & $[03, 10)$ \\
      Joe & NS & $[08, 16)$ \\
      Sam & SP & $[08, 16)$ \\
      Ann & SP & $[18, 20)$ \\
      \cline{1-3}
    \end{tabular}
    \hfill
    \begin{tabular}{|ll|l|}
      \multicolumn{3}{@{}l}{\textbf{assign}} \\
      \cline{1-3}
      \thead{mach} &\thead{skill} & \thead{period}\\
      \cline{1-3}
      M1 & SP & $[03, 12)$ \\
      M2 & SP & $[06, 14)$ \\
      M3 & NS & $[03, 16)$ \\
      \cline{1-3}
    \end{tabular}
    \caption{Input period relations}
    \label{sfig:ex1-input}
  \end{subfigure}

  \begin{subfigure}{.5\linewidth}
    \centering
    \begin{tabular}{|l|c|}
      \multicolumn{2}{@{}l}{\textbf{$\qAggEx$}}\\
      \cline{1-2}
      \thead{cnt} & \thead{period}\\
      \hhline{*{1}{|--}}
      \bugcell{0} & \bugcell{$[00, 03)$} \\
      1 & $[03, 08)$ \\
      2 & $[08, 10)$ \\
      1 & $[10, 16)$ \\
      \bugcell{0} & \bugcell{$[16, 18)$} \\
      1 & $[18, 20)$ \\
      \bugcell{0} & \bugcell{$[20, 24)$} \\
      \cline{1-2}
    \end{tabular}
    \caption{Snapshot aggregation\newline}
    \label{sfig:ex1-result-aggregation}
  \end{subfigure}
  \hfill
  \begin{subfigure}{.5\linewidth}
    \centering
    \begin{tabular}{|l|l|}
      \multicolumn{2}{@{}l}{\textbf{$\qDiffEx$}} \\
      \cline{1-2}
      \thead{skill} & \thead{period}\\
      \hhline{*{1}{|--}}
      \bugcell{SP} & \bugcell{$[06, 08)$} \\
      \bugcell{SP} & \bugcell{$[10, 12)$} \\
      NS & $[03, 08)$ \\
      \cline{1-2}
    \end{tabular}

    \caption{Snapshot difference}
    \label{sfig:ex1-result-difference}
  \end{subfigure}

  \caption{Snapshot semantics query evaluation -- highlighted tuples
    are erroneously omitted by approaches that exhibit the aggregation
    gap (\abbrAGB) and bag difference (\abbrBDB) bugs.}
  \label{fig:seq-semantics-bugs-example}
\end{figure}

\begin{exam}[Snapshot Aggregation]\label{ex:running-example-snapshot}
  Consider the \SQLrel{} \texttt{works} in Figure~\ref{sfig:ex1-input} that
  records factory workers, their skills, and when they are on duty.
  The validity period of each tuple is stored in the temporal attribute
  \texttt{period}.
  \revc{To simplify examples, we restrict the time domain to the hours of 2018-01-01 represented as integers $00$ to $23$.}   The company   requires that at least one SP worker is in the
  factory at any given time. This can be checked by evaluating the following
   query under snapshot semantics.
  \begin{tt}
    \begin{tabbing}
      $\qAggEx$:
      \=\textbf{SELECT} count(*) \textbf{AS} cnt \textbf{FROM} works\\
      \>\textbf{WHERE} skill = 'SP'
    \end{tabbing}
  \end{tt}
  Evaluated under snapshot semantics, \revb{a query returns a snapshot
  (time-varying) result that records when the result
  is valid, i.e., $\qAggEx$ returns} the number of SP workers that are
  on duty at any given point of time. The result is shown in
  Figure~\ref{sfig:ex1-result-aggregation}. For instance, at 08:00am two
  SP workers (Ann and Joe) are on duty. The query exposes several
  safety violations, e.g., no SP worker is on duty between \revc{00}   and
  03. \end{exam}

In the example above, safety violations correspond to gaps, i.e.,
periods of time where the aggregation's input is empty. As we will
demonstrate, all approaches for snapshot semantics that we are aware
of do not return results for gaps (tuples marked in red)
and, therefore, violate snapshot-reducibility.
\revc{Teradata~\cite[p.149]{teradata1510} for instance, realized the importance of reporting results for gaps, but in contrast to snapshot-reducibility provides gaps in the presence of grouping, while omitting them otherwise.}
 As a consequence, in
our example these approaches fail to identify safety violations.  We refer to this
type of error as the \emph{aggregation gap bug} (\emph{\abbrAGB{}
  bug}).

Similar to the case of  aggregation, we also identify a
common error related to snapshot bag difference
(\texttt{\textbf{EXCEPT ALL}}).

\begin{exam}[Snapshot Bag Difference]\label{ex:bag-difference-running-example}
  Consider again Figure~\ref{fig:seq-semantics-bugs-example}. Relation
  \texttt{assign} records machines (mach) that need to be assigned to
  workers with a specific skill over a specific period of time. For
  instance, the third tuple records that machine M3 requires a
  non-specialized (NS) worker for the time period $[03, 16)$.  To
  determine which skill sets are missing during which time period, we
  evaluate the following query under snapshot semantics:
  \begin{tt}
    \begin{tabbing}
      $\qDiffEx$:
      \=\textbf{SELECT} skill \textbf{FROM} assign\\
      \>\textbf{EXCEPT ALL}\\
      \>\textbf{SELECT} skill \textbf{FROM} works
    \end{tabbing}
  \end{tt}
  The result in Figure~\ref{sfig:ex1-result-difference} indicates that
  one more SP worker is required during the periods
  $[06,08)$ and $[10, 12)$. \end{exam}

Many approaches treat bag difference as a
\texttt{\textbf{NOT EXISTS}} subquery, and therefore do not return a
tuple $t$ from the left input if this tuple exists in the
right input (independent of their multiplicity).  For instance, the
two tuples for the SP workers (highlighted in red) are not
returned, since there exists an SP worker at each snapshot in the
\texttt{works} relation.  This violates snapshot-reducibility.  We
refer to this type of error as the \emph{bag difference bug}
(\emph{\abbrBDB{} bug}).

The interval-based representation of temporal relations creates an
additional problem: the encoding of a temporal query result is
typically not unique.  For instance, tuple
$(\mathit{Ann},\mathit{SP},[03,10))$ from the \texttt{works} relation
in Figure~\ref{fig:seq-semantics-bugs-example} can equivalently be
represented as two tuples $(\mathit{Ann},\mathit{SP},[03,08))$ and
$(\mathit{Ann},\mathit{SP},[08,10))$.  We refer to a method that
determines how temporal data and snapshot query results are grouped
into intervals as an \textit{interval-based representation system}.  A
unique and predictable representation of temporal data is a desirable
property, because equivalent relational algebra
expressions should not lead to syntactically different result relations.\crtodo{highlight this previous statement}
This problem can be addressed by using a representation system that
associates a unique encoding with each temporal database.
Furthermore, overlap between multiple periods associated with a tuple
and unnecessary splits of periods complicate the interpretation of
data and, thus, should be avoided if possible.
\reva{Given these limitations and the lack of implementations for snapshot semantics queries over bag relations, users currently resort to manually implementing such queries in SQL which is time-consuming and error-prone~\cite{Snodgrass99}.}
We address the above limitations of previous approaches for
snapshot semantics and develop a framework based on the following desiderata:
(i) support for set and multiset relations, (ii) snapshot-reducibility for all
operations, and (iii) a unique interval-based encoding of temporal relations.
Note that while previous work on sequenced semantics (e.g.,~\cite{DignosBGJ16,DignosBG12}) also aims to support snapshot-reducibility, we emphasize a unique encoding instead of trying to preserve intervals from the input of a query.
We address these desiderata using a three-level approach. \reva{Note that we focus on data with a single time dimension, but are oblivious to whether this is transaction time or valid time.}
First, we
introduce an \textit{\revc{abstract} model} that supports both sets and
multisets, and by definition is snapshot-reducible. This model,
however, uses a verbose encoding of temporal data and, thus, is not
practical. \revm{Afterwards, we develop a more compact \textit{logical model} as a
representation system, where the complete temporal history of all equivalent
tuples from the abstract model is stored in an annotation
attached to one tuple.}
The \revc{abstract} and the logical models leverage the theory of
K-relations, which are a general class of annotated relations that
cover both set and multiset relations.  For our
\textit{implementation}, we use SQL over period relations to
ensure compatibility with SQL:2011  and existing DBMSs.
We prove the equivalence between the three layers (i.e., the
\revc{abstract} model, the logical model and the implementation) and show
that the logical model determines a unique interval-encoding for the
implementation and a correct rewriting scheme for queries over this
 encoding.

Our main technical contributions are: \begin{itemize}[noitemsep,topsep=0pt,parsep=0pt,partopsep=0pt]
\item \textit{\revc{Abstract} model}: We introduce \textit{\SKrels{}} as a generalization of snapshot
  set and multiset relations. These relations   are by definition
  snapshot-reducible.

\item \textit{Logical model}: We define an interval-based representation, termed
  \emph{\periodKrels}, and prove that
  these relations are a compact and unique representation system for
  snapshot semantics over \SKrels{}.  We show this for the full
  relational algebra plus aggregation ($\raAgg$).

\item
We achieve a unique encoding of temporal data as \periodKrels{} by  generalizing set-based  coalescing~\cite{DBLP:conf/vldb/BohlenSS96}.

  \item We demonstrate that the multiset version of {\periodKrels} can be encoded as \textit{\SQLrels{}}, a
  common interval-based model in DBMSs, and how to translate queries with
  snapshot semantics over \periodKrels{} into SQL.

\item We implement our approach as a database middleware and present
  optimizations that eliminate redundant coalescing steps. We demonstrate
  experimentally   that we do not
  need to sacrifice performance to achieve correctness.
\end{itemize}

\section{Related Work}
\label{sec:related-work}

\parttitle{Temporal Query Languages}
There is a long history of research on temporal query
languages~\cite{DBLP:reference/db/JensenS09r,DBLP:reference/db/BohlenGJS09}.
Many temporal query languages including
TSQL2~\cite{Snodgrass95,DBLP:journals/sigmod/SnodgrassAABCDEGJKKKLLRSSS94},
\mbox{ATSQL2} (Applied TSQL2)~\cite{Bohlen95evaluatingand},
IXSQL~\cite{LorentzosM97}, ATSQL~\cite{BohlenJS00}, and
SQL/TP~\cite{T98} support sequenced semantics, i.e., these languages support a specific type of snapshot semantics.  In this paper, we provide a general
framework that can be used to correctly implement snapshot semantics
over period set and multiset relations for any language.

\parttitle{Interval-based Approaches for Sequenced Semantics}
In the following, we discuss interval-based approaches for sequenced
semantics.
Table~\ref{tab:table-rel-work-bugs} shows for each approach whether it
supports multisets, whether it is free of the aggregation gap and bag
difference bugs, and whether its interval-based encoding of a
sequenced query result is unique. An N/A indicates that the approach
does not support the operation for which this type of bug can occur or
the semantics of this operation is not defined precisely enough to
judge its correctness.  Note that while temporal query languages may
be defined to apply sequenced semantics and, thus, by definition are
snapshot-reducible, (the specification of) their implementation might
fail to be snapshot-reducible.  In the following discussion of the
temporal query languages in Table~\ref{tab:table-rel-work-bugs}, we
refer to their semantics as provided in the referenced publication(s).

\emph{Interval preservation} (ATSQL)~\cite[Def.\ 2.10]{BohlenJS00} is a
representation system for SQL period relations (multisets) that tries to
preserve the intervals associated with input
tuples, i.e., fragments of all intervals (including duplicates)
associated with the input tuples ``survive'' in the output.  Interval
preservation is snapshot-reducible for multiset semantics for positive
relational algebra\revc{\cite{DBLP:reference/db/Sirangelo09c} (selection, projection, join, and union)}, but exhibits the aggregation gap and bag difference
bug.  Moreover, the period encoding of a query result is not unique as
it depends both on the query and the input representation.
\emph{Teradata}~\cite{teradata1510} is a commercial DBMS that supports
sequenced operators using ATSQL's statement modifiers. The
implementation is based on query
rewriting~\cite{DBLP:conf/edbt/Al-KatebGCBCP13} and does not support
difference.  Teradata's implementation exhibits the aggregation gap
bug. Since the application of coalescing is optional, the encoding of
snapshot relations as period relations is not unique.
\revc{\emph{Change preservation}~\cite[Def.\ 3.4]{DignosBGJ16} determines the
interval boundaries of a query result tuple $\tuple$ based on the maximal
interval for which there is no change in the input. To track changes, it
employs the lineage provenance model in~\cite{DignosBG12} and the PI-CS model
in~\cite{DignosBGJ16}. The approach
uses timestamp adjustment in combination with traditional database operators, but does not provide a unique encoding, exhibits the AG bug,
and only supports set semantics. Our work addresses these issues and significantly generalizes this approach, in particular by supporting bag semantics. }\emph{TSQL2}~\cite{Snodgrass95,DBLP:journals/sigmod/SnodgrassAABCDEGJKKKLLRSSS94,SJ95}
implicitly applies coalescing~\cite{DBLP:conf/vldb/BohlenSS96} to
produce a unique representation. Thus, it only supports  set semantics,
and it does not support aggregation.
Snodgrass et al.~\cite{SB96} present a validtime extension of
\emph{SQL/Temporal} and an algebra with sequenced semantics. The
algebra supports multisets, but exhibits both the aggregation gap and
bag difference bug. Since intervals from the input are preserved
where possible, the interval representation of a snapshot
relation is not unique.
\emph{TimeDB}~\cite{S98} is an implementation of
ATSQL2~\cite{Bohlen95evaluatingand}. It uses a semantics for bag
difference and intersection that is not snapshot-reducible
(see~\cite[pp.\ 63]{S98}).
 Our approach is the first that supports set and multiset relations,
 is resilient against the two bugs, and specifies a unique
 interval-encoding.

\begin{table}[tb]
  \caption{Interval-based approaches for snapshot semantics.}
  \label{tab:table-rel-work-bugs}

  \centering
  \small

  \newcommand{\bugtableheaderheight}{1.2cm}
  \renewcommand{\arraystretch}{1.1}

  \begin{tabular}{|l|c|c|c|c|}
    \hline
    \thead{Approach}
    & \thead{
      \rotatebox{90}{
      \begin{minipage}{\bugtableheaderheight}
        Multisets
      \end{minipage}
      }
      }
    & \thead{
      \rotatebox{90}{
      \begin{minipage}{\bugtableheaderheight}
        \abbrAGB{} bug\\ free
      \end{minipage}
      }
      }
    & \thead{
      \rotatebox{90}{
      \begin{minipage}{\bugtableheaderheight}
        \abbrBDB{} bug\\ free
      \end{minipage}
      }
      }
    & \thead{
      \rotatebox{90}{
      \begin{minipage}{\bugtableheaderheight}
        Unique \\ encoding
      \end{minipage}
    }
    }
    \\
    \hline
    Interval preservation~\cite{BohlenJS00} (ATSQL) & \checkmark & $\times$ & $\times$ & $\times$
    \\
    Teradata~\cite{teradata1510} & \checkmark & $\times$ & N/A  & $\times$\tablefootnote{Optionally, coalescing (\texttt{NORMALIZE ON} in Teradata) can be applied to get a unique encoding at the cost of loosing multiplicities.}
    \\
    Change preservation~\cite{DignosBG12,DignosBGJ16} & $\times$ & $\times$& N/A & $\times$
    \\
    TSQL2~\cite{Snodgrass95,DBLP:journals/sigmod/SnodgrassAABCDEGJKKKLLRSSS94,SJ95} & $\times$ & N/A & N/A & $\checkmark$
    \\
    ATSQL2~\cite{Bohlen95evaluatingand} & \checkmark & N/A & $\times$ & $\times$
    \\
    TimeDB~\cite{S98} (ATSQL2) & \checkmark & N/A & $\times$& $\times$
    \\
    SQL/Temporal~\cite{SB96} & \checkmark & $\times$ & $\times$ & $\times$
    \\
    \revc{SQL/TP~\cite{T98}\tablefootnote{Sequenced semantics can be expressed, but this is inefficient}} & \revc{\checkmark{}} & \revc{\checkmark{}} & \revc{\checkmark{}} & \revc{$\times$}
    \\
    {\bf Our approach} & \checkmark & \bf \checkmark & \checkmark & \checkmark
    \\
    \hline
  \end{tabular}
\end{table}

\begin{figure*}[htb]
  \centering
  \begin{tikzpicture}[scale=0.95]

\newcommand{\ovfigLeftColumnWidth}{0.4\linewidth}
\newcommand{\ovfigMiddleColumnWidth}{0.09\linewidth}
\newcommand{\ovfigRightColumnWidth}{0.43\linewidth}

  \tikzstyle{every node}=[font=\scriptsize]
  {\renewcommand{\arraystretch}{1.1}

  \draw[fill=red!15,line width=0pt,color=red!15] (-0.7,3) rectangle (17.9,-0);
  \draw[fill=orange!15,line width=0pt,color=orange!15] (-0.7,-0.5) rectangle (17.9,-2.5);
  \draw[fill=green!15,line width=0pt,color=green!15] (-0.7,-3) rectangle (17.9,-6);

  \begin{scope}[shift={(0,0)}]

    \node[above, rotate=90] at (0,1.5) {{\bf \small Implementation}};
    \node[color=white, fill=black, above left] at (10.3,2.5) {{\normalsize \bf \SQLrels{}}};

    \draw [rounded corners=0.5cm, line width=1pt] (0,0) rectangle ++(6,3);
    \node at (3,1.5) {
      \begin{tabular}{|lll|}
			\cline{1-3}
			\thead{name} &\thead{skill} & \thead{period}\\
			\cline{1-3}
			Ann & SP & $[03, 10)$ \\
      Joe & NS & $[08, 16)$ \\
			Sam & SP & $[08, 16)$ \\
      Ann & SP & $[18, 20)$ \\
			\cline{1-3}
		  \end{tabular}
    };
    \draw[rounded corners=0.5cm, line width=1pt] (11,0) rectangle ++(6.5,3);
    \node at (14.25,1.5) {
  		\begin{tabular}{|l|c|}
  			\cline{1-2}
  			\thead{cnt} & \thead{period}\\
  			\cline{1-2}
  			0 & $[00, 03)$ \\
  			1 & $[03, 08)$ \\
  			2 & $[08, 10)$ \\
  			1 & $[10, 16)$ \\
  			0 & $[16, 18)$ \\
        1 & $[18, 20)$ \\
        0 & $[20, 24)$ \\
  			\cline{1-2}
  		\end{tabular}
    };
    \draw[->, very thick] (6.2,1.5) -- ++(4.6,0) node[pos=0.5, above] {$\reprRewr(Q_{onduty})$};
  \end{scope}

  \draw[->, line width=1pt] (3, 0) -- (3, -0.4) node[pos=0.6, left] {$\reprN^{-1}$};
  \draw[->, line width=1pt] (4, -0.5) -- (4, -0.1) node[pos=0.6, right] {$\reprN$};

  \draw[->, line width=1pt] (13.25, 0) -- (13.25, -0.4) node[pos=0.6, left] {$\reprN^{-1}$};
  \draw[->, line width=1pt] (14.25, -0.5) -- (14.25, -0.1) node[pos=0.6, right] {$\reprN$};

  \begin{scope}[shift={(0,-3)}]

    \node[above, rotate=90] at (0,1.5) {{\bf \small Logical}};
    \node[color=white, fill=black, above left] at (9.8,2) {{\normalsize \bf Period K-relations}};

    \draw [rounded corners=0.5cm, line width=1pt] (0,0.5) rectangle ++(6.75,2);
    \node at (3.5,1.5) {
  		\begin{tabular}{|ll|c}
  			\cline{1-2}
  			\thead{name} &\thead{skill} & \annotheadcell{$\semTimeNIN$}\\
  			\cline{1-2}
  			Ann & SP & \annotcell{$\{[03, 10) \mapsto 1, [18,20) \mapsto 1\}$} \\
  			Sam & SP & \annotcell{$\{[08, 16) \mapsto 1 \}$} \\
  			Joe & NS & \annotcell{$\{[08, 16) \mapsto 1\}$} \\
  			\cline{1-2}
  		\end{tabular}
    };
    \draw[rounded corners=0.5cm, line width=1pt] (10,0.5) rectangle ++(7.5,2);
    \node at (13.75,1.5) {
  		\begin{tabular}{|l|l}
  			\cline{1-1}
  			\thead{cnt} & \multicolumn{1}{c}{\annotheadcell{$\semTimeNIN$}}\\[1mm]
  			\cline{1-1}
            0 & \annotcell{$\{ [00, 03) \mapsto 1, [16, 18) \mapsto 1,[20, 24) \mapsto 1 \}$} \\
            1 & \annotcell{$\{ [03, 08) \mapsto 1, [10, 16) \mapsto 1 ,[18, 20) \mapsto 1\}$} \\
  			    2 & \annotcell{$\{ [08, 10) \mapsto 1\}$} \\

            \cline{1-1}
  			\cline{1-1}
  		\end{tabular}
    };
    \draw[->, very thick] (7,1.5) -- ++(2.7,0) node[pos=0.5, above] {$Q_{onduty}$} ;
  \end{scope}

  \draw[->, line width=1pt] (3, -2.5) -- (3, -2.9) node[pos=0.6, left] {$\tSlice{00}, \ldots, \tSlice{23}$};
  \draw[->, line width=1pt] (4, -3) -- (4, -2.6) node[pos=0.6, right] {$\repr_{\semN}$};

  \draw[->, line width=1pt] (13.25, -2.5) -- (13.25, -2.9) node[pos=0.6, left] {$\tSlice{00}, \ldots, \tSlice{23}$};
  \draw[->, line width=1pt] (14.25, -3) -- (14.25, -2.6) node[pos=0.6, right] {$\repr_{\semN}$};

    \begin{scope}[shift={(0,-6)}]

      \node[above, rotate=90] at (0,1.5) {{\bf \small \revc{Abstract}}};
      \node[color=white,fill=black, above left] at (10,2.5) {{\normalsize \bf Snapshot K-relations}};

      \draw [ rounded corners=0.5cm, line width=1pt] (0,0) rectangle ++(6,3);
      \node at (2.2,2.7) {
        \begin{tabular}{l|ll|c}
          \cline{2-3}
           $00 \mapsto {}$ & \thead{name} &\thead{skill} &  \annotheadcell{$\semN$}\\
          \cline{2-3}
        \end{tabular}
      };
      \node at (2.2,1.8) {
        \begin{tabular}[t]{l|ll|c}
          \multicolumn{2}{l}{{\bf \ldots}} & \multicolumn{1}{l}{} &\\
          \cline{2-3}
           $08 \mapsto {}$ & \thead{name} &\thead{skill} &  \annotheadcell{$\semN$}\\
           \cline{2-3}
           & Ann & SP & \annotcell{1} \\
           & Joe & NS & \annotcell{1} \\
           & Sam & SP & \annotcell{1} \\
          \cline{2-3}
        \end{tabular}
      };
      \node at (2.2,0.6) {
        \begin{tabular}[t]{l|ll|c}
          \multicolumn{2}{l}{{\bf \ldots}} & \multicolumn{1}{l}{} &\\
          \cline{2-3}
           $18 \mapsto {}$ & \thead{name} &\thead{skill} &  \annotheadcell{$\semN$}\\
           \cline{2-3}
           \multicolumn{1}{l|}{{\bf \ldots}} & Ann & SP & \annotcell{1} \\
          \cline{2-3}
        \end{tabular}
      };

      \draw[rounded corners=0.5cm, line width=1pt] (10.5,0) rectangle ++(7,3);
      \node at (13.75,2.55) {
        \begin{tabular}{l|l|c}
          \cline{2-2}
           $00 \mapsto {}$ & \thead{cnt} &  \annotheadcell{$\semN$}\\
           \cline{2-2}
            & 0 & \annotcell{1} \\
          \cline{2-2}
        \end{tabular}
      };
      \node at (13.75,1.8) {
        \begin{tabular}{l|l|c}
          \multicolumn{2}{l}{{\bf \ldots}} & \multicolumn{1}{l}{}\\
          \cline{2-2}
           $08 \mapsto {}$ & \thead{cnt} &  \annotheadcell{$\semN$}\\
           \cline{2-2}
            & 2 & \annotcell{1} \\
          \cline{2-2}
        \end{tabular}
      };
      \node at (13.75,0.6) {
        \begin{tabular}{l|l|c}
          \multicolumn{2}{l}{{\bf \ldots}} & \multicolumn{1}{l}{}\\
          \cline{2-2}
           $18 \mapsto {}$ & \thead{cnt} &  \annotheadcell{$\semN$}\\
           \cline{2-2}
            {\bf \ldots}& 1 & \annotcell{1} \\
          \cline{2-2}
        \end{tabular}
      };
      \node at (8.5, 2.15) {{\bf \ldots}};
      \draw[->, very thick] (6.2,1.5) -- ++(4,0) node[pos=0.5, above] {$Q_{onduty}
      $} ;
      \node at (8.5, 1.2) {{\bf \ldots}};
      \draw[->, very thick] (6.2,0.6) -- ++(4,0) node[pos=0.5, above] {$Q_{onduty}
      $} ;
      \node at (8.5, 0.3) {{\bf \ldots}};
    \end{scope}

  }

  \end{tikzpicture}
  \caption{Overview of our approach. Our \emph{\revc{abstract} model} is
  \emph{snapshot K-relations} and nontemporal queries over snapshots (snapshot semantics). Our
  \emph{logical} model is \emph{period K-relations} and   queries
  corresponding to the \reva{abstract} model's snapshot queries.   Our
  \emph{implementation} uses \emph{\SQLrels{}} and
  rewritten non-temporal queries implementing the other model's snapshot queries.   Each model is associated with
  transformations to the other models   which commute with queries (modulo the rewriting $\reprRewr$ when mapping to the implementation). }   \label{fig:overview-approach}
\end{figure*}
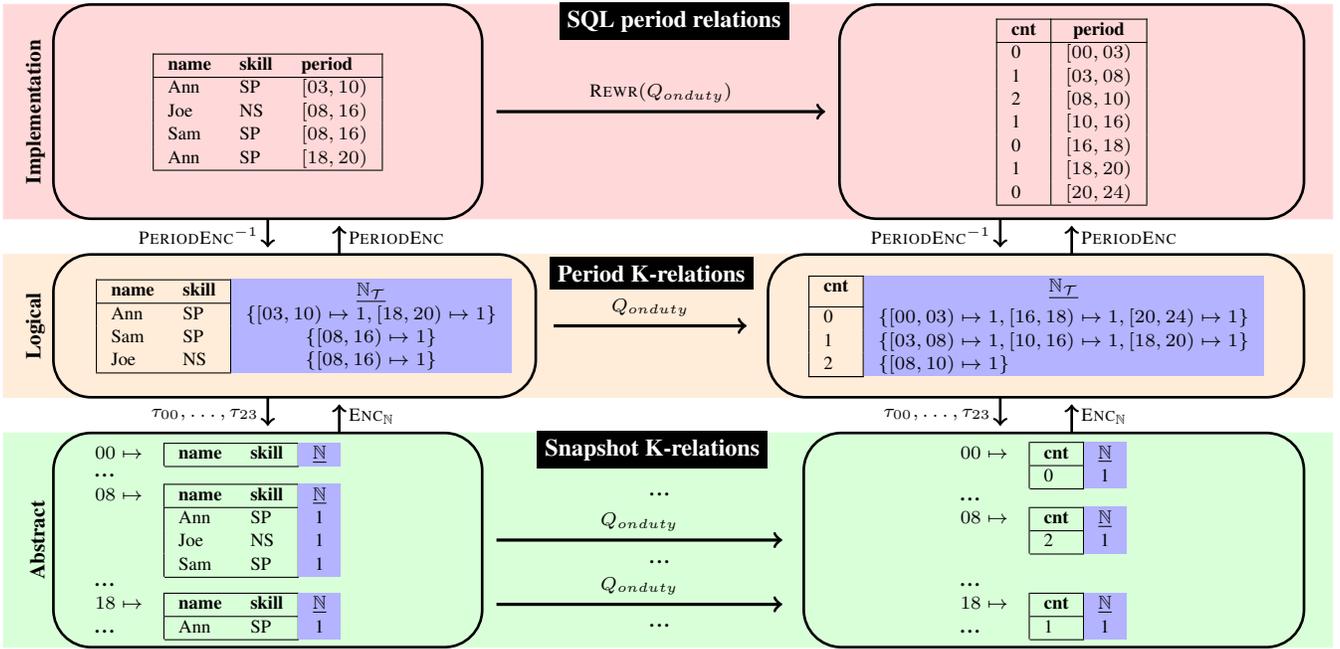

\parttitle{Non-sequenced Temporal Queries}
Non-sequenced temporal query languages, such as
IXSQL~\cite{LorentzosM97} and SQL/TP~\cite{T98}, do not explicitly
support sequenced semantics.
Nevertheless, we review these languages here since they allow to
express queries with sequenced semantics.
\revc{SQL/TP~\cite{T98} introduces a point-wise semantics for temporal
queries~\cite{BowmanT03,T96}, where time is handled as a regular attribute.
Intervals are used as an efficient encoding of time points, and a normalization
operation is used to split intervals. The language supports multisets and a
mechanism to manually produce sequenced semantics. However, sequenced semantics queries are specified as the union of non-temporal queries over snapshots. Even if such subqueries are grouped together for adjacent time points where the non-temporal query's result is constant this still results in a large number of subqueries to be executed. Even worse, the number of subqueries that is required is data dependent.
Also, the interval-based encoding is not unique, since time
points are grouped into intervals depending on query syntax and encoding of the
input. While this has no effect on the semantics since SQL/TP queries cannot
distinguish between different interval-based encodings of a temporal database,
it might be confusing to users that observe different query results for
equivalent queries/inputs.}

\parttitle{Implementations of Temporal Operators}
A large body of work has focused on the implementation of individual
temporal algebra operators such as
joins~\cite{DBLP:conf/sigmod/DignosBG14,DBLP:conf/icde/PiatovHD16,DBLP:journals/pvldb/BourosM17}
and
aggregation~\cite{DBLP:conf/edbt/BohlenGJ06,DBLP:conf/sigmod/PilmanKKKP16,DBLP:conf/ssd/PiatovH17}. Some
exceptions supporting multiple operators
are~\cite{DBLP:conf/sigmod/KaufmannMVFKFM13, DignosBGJ16,
  CafagnaB17}. These approaches introduce efficient evaluation
algorithms for a particular semantics of a temporal algebra operator.
Our approach can utilize efficient operator implementations as long as
(i) their semantics is compatible with our interval-based encoding of
snapshot query results and (ii) they are snapshot-reducible.

\parttitle{Coalescing}
Coalescing produces a unique representation of a \emph{set} semantics
temporal database.  B\"ohlen et al.~\cite{DBLP:conf/vldb/BohlenSS96}
study optimizations for coalescing that eliminate unnecessary coalescing operations. Zhou et
al.~\cite{ZhouWZ06} and~\cite{DBLP:conf/dexa/Al-KatebGC12} use analytical functions to efficiently
implement coalescing in SQL.
We generalize coalescing to $\semK$-relations to define a unique encoding of
interval-based temporal relations, including \emph{multiset} relations. Similar
to~\cite{DBLP:conf/vldb/BohlenSS96}, we remove unnecessary K-coalescing steps
and, similar to~\cite{ZhouWZ06}, we use OLAP functions for efficient implementation.

\BGDel{Our approach is guided by providing a clean framework that is
  able to generally define all operators under different data model
  semantics, such as set and multiset semantics. This framework is
  crucial to categorize existing implementations, but also to generate
  new ones for cases where no existing is available, such as for
  example multiset difference.}

\parttitle{Temporality in Annotated Databases}
Kostiley et al.~\cite{KB12} is to the best of our knowledge the only
previous approach that uses semiring annotations to express
temporality. The authors define a semiring whose elements are sets of
time points. This approach is limited to set semantics, and no
interval-based encoding was presented.  The LIVE
system~\cite{DT10} combines provenance and uncertainty annotations
with versioning. The system uses interval timestamps, and query
semantics is based on snapshot-reducibility~\cite[Def.\
2]{DT10}. However, computing the intervals associated with a query
result requires provenance to be maintained for every query result.

\section{Solution Overview}
\label{sec:probl-stat-solut}

In this section, we give an overview of our three-level framework,
which is illustrated in Figure~\ref{fig:overview-approach}.

\parttitle{\revc{Abstract} model -- Snapshot $\semK$-relations}
\label{sec:conc-level-seuq}
As an \emph{\revc{abstract} model} we use snapshot relations which map time
points to snapshots. Queries over such relations are evaluated over
each snapshot, which trivially satisfies snapshot-reducibility.  To
support both sets and multisets, we introduce
\emph{snapshot $\semK$-relations}~\cite{GK07}, which are
snapshot relations where each snapshot is a $\semK$-relation.  In a
$\semK$-relation, each tuple is annotated with an element from a
domain $\semK$. For example,
relations annotated with elements from the semiring $\semN$ (natural
numbers) correspond to multiset semantics. The result of a snapshot query $\query$ over a snapshot
$\semK$-relation is the result of evaluating $\query$ over the
$\semK$-relation at each time point.

\begin{exam}[\revc{Abstract} Model]
  Figure~\ref{fig:overview-approach} (bottom) shows the snapshots at
  times \texttt{00}, \texttt{08}, and \texttt{18} of an encoding of
  the running example as snapshot $\semN$-relations.  Each snapshot
  is an $\semN$-relation where tuples are annotated with their
  multiplicity (shown   with shaded
  background). For instance, the snapshot at time \texttt{08} has
  three tuples, each with multiplicity 1. The result of query
  $\qAggEx$ is shown on the bottom right.  Every snapshot in the
  result is computed by running $\qAggEx$ over the corresponding
  snapshot in the input.  For instance, at time \texttt{08} there are
  two SP workers, i.e., $cnt = 2$.
\end{exam}

\parttitle{Logical Model -- Period $\semK$-relations}
\label{sec:logical-level}
\revm{We introduce \periodKrels{} as a \emph{logical model}, which merges
equivalent tuples over all snapshots from the abstract model into one tuple. In
a \emph{\periodKrel{}}, every tuple is annotated with a \textit{temporal
$\semK$-element} that is a unique interval-based representation for all time points of the merged tuples from the abstract model.}
We define a class of semirings called \textit{period semirings} whose elements are temporal $\semK$-elements.
Specifically, for any semiring $\semK$ we can construct a period semiring $\semTimeNI$ whose annotations are temporal $\semK$-elements.
For instance, $\semTimeNIN$ is the period semiring corresponding to semiring $\semN$ (multisets).
We define necessary conditions for an interval-based model to
correctly encode \SKrels{}  and prove that
\periodKrels{} fullfil these conditions.  Specifically, we call an
interval-based model a \emph{representation system} iff
 the encoding of every \SKrel{} $R$ is (i) unique and (ii) snapshot-equivalent to $R$.
Furthermore, (iii) queries over encodings are
snapshot-reducible.

\begin{exam}[Logical Model]
  Figure~\ref{fig:overview-approach} (middle) shows an encoding of the
  running example as \periodKrels{}.
  \revm{For instance, all tuples $(Ann,SP)$ from the abstract model are merged
  into one tuple in the logical model with annotation $\{[03, 10) \mapsto 1,
  [18,20) \mapsto 1\}$, because at each time point during $[03, 10)$ and
  $[18,20)$ a tuple $(Ann,SP)$ with multiplicity $1$ exists.}
  In Section~\ref{sec:snapshot-k-relations-1}, we will introduce a mapping
  $\repr_{\semN}$ from snapshot $\semN$ to $\semTimeNIN$-relations and the  \textit{time slice} operator $\tSlice{\tPoint}$ which restores an
the snapshot at time $\tPoint$.

\end{exam}

\parttitle{Implementation -- SQL Period Relations}
\label{sec:impl-level}
To ensure compatibility with the SQL standard, we use \emph{\SQLrels{}} in our \emph{implementation} and translate snapshot
semantics queries into SQL queries over these period relations.
For this we define an encoding of $\semTimeNIN$-relations as \SQLrels{} ($\reprN$) together with a rewriting scheme for
queries ($\reprRewr$).

\begin{exam}[Implementation]
  Consider the \SQLrels{} shown on the top of
  Figure~\ref{fig:overview-approach}.  Each interval-annotation pair
  of a temporal $\semN$-element in the logical model is encoded as a
  separate tuple in the implementation.
  \revm{For instance, the annotation of tuple $(Ann,SP)$ from the logical model
  is encoded as two tuples, each of which records one of the two intervals from
  this annotation}
\end{exam}

We present an implementation of our framework as a database middleware
that exposes snapshot semantics as a new language feature in
SQL and rewrites snapshot queries into SQL queries over
\SQLrels{}.  That is, we directly evaluate
snapshot queries over data stored natively as period relations.

\section{Snapshot K-relations} \label{sec:snapshot-k-relations}

We first review background on the semiring annotation framework ($\semK$-relations). Afterwards,  we define  \textit{snapshot $\semK$-relations} as our \revc{abstract}
model and snapshot semantics for this model. Importantly, queries over snapshot $\semK$-relations are snapshot-reducible by construction. Finally, we state requirements for a logical model to be a representation system
for this \revc{abstract} model.

\subsection{K-relations}
\label{sec:k-relations}

In a $\semK$-relation~\cite{GK07}, every tuple is annotated with an
element from a domain $K$ of a commutative semiring $\semK$. A structure $(\semK, \addK, \multK, \zeroK, \oneK)$ over a set $K$
with binary operations $\addK$ and $\multK$
is a commutative semiring iff (i) addition and
multiplication are commutative, associative, and have a neutral element ($\zeroK$ and $\oneK$, respectively);
(ii) multiplication distributes over addition; and (iii) multiplication with zero
returns zero. Abusing notation, we
will   use $\semK$ to denote both a semiring structure as well as its
  domain.

Consider a universal countable domain $\uDom$
of values. An n-ary $\semK$-relation $R$ over $\uDom$ is a
(total) function that maps tuples
 (elements from $\uDom^n$)
to elements from $\semK$ with the convention that tuples mapped to
$0_\semK$ are not in the
relation. Furthermore, we require that $R(t) \neq 0_{\semK}$ only
holds for finitely many $t$.  Two semirings are of particular interest
to
us: The semiring $(\mathbb{B}, \vee, \wedge, false, true)$ with elements \textit{true} and \textit{false} using
$\vee$ as addition and $\wedge$ as multiplication corresponds to set semantics.
The semiring $(\mathbb{N}, +, \cdot, 0, 1)$ of natural numbers with standard arithmetics
corresponds to multisets.

The operators of the positive relational algebra\revc{~\cite{DBLP:reference/db/Sirangelo09c}} ($\raPlus$) over
$\semK$-relations are defined by applying the $+_\semK$ and
$\cdot_\semK$ operations of the semiring $\semK$ to input
annotations. Intuitively, the $+_\semK$ and $\cdot_\semK$ operations
of the semiring correspond to the alternative and
conjunctive use of tuples, respectively. For instance, if an output tuple $t$ is
produced by joining two input tuples annotated with $k$ and $k'$, then
the tuple $t$ is annotated with $k \cdot_\semK k'$. Below we provide
the standard definition of $\raPlus$ over
$\semK$-relations~\cite{GK07}. For a tuple $t$, we use $t.A$ to denote the projection
of $t$ on a list of projection expressions $A$ and $t[R]$ to
denote the projection of $t$ on the attributes of relation
$R$. For a condition $\theta$ and tuple $t$, $\theta(t)$ denotes a
function that returns $1_{\semK}$ if $t \models \theta$ and
$0_{\semK}$ otherwise.

\begin{defi}[$\raPlus$ over $\semK$-relations]
  Let $\semK$ be a semiring, $R$, $S$ denote $\semK$-relations,
 $t$, $u$ denote
  tuples of appropriate arity, and $k \in K$.    $\raPlus$ on
  $\semK$-relations is defined as:
\begin{align*}
  \selection_\theta(R)(t) &= R(t) \cdot \theta(t) && (\text{selection}) \\
  \projection_A(R)(t) &= \sum\nolimits_{u: u.A = t} R(u) && (\text{projection}) \\
  (R \join S)(t) &= R(t[R]) \cdot S(t[S]) && (\text{join}) \\
  (R \union S)(t) &= R(t) + S(t) && (\text{union})
\end{align*}
\end{defi}

We will make use of homomorphisms, functions from
the domain of a semiring $\semK_1$ to the domain of a semiring
$\semK_2$ that commute with the semiring operations.  Since $\raPlus$ over $\semK$-relations is defined in terms
of these operations, it follows that semiring homomorphisms commute
with queries, as was proven in~\cite{GK07}.

\begin{defi}[Homomorphism]
A mapping $\homo: \semK_1 \to \semK_2$
is called a homomorphism iff for all $k,k'\in \semK_1$:
\begin{align*}
\homo(0_{\semK_1}) &= 0_{\semK_2}
&\homo(1_{\semK_1}) &= 1_{\semK_2}\\
\homo(k +_{\semK_1} k') &= \homo(k) +_{\semK_2} \homo(k')
&\homo(k \cdot_{\semK_1} k') &= \homo(k) \cdot_{\semK_2} \homo(k')
\end{align*}
\end{defi}

\begin{exam}
  Consider the $\semN$-relations shown below which are non-temporal versions of our running example. Query $Q = \projection_{mach}($
  $works \join assign)$ returns machines for which there are workers with the right skill to operate the machine.
  Under multiset semantics we expect M1 to occur in the
  result of $Q$ with multiplicity $8$ since $(M1,SP)$ joins with
  $(Pete, SP)$ and  with
  $(Bob, SP)$. Evaluating the query in $\semN$ yields the expected
  result by multiplying the annotations of these join partners. Given
  the $\semN$ result of the query, we can compute the result of the
  query under set semantics by applying a homomorphism $h$ which maps
  all non-zero annotations to $\textit{true}$ and $0$ to
  $\textit{false}$. For example, for result $(M1)$ we get
  $h(8) = true$, i.e., this tuple is in the result under set
  semantics.

  \medskip
\noindent
  \begin{minipage}{0.3\linewidth}
\centering
{\scriptsize\upshape
    \begin{tabular}{|cc|r}
      \multicolumn{3}{l}{\textbf{$\bf works$}} \\
      \cline{1-2}
       \thead{name} & \thead{skill} & \annotcell{$\semN$}\\
       \cline{1-2}
       Pete & SP & \annotcell{$1$}\\
       Bob & SP & \annotcell{$1$} \\
       Alice & NS & \annotcell{$1$}\\
      \cline{1-2}
    \end{tabular}
}
\end{minipage}
\begin{minipage}{0.3\linewidth}
\centering
{\scriptsize\upshape
    \begin{tabular}{|cc|r}
      \multicolumn{3}{l}{\textbf{$\bf assign$}} \\
      \cline{1-2}
       \thead{mach} & \thead{skill} & \annotcell{$\semN$}\\
       \cline{1-2}
       M1 & SP & \annotcell{$4$}\\
       M2 & NS & \annotcell{$5$} \\
      \cline{1-2}
    \end{tabular}
}
\end{minipage}
\begin{minipage}{0.35\linewidth}
    \centering
    {\scriptsize\upshape
    \begin{tabular}{|c|l}
      \multicolumn{2}{l}{\textbf{Result}} \\
      \cline{1-1}
       \thead{A} & \annotcell{$\semN$}\\
       \cline{1-1}
  M1 &       \annotcell{$1 \cdot 4 + 1 \cdot 4 = 8$}\\
  M2 &       \annotcell{$5 \cdot 1 = 5$} \\
      \cline{1-1}
    \end{tabular}
}
  \end{minipage}
\end{exam}

\subsection{Snapshot K-relations} \label{sec:snapshot-k-relations-1}

We now formally define snapshot
$\semK$-relations, snapshot semantics over such relations, and then define representation systems.  We assume a totally ordered and finite domain
$\timeDomain$ of time points and use $\tLeq$ to denote its order.
$\tMin$ and $\tMax$ denote the minimal and maximal (exclusive) time
point in $\timeDomain$ according to $\tLeq$, respectively. We use $\tPoint + 1$ to denote the successor of $\tPoint \in \timeDomain$ according to $\tLeq$.

A \SKrel{} over a relation schema $\relSchema$ is a function $\timeDomain \to \kRelDom{\semK}{\relSchema}$, where $\kRelDom{\semK}{\relSchema}$ is the set of all $\semK$-relations with schema $\relSchema$. \CapSKdbs{} are defined analog. We use $\kDBDom{\timeDomain,\semK}{}$ to denote the set of all
\SKdbs{} for time domain $\timeDomain$.

\begin{defi}[Snapshot $\semK$-relation]\label{def:snapshot-k-rel}
  Let $\semK$ be a commutative semiring and $\relSchema$ a relation
  schema. A \SKrel{} $\rel$ is a function
  $\rel : \timeDomain \to \kRelDom{\semK}{\relSchema}$.
\end{defi}

For instance, a snapshot $\semN$-relation is shown in
Figure~\ref{fig:overview-approach} (bottom). Given a \SKrel{}, we use the
\textit{timeslice operator}~\cite{DBLP:reference/db/JensenS09w} to access its
state (snapshot) at a time point $\tPoint$: $$\tSlice{\tPoint}(\rel) = \rel(\tPoint)$$

The evaluation of a query $\query$ over a snapshot database (set of snapshot relations)
  $\db$ under \emph{snapshot semantics} returns a snapshot
relation $\query(\db)$ that is constructed as follows: for each time
point $\tPoint \in \timeDomain$ we have
$\query(\db)(\tPoint) = \query(\db(\tPoint))$.
Thus, snapshot temporal queries over \SKrels{}
behave like queries over $\semK$-relations for each
snapshot, i.e., their semantics is uniquely determined by the semantics of queries over $\semK$-relations.

\begin{defi}[Snapshot Semantics]   \label{def:snapshot-k-rel-queries}
  Let $\db$ be a snapshot $\semK$-database and $\query$ be a query.   The result $\query(\db)$ of $\query$ over $\db$ is a snapshot $\semK$-relation that is defined point-wise as follows:
  \begin{align*}
    \forall \tPoint \in \timeDomain: \query(\db)(\tPoint) = \query(\tSlice{\tPoint}(\db))
  \end{align*}
\end{defi}

For example, consider the snapshot $\semN$-relation shown at the bottom of Figure~\ref{fig:overview-approach} and the evaluation of $Q_{onduty}$ under snapshot semantics as also shown in this figure. Observe how the query result is computed by evaluating $Q_{onduty}$ over each snapshot individually using multiset ($\semN$) query semantics.
Furthermore, since $\tSlice{\tPoint}(\query(\rel)) = \query(\rel)(\tPoint)$, per
the above definition, the timeslice operator commutes with queries:
$\tSlice{\tPoint}(\query(\rel)) = \query(\tSlice{\tPoint}(\rel))$.
This property is \textit{snapshot-reducibility}. \BGDel{and has been
  established as an important correctness criterion for temporal
  queries, in particular, for snapshot semantics.}{said before}
\BGDel{Also note that so far we do not restrict the class of queries
  or the formalism they are expressed in. The above definition is
  well-defined for any query language over $\semK$-relations, e.g.,
  relational algebra.}{not so important}

\BGDel{Furthermore, from Definition~\ref{def:snapshot-k-rel-queries} follows that equivalence of queries for \SKrels{}  is equivalent to equivalence of queries  for non-temporal $\semK$-relations.

\begin{lem}\label{lem:repr-system-eq}
  Let $\query$ and $\query'$ be two queries. Furthermore, let $\equiv_{\semK}$ denote query equivalence for $\semK$-relations and  $\equiv_{\timeDomain,\semK}$ for \SKrels{}. Then
  \begin{align*}
    \query \equiv_{\semK} \query' \Leftrightarrow \query \equiv_{\timeDomain,\semK} \query'
  \end{align*}
\end{lem}
}{side track}

\subsection{Representation Systems}
\label{sec:repr-syst}

To compactly encode \SKrels{},
 we study
representation systems that consist of a set of representations
$\reprDomain$, a function
$\repr: \reprDomain \to \kDBDom{\timeDomain,\semK}{}$ which associates an encoding in $\reprDomain$ with the
snapshot $\semK$-database it represents, and a
timeslice operator $\tSlice{\tPoint}$ which extracts the snapshot at
time $\tPoint$ from an  encoding.
If $\repr$ is injective, then we use $\repr^{-1}(D)$ to denote the unique encoding associated with $D$.
We use $\tSlice{}$ to
denote the timeslice over both snapshot databases and
representations. It will be clear from the input which
operator $\tSlice{}$ refers to.
For such a representation system, we consider two encodings $\db_1$ and
$\db_2$ from $\reprDomain$ to be
\textit{snapshot-equivalent}~\cite{DBLP:reference/db/JensenS09j}
(written as $\db_1 \intervalEq \db_2$)
if they encode the same snapshot $\semK$-database. Note that this is the case if they encode the same snapshots, i.e.,
iff for all
$\tPoint \in \timeDomain$ we have
$\tSlice{\tPoint}(\db_1) = \tSlice{\tPoint}(\db_2)$. For a representation
system to behave correctly, the following conditions have to be met:
1) \textbf{uniqueness}:
for each snapshot $\semK$-database $\db$ there exists a unique element
from $\reprDomain$ representing $\db$;
2)
\textbf{snapshot-reducibility}: the timeslice operator commutes with
queries; and
3) \textbf{snapshot-preservation}: the encoding function $\repr$
preserves the snapshots of the input.

\BGDel{Snapshot $\semK$-relations are easy to understand and provide a clean temporal extension and snapshot temporal query semantics for annotated relations (and, thus, multiset semantics). However, such relations are quite redundant representations since tuples that occur at multiple time points are replicated across multiple instances. Thus, our goal is to study more compact representations of \SKrels{} which have a  query semantics
that is compatible with the snapshot query semantics over \SKrels{}.}{redundant}

\BGDel{For such a representation system to behave ``correctly'', we have to require that the timeslice operator yields the same result as the timeslice over snapshot relations and that query semantics are snapshot-reducible, i.e., they correctly implement snapshot semantics.}{Explained above now}

\begin{defi}[Representation System]\label{def:repr-system}
We call a triple $(\reprDomain, \repr, \tSlice{})$ a \emph{representation system} for snapshot $\semK$-databases with regard to a class of queries $\qClass$ iff for every snapshot database $\db$, encodings $E$, $E' \in \reprDomain$, time point $\tPoint$, and query $\query \in \qClass$ we have
\begin{enumerate}
\item $\repr(E) = \repr(E') \Rightarrow E=E'$\hfill(uniqness)
\item $  \tSlice{\tPoint}(\query(E)) = \query(\tSlice{\tPoint}(E))$
\hfill (snapshot-reducibility)
\item $\repr(E) = D \Rightarrow \tSlice{\tPoint}(E) = \tSlice{\tPoint}(\db)$ \hfill (snapshot-preservation)
\end{enumerate}
\end{defi}

\BGDel{To see why snapshot-reducibility implies that the implementation of snapshot queries over such encodings is correct, consider the following lemma which proves that to check whether queries under the encoding are equivalent to snapshot semantics queries over \SKrels{} it suffices to prove snapshot-reducibility.}{too much info}

\BGDel{
\begin{lem}\label{lem:q-commute-eq-snapshot-red}
Consider the 3 conditions of Definitions~\ref{def:repr-system} and condition (4) shown below. We have,
  \begin{align*}
  (1), (2), (4) \Leftrightarrow (1), (2), (3)
  \end{align*}
\begin{align*}
\tSlice{\tPoint}(\query(\repr(\db)) = \query(\tSlice{\tPoint}(\repr(\db)))  \tag{4}
\end{align*}
\end{lem}
}{not needed anymore}

\BGDel{When proving our interval-based semiring encodings to be representation systems, we will make use of this lemma.}{No longer needed}

\section{Temporal $\semK$-elements}
\label{sec:temporal-k-normalform}

We now introduce temporal $\semK$-elements that are the annotations we use to
define our logical model (representation system).  Temporal
$\semK$-elements record, using an interval-based encoding, how the
$\semK$-annotation of a tuple in a \SKrel{} changes over time. We introduce a
unique normal form for temporal $\semK$-elements based on a generalization of
coalescing~\cite{DBLP:conf/vldb/BohlenSS96}.

\subsection{Defining Temporal $\semK$-elements}
\label{sec:concrete-temporal-k}

To define temporal $\semK$-elements, we need to introduce some background on
intervals. Given the time domain $\timeDomain$ and its associated
total order $\tLeq$, an interval $\interval = [\tBegin, \tEnd)$ is a
pair of time points from $\timeDomain$, where $\tBegin \tLe \tEnd$.
Interval $\interval$ represents the set of contiguous time points
$\{ \tPoint \mid \tPoint \in \timeDomain \wedge \tBegin \tLeq \tPoint
\tLe \tEnd\}$.
For an interval $\interval = [\tBegin,\tEnd)$ we use
$\iBegin{\interval}$ to denote $\tBegin$ and $\iEnd{\interval}$ to
denote $\tEnd$. We use $I, I', I_1, \ldots $ to represent intervals.
We define a relation $\adjacent{\interval_1}{\interval_2}$
that contains all interval pairs that are adjacent:
$\adjacent{\interval_1}{\interval_2} \Leftrightarrow
(\iEnd{\interval_1} = \iBegin{\interval_2}) \vee (\iEnd{\interval_2}
= \iBegin{\interval_1})$.
We will implicitly understand set operations, such as
$t \in \interval$ or $\interval_1 \subseteq \interval_2$, to be
interpreted over the set of points represented by an
interval. Furthermore, $\interval \cap \interval'$ denotes the
interval that covers precisely the intersection of the sets of time
points defined by $\interval$ and $\interval'$ and
$\interval \cup \interval'$ denotes their union (only well-defined if
$\interval \cap \interval' \neq \emptyset$ or
$\adjacent{\interval}{\interval'}$). For convenience, we define
$\interval \cup \interval' = \emptyset$ iff
$\interval \cap \interval' = \emptyset \land \neg
\adjacent{\interval}{\interval'}$.
We use $\intervalDom$ to denote the set of all intervals over
$\timeDomain$.

\begin{defi}[Temporal $\semK$-elements]
Given a semiring $\semK$, a \emph{temporal $\semK$-element} $\anyTE$ is a function $\intervalDom \to \semK$. We use $\cTEDom{\semK}$ to denote the set of all such temporal elements for $\semK$.
\end{defi}

We represent temporal $\semK$-elements  as sets of input-output pairs.  Intervals that are not explicitly mentioned are mapped
to $0_\semK$. \BGDel{For instance,
  $\{ [\tPoint_1,\tPoint_2] \mapsto k_1, [\tPoint_3,\tPoint_{10}]
  \mapsto k_2 \}$
  is the temporal element that maps all time intervals to $0_\semK$
  except for $[\tPoint_1,\tPoint_2]$ which is mapped to $k_1$ and
  $[\tPoint_3,\tPoint_{10}]$ which is mapped to $k_2$.}{}

\begin{exam}\label{ex:temp-k-elem}
Reconsider our running example with $\timeDomain = \{ 00, \ldots, 23
  \}$.
  The history of the annotation of tuple $t = $ \texttt{(Ann,SP)}
  from the \texttt{works} relation is as shown in
  Figure~\ref{fig:overview-approach} (middle). For sake of the example, we change the multiplicity of this tuple to $3$ during $[03,09)$ and $2$ during $[18,20)$. This information is encoded as the temporal
  $\semN$-element $\anyTE_1 = \{[03, 09) \mapsto 3, [18,20) \mapsto 2\}$. \end{exam}

Note that a temporal $\semK$-element $\anyTE$ may map overlapping
intervals to non-zero elements of $\semK$.
We assign the following semantics to overlap: the
annotation at a time point $\tPoint$ recorded by $\anyTE$ is the
sum of the annotations assigned to intervals containing $\tPoint$.
For instance, the annotation at time $04$ for the  $\semN$-element $\anyTE = \{ [00, 05) \mapsto 2, [04, 05) \mapsto 1 \}$ would be $2 + 1 = 3$.
To extract the annotation
valid at time $\tPoint$ from a temporal $\semK$-element $\anyTE$, we
define a timeslice operator for temporal $\semK$-elements as follows:
\begin{align*}\label{eq:temp-k-elem-timeslice}
  &\tSlice{\tPoint}(\anyTE) = \sum_{\tPoint \in \interval} \anyTE(\interval) \tag{timeslice operator}
\end{align*}

Given two temporal $\semK$-elements $\anyTE_1$ and $\anyTE_2$, we
would like to know if they represent the same history of annotations.
For that, we define \emph{snapshot-equivalence} ($\intervalEq$) for
temporal $\semK$-elements:
\begin{align*}\label{eq:temp-k-element-equivalence}
\anyTE_1 \intervalEq \anyTE_2 \Leftrightarrow \forall \tPoint \in \timeDomain: \tSlice{\tPoint}(\anyTE_1) = \tSlice{\tPoint}(\anyTE_2) \tag{snapshot-equivalence}
\end{align*}

\BGDel{In the temporal database literature, equivalent but representationally different temporal databases have sometimes been called \textit{snapshot-equivalent}. Relation  $\intervalEq$ is  snapshot-equivalence for temporal $\semK$-elements.}{earlier}

\subsection{A \revc{Normal Form} Based on $\semK$-Coalescing}
\label{sec:coalescing}

The encoding of the annotation history of a tuple as a temporal $\semK$-element is typically not unique.

\begin{exam}\label{ex:non-unique-concrete-temp-k-elem}
Reconsider the temporal $\semN$-element $\anyTE_1$ from Example~\ref{ex:temp-k-elem}. Recall that intervals not shown are mapped to $0$. The $\semN$-elements shown below are snapshot-equivalent to $\anyTE_1$.
  \begin{align*}
    &\anyTE_2 = \{[03, 09) \mapsto 1, [03,06) \mapsto 2, [06,09) \mapsto 2, [18,19) \mapsto 2\} \\[2mm]
    &\anyTE_3 = \{[03, 05) \mapsto 3, [05,09) \mapsto 3, [18,19) \mapsto 2\}
  \end{align*}
\end{exam}

To be able to build a representation system based on temporal $\semK$-elements we need a unique way to encode the annotation history of a tuple as a  temporal $\semK$-element (condition 1 of Definition~\ref{def:repr-system}). \BGDel{Also from a user perspective there is a benefit in representing a query result in a unique and predictable way.}{Somewhere else?}
That is, we need to define a normal form that is unique for snapshot-equivalent temporal $\semK$-elements. To this end, we generalize \textit{coalescing}, which was defined for temporal databases with set semantics in~\cite{DBLP:journals/tods/Snodgrass87,DBLP:conf/vldb/BohlenSS96}. The generalized form, which we call $\semK$-coalescing, coincides with standard coalescing for semiring $\semB$ (set semantics) and, for any semiring $\semK$, yields a unique encoding.

$\semK$-coalescing creates maximal intervals of contiguous time points
with the same annotation. The output is a temporal $\semK$-element such that
(a) no two intervals mapped to a non-zero element overlap and (b)
adjacent intervals assigned to non-zero elements are guaranteed to be
mapped to different annotations. To determine such intervals, we
define annotation changepoints, time points $\tPoint$ where the annotation of a temporal $\semK$-element differs from the annotation at $\tPoint - 1$, i.e., $\tSlice{\tPoint}(\anyTE) \neq \tSlice{\tPoint-1}(\anyTE)$). It will be convenient to also
consider $\tMin$ as an annotation
changepoint.

\begin{defi}[Annotation Changepoint]\label{def:cp}
  Given a temporal $\semK$-element $\anyTE$, a time point $\tPoint$ is called a \emph{changepoint} in $\anyTE$ if one of the following conditions holds:
  \begin{itemize}
  \item  $\tPoint = \tMin$     \hfill(smallest  time point)
  \item $\tSlice{\tPoint - 1}(\anyTE) \neq \tSlice{\tPoint}(\anyTE)$ \hfill(change of annotation)
  \end{itemize}
  We use $\CPs{\anyTE}$ to denote the set of all annotation
  changepoints for $\anyTE$. Furthermore, we define $\CPIs{\anyTE}$ to
  be the set of all intervals that consist of consecutive change
  points:
  \begin{align*}
    \CPIs{\anyTE} &= \{ [\tPoint_b, \tPoint_e) \mid \tPoint_b \tLe \tPoint_e \wedge \tPoint_b \in \CPs{\anyTE} \wedge{}
    \\
                  &\mathtab\mathtab(\tPoint_e \in \CPs{\anyTE} \vee T_e = \tMax) \wedge{}
    \\
                  &\mathtab\mathtab \not\exists \tPoint' \in \CPs{\anyTE}: \tPoint_b \tLe \tPoint' \tLe \tPoint_e \}
  \end{align*}
\end{defi}

In Definition~\ref{def:cp}, $\CPIs{\anyTE}$ computes maximal intervals
such that the annotation assigned by $\anyTE$ to each point in such an
interval is constant. In the coalesced representation of
$\anyTE$ only such intervals are mapped to non-zero annotations.

\begin{defi}[$\semK$-Coalesce]
 Let $\anyTE$  be a  temporal $\semK$-element.
 We define $\semK$-coalescing $\kCoalesce{\semK}$ as a function $\cTEDom{\semK} \to \cTEDom{\semK}$:
\begin{align*}
\kCoalesce{\semK}(\anyTE)(\interval) & =
                                     \begin{cases}
                                       \tSlice{\iBegin{\interval}}(\anyTE) &\text{if}\; \interval \in \CPIs{\anyTE}\\
                                       \zeroK &\text{otherwise}
                                     \end{cases}
  \end{align*}
  We use $\nTEDom{\semK}$ to denote all normalized temporal $\semK$-elements, i.e., elements $\anyTE$ for which $\kCoalesce{\semK}(\anyTE) = \anyTE'$ for some $\anyTE'$.

\end{defi}

\begin{figure}
  \centering
  \scriptsize
  \renewcommand{\arraystretch}{1.2}

  \begin{tabular}[t]{|l|l|}
    \hline
    \thead{sal} & \thead{period} \\
    \hline
    $50$k & $[1,13)$ \\
    $30$k & $[3,13)$ \\
    $30$k & $[3,10)$ \\
    $40$k & $[11,13)$ \\
    \hline
  \end{tabular}
  \hspace{1cm}
  \begin{minipage}[t]{0.35\columnwidth}
    \begin{align*}
      &\anyTE_{50k} = \{[1,13) \mapsto 1 \}\\[1mm]
      &\anyTE_{30k} = \{[3,10) \mapsto 1, [3,13) \mapsto 1 \}\\[1mm]
      &\anyTE_{40k} = \{[11,13) \mapsto 1 \}
    \end{align*}
  \end{minipage}

  \caption{Example period multiset relation $S$ and temporal
    $\semN$-elements encoding the history of tuples.}
  \label{fig:example_temp_coalesce}
\end{figure}

\begin{exam}\label{ex:k-coalesce}
  Consider the \SQLrel{} shown in
  Figure~\ref{fig:example_temp_coalesce}. The temporal
  $\semN$-elements encode the history of tuples $(30k)$,
  $(40k)$ and $(50k)$. Note that $\anyTE_{30k}$ is not coalesced since
  the two non-zero intervals of this   $\semN$-element
  overlap. Applying $\semN$-coalesce we get:
  \begin{align*}
    \kCoalesce{\semN}(\anyTE_{30k}) = \{ [3,10) \mapsto 2, [10,13) \mapsto 1 \}
  \end{align*}
  That is, this tuple occurs twice within the time interval $[3,10)$
  and once in $[10,13)$, i.e., it has annotation changepoints
  $3$, $10$, and $14$. Interpreting the same relation under set
  semantics (semiring $\semB$), the history of $(30k)$ can be encoded
  as a temporal $\semB$-element
  ${\anyTE_{30k}}' = \{ [3,10) \mapsto true, [3,13) \mapsto true
  \}$. Applying $\semB$-coalesce:   \begin{align*}
    \kCoalesce{\semB}({\anyTE_{30k}}') = \{ [3,13) \mapsto true \}
  \end{align*}
  That is, this tuple occurs (is annotated with $true$) within the time interval $[3,13)$ and its annotation changepoints are $3$ and $14$.
\end{exam}

We now prove several important properties of the $\semK$-coalesce operator
establishing that $\nTEDom{\semK}$ (coalesced temporal $\semK$-elements) is a good choice for a \revc{normal form} of
 temporal $\semK$-elements.

\begin{lem}\label{lem:coalesce-properties}
  Let $\semK$ be a semiring and $\anyTE$, $\anyTE_1$ and $\anyTE_2$  temporal $\semK$-elements. We have:
  \begin{align*}
    \kCoalesce{\semK}(\kCoalesce{\semK}(\anyTE)) &= \kCoalesce{\semK}(\anyTE) \tag{idempotence}\\
    \anyTE_1 \intervalEq \anyTE_2 &\Leftrightarrow \kCoalesce{\semK}(\anyTE_1) = \kCoalesce{\semK}(\anyTE_2) \tag{uniqueness}\\
    \anyTE &\intervalEq \kCoalesce{\semK}(\anyTE) \tag{equivalence preservation}
  \end{align*}
\end{lem}
  \proofsketch{We provide the proofs for all lemmas and theorems in an accompanying report~\cite{DG18}, \revm{but show sketches for important proofs here. Equivalence preservation and uniqueness follow annotation change points being defined based on snapshots. Uniqueness plus equivalence preservation implies idempotence.}}
  \iftechreport{
\begin{proof}
  All proofs are shown in Appendix~\ref{sec:proofs}.
\end{proof}
}

\section{Period Semirings}
\label{sec:interv-temp-annot}

Having established a unique normal form of temporal $\semK$-elements,
we now proceed to define period semirings as our logical model. The elements of a period semiring are temporal
$\semK$-elements in normal form. We prove that these structures are
semirings and ultimately that relations annotated with period semirings form a
representation system for snapshot $\semK$-relations for $\raPlus$. In
Section~\ref{sec:complex-queries}, we then prove them to also be a
representation system for $\raAgg$, i.e., queries involving difference
and aggregation.

When defining the addition and multiplication operations and their
neutral elements in the semiring structure of temporal
$\semK$-elements, we have to ensure that these definitions are
compatible with semiring $\semK$ on snapshots. Furthermore, we need to
ensure that the output of these operations is guaranteed to be $\semK$-coalesced. The latter can be ensured by applying $\semK$-coalesce to
the output of the operation. For addition, snapshot reducibility
is achieved by pointwise addition (denoted as $\addP$)
of the two functions that constitute the two input temporal
$\semK$-elements. That is, for each interval $\interval$, the function
that is the result of the addition of temporal $\semK$-elements $\anyTE_1$ and
$\anyTE_2$ assigns to $\interval$ the value
$\anyTE_1(\interval) +_{\semK} \anyTE_2(\interval)$. For
multiplication, the multiplication of two $\semK$-elements assigned to
an overlapping pair of intervals $\interval_1$ and $\interval_2$ is
valid during the intersection of $\interval_1$ and
$\interval_2$. Since both input temporal $\semK$-elements may assign non-zero
values to multiple intervals that have the same overlap, the resulting
$\semK$-value at a point $\tPoint$ would be the sum over all pairs of
overlapping intervals. We denote this operation as $\multP$. Since
$\addP$ and $\multP$ may return a temporal $\semK$-element that is not
coalesced, we define the operations of our structures to apply
$\kCoalesce{\semK}$ to the result of $\addP$ and $\multP$.  The zero
element of the temporal extension of $\semK$ is the temporal $\semK$-element
that maps all intervals to $0$ and the $1$ element is the temporal
element that maps every interval to $0_{\semK}$ except for
$[\tMin, \tMax)$ which is mapped to $1_{\semK}$.

\begin{defi}[Period Semiring]
  For a time domain $\timeDomain$ with minimum $\tMin$ and maximum
  $\tMax$ and a semiring $\semK$, the period semiring
  $\semTimeNI$ is defined as:
  $$\semTimeNI = (\nTEDom{\semK}, \addNI, \multNI, \zeroNI, \oneNI)$$
  where for $k,k' \in \nTEDom{\semK}$ and :
  \begin{align*}
        \forall \interval \in \intervalDom: \zeroNI(\interval) &= 0_{\semK} &
                                                                              \oneNI(\interval) &=
                             \begin{cases}
                               1_{\semK} &\text{if}\; \interval=[\tMin,\tMax)\\
                               0_{\semK} &\text{otherwise}
                             \end{cases}
  \end{align*}\\[-8mm]
  \begin{align*}
    k \addNI k' &= \kCoalesce{\semK}(k \addP k')\\
\forall \interval \in \intervalDom:
    (k \addP k')(\interval) &= k(\interval) +_{\semK} k'(\interval)\\
    k \multNI k' &= \kCoalesce{\semK}(k \multP k')\\
\forall \interval \in \intervalDom:    (k \multP k')(\interval) &= \sum_{\forall \interval', \interval'': \interval = \interval' \cap \interval''} k(\interval') \cdot_{\semK} k'(\interval'')
  \end{align*}

\end{defi}

\begin{exam}\label{ex:interval-sem-storage-and-query-evaluation}
  Consider the $\semTimeNIN$-relation \texttt{works} shown in
  Figure~\ref{fig:overview-approach} (middle) and query
  $\projection_{skill}(works)$. Recall that the annotation of a tuple $\tuple$
  in the result of a projection over a $\semK$-relation is the sum of all
  input tuples which are projected onto $\tuple$. For result tuple
  \texttt{(SP)} we have input tuples \texttt{(Ann,SP)} and \texttt{(Sam,SP)} with $\anyTE_1 = \{[03,10) \mapsto 1,
  [18,20) \mapsto 1 \}$ and $\anyTE_{2} =
  \{[08,16) \mapsto 1\}$, respectively. The tuple \texttt{(SP)} is annotated
  with the sum of these annotations, i.e., $\anyTE_1 +_{\semTimeNIN} \anyTE_2$. Substituting
  definitions we get:
  \begin{align*}
    &\anyTE_1 +_{\semTimeNIN} \anyTE_2 = \kCoalesce{\semN}(\anyTE_1 \addPN \anyTE_2)\\ = &\kCoalesce{\semN}(\{[03,10) \mapsto 1,
  [18,20) \mapsto 1, [08,16) \mapsto 1 \}) \\ = &\{[03,08) \mapsto 1, [08,10) \mapsto 2, [10,16) \mapsto 1, [18,20) \mapsto 1 \}
  \end{align*}
  Thus, as expected, the result records that, e.g., there are two skilled workers (\emph{SP}) on duty during time interval $[08,10)$.
\end{exam}

Having defined the family of period semirings, it remains to be shown that $\semTimeNI$ with standard K-relational query semantics is a representation system for \SKrels{}.

\subsection{$\semTimeNI$ is a Semiring}
\label{sec:semtimeni-semiring}

As a first step, we prove that for any semiring $\semK$, the structure $\semTimeNI$ is  also  a semiring.
 The following lemma shows that $\semK$-coalesce can be redundantly pushed into $\addP$ and $\multP$ operations.

\begin{lem}\label{lem:coalesce-push}
  Let $\semK$ be a semiring and $k, k' \in \nTEDom{\semK}$. Then,
     \begin{align*}
     \kCoalesce{\semK}(k \addP k') &= \kCoalesce{\semK}(\kCoalesce{\semK}(k) \addP k') \\
       \kCoalesce{\semK}(k \multP k')
      &= \kCoalesce{\semK}(\kCoalesce{\semK}(k) \multP k')
   \end{align*}
\end{lem}

Using this lemma, we now prove that for any semiring $\semK$, the structure $\semTimeNI$ is also a semiring.

\begin{theo}\label{theo:timeib-is-semiring}
For any semiring $\semK$, structure $\semTimeNI$ is a semiring.
\end{theo}
\proofsketch{
\revm{We prove that $\semTimeNI$ obeys the laws of a commutative semiring based on $\semK$ being a semiring and Lemma~\ref{lem:coalesce-push}.}
  }

\subsection{Timeslice Operator}
\label{sec:time-slic-repr}

We define a timeslice operator for $\semTimeNI$-relations based on the timeslice operator for temporal $\semK$-elements. We annotate each tuple in the output of this operator with the result of  $\tSlice{\tPoint}$ applied to the temporal $\semK$-element the tuple is annotated with.

\begin{defi}[Timeslice for $\semTimeNI$-relations]
  Let $\rel$ be a $\semTimeNI$-relation and $\tPoint \in \timeDomain$. The timeslice operator $\tSlice{\tPoint}(\rel)$ is defined as:
  \begin{align*}
    \tSlice{\tPoint}(\rel)(t) = \tSlice{\tPoint}(\rel(t))
  \end{align*}
\end{defi}

We now prove that the $\tSlice{\tPoint}$ is a homomorphism  $\semTimeNI \to \semK$. Since semiring homomorphisms commute with queries~\cite{GK07}, $\semTimeNI$ equipped with this timeslice operator does fulfill the snapshot-reducibility condition of representation systems (Definition~\ref{def:repr-system}).

\begin{theo}\label{theo:hib-is-homomorphism}
For any $\tPoint \in \timeDomain$, the timeslice operator $\tSlice{\tPoint}$ is a semiring homomorphism from $\semTimeNI$ to $\semK$.
\end{theo}
\proofsketch{\revm{
Proven by substitution of definitions and by regrouping terms using semiring laws.
}}

As an example of the application of this homomorphism, consider  the period $\semN$-relation \texttt{works} from our running example as shown on the left of Figure~\ref{fig:overview-approach}. Applying $\tSlice{08}$ to this relation yields the snapshot shown on the bottom of this figure (three employees work between 8am and 9am out of whom two are specialized). If we evaluate query $Q_{onduty}$ over this snapshot we get the snapshot shown on the right of this figure (the count is 2). By Theorem~\ref{theo:hib-is-homomorphism} we get the same result if we evaluate $Q_{onduty}$ over the input period $\semN$-relation and then apply $\tSlice{08}$ to the result.

\subsection{Encoding of Snapshot \semK-relations}
\label{sec:encod-snapsh-semk}

We now define a bijective mapping $\repr_{\semK}$ from \SKrels{} to $\semTimeNI$-relations. We then prove that the set of $\semTimeNI$-relations together with the timeslice operator for such relations and the mapping $\reprInv{\semK}$ (the inverse of $\repr_{\semK}$) form a representation system for \SKrels{}. Intuitively, $\repr_{\semK}(\rel)$ is constructed by assigning each tuple $\tuple$ a  temporal $\semK$-element where the annotation of the tuple at time $\tPoint$ (i.e., $\rel(\tPoint)(\tuple)$) is assigned to a singleton interval $[\tPoint,\tPoint+1)$. This temporal $\semK$-element $\anyTE_{R,\tuple}$ is then coalesced to create a $\nTEDom{\semK}$ element.

\begin{defi}\label{def:snapshot-K-relation-encoding}
  Let $\semK$ be a semiring and $\rel$ a \SKrel{},
  $\repr_{\semK}$ is a mapping from \SKrels{} to
  $\semTimeNI$-relations defined as follows.
  \begin{align*}
\forall t:  \repr_{\semK}(\rel)(\tuple) &= \kCoalesce{\semK}(\anyTE_{R,\tuple})\\
\forall t, I: \anyTE_{R,\tuple}(I) &=
    \begin{cases}
      \rel(\tPoint)(\tuple) &\mathtext{if} I = [\tPoint,\tPoint+1)\\
      0_\semK &\mathtext{otherwise}
    \end{cases}
  \end{align*}
\end{defi}

We first prove that this mapping is bijective, i.e., it is invertible,
which guarantees that $\reprInv{\semK}$ is well-defined and also implies uniqueness (condition 1 of
Definition~\ref{def:repr-system}).

\begin{lem}\label{lem:encoding-is-bijective}
For any semiring $\semK$,  $\repr_{\semK}$ is bijective.
\end{lem}

Next, we have to show that $\repr_{\semK}$ preserves snapshots, i.e., the instance at a time point $\tPoint$ represented by $\rel$ can be extracted from $\repr_{\semK}(\rel)$ using the timeslice operator.

\begin{lem}\label{lem:encoding-timeslice-is-compatible}
  For any semiring $\semK$, \SKrel{} $\rel$, and time point $\tPoint \in \timeDomain$, we have
    $\tSlice{\tPoint}(\repr_{\semK}(\rel)) = \tSlice{\tPoint}(\rel)$.
\end{lem}

Based on these properties of $\repr_{\semK}$ and the fact that the timeslice operator over $\semTimeNI$-relations is a homomorphism $\semTimeNI \to \semK$, our main technical result follows immediately. That is, the set of $\semTimeNI$-relations equipped with the timeslice operator and $\reprInv{\semK}$ is a representation system for positive relational algebra queries ($\raPlus$) over \SKrels{}.

\begin{theo}[Representation System]\label{theo:repr-system}
  Given a semiring $\semK$, let $\domTimeRel{\semK}$ be the set of all $\semTimeNI$-relations. The triple $(\domTimeRel{\semK}, \reprInv{\semK}, \tSlice{})$ is a representation system for $\raPlus$ queries over \SKrels{}.
\end{theo}
\proofsketch{\revm{
We have to show that conditions (1), (2), and (3)  of Definition~\ref{def:repr-system} hold. Conditions (1) and (2) have been proven in Lemmas~\ref{lem:encoding-is-bijective} and~\ref{lem:encoding-timeslice-is-compatible}, respectively. Condition (3) follows from the fact that $\tSlice{\tPoint}$ is a homomorphism (Theorem~\ref{theo:hib-is-homomorphism}) and that homomorphisms commute with $\raPlus$-queries~\cite[Proposition 3.5]{GK07}.
  }}

\section{Complex Queries}
\label{sec:complex-queries}

Having proven that
$\semTimeNI$-relations form a representation system for $\raPlus$, we
now study extensions for difference and aggregation.

\subsection{Difference}
\label{sec:set-difference}

Extensions of $\semK$-relations for difference have been studied
in~\cite{GP10,AD11a}. \revm{For instance, the difference operator on $\semN$
  relations corresponds to bag difference (SQL's \lstinline!EXCEPT ALL!).}
  Geerts et al.~\cite{GP10} apply an extension of semirings with a monus
  operation that is defined based on the \textit{natural order} of a semiring
  and demonstrated how to define a difference operation for $\semK$-relations
  based on the monus operation for semirings where this operations is
  well-defined.
Following the terminology introduced in this work, we refer
to semirings with a monus operation as m-semirings. We now prove that
if a semiring $\semK$ has a well-defined monus, then so does
$\semTimeNI$. From this follows, that for any such $\semK$, the
difference operation is well-defined for $\semTimeNI$. We proceed to
show that the timeslice operator is an m-semiring homomorphism, which
implies that $\semTimeNI$-relations for any m-semiring $\semK$ form a
representation system for $\ra$ (full relational algebra).  The
definition of a monus operator is based on the so-called natural order
$\naturalOrder_{\semK}$. For two elements $k$ and $k'$ of a semiring
$\semK$,
$k \naturalOrder_{\semK} k' \Leftrightarrow \exists k'': k \addK k'' =
k'$.
If $\naturalOrder_{\semK}$ is a partial order then $\semK$ is called
\textit{naturally ordered}. For instance, $\semN$ is naturally ordered
($\naturalOrder_{\semN}$ corresponds to the order of natural numbers)
while $\mathbb{Z}$ is not (for any $k,k' \in \mathbb{Z}$ we have
$k \naturalOrder_{\mathbb{Z}} k'$). For the monus to be well-defined
on $\semK$, $\semK$ has to be naturally ordered and for any
$k, k' \in \semK$, the set
$\{ k'' \mid k \naturalOrder_{\semK} k' \addK k''\}$ has to have a
smallest member. For any semiring fulfilling these two conditions, the
monus operation $-_{\semK}$ is defined as $k -_{\semK} k' = k''$ where
$k''$ is the smallest element such that
$k \naturalOrder_{\semK} k' + k''$. For instance, the
monus for $\semN$ is the truncating minus: $k -_{\semN} k' = max(0, k - k')$.

\begin{theo}\label{theo:bag-interval-norm-monus}
For any m-semiring $\semK$, semiring $\semTimeNI$ has a well-defined monus, i.e., is an m-semiring. \end{theo}
\proofsketch{\revm{
    We first prove that $\semK$ being naturally ordered implies that $\semTimeNI$ is naturally ordered where $k \naturalOrder_{\semTimeNI} k'$ is $\tSlice{\tPoint}(k) \naturalOrder_{semK} \tSlice{\tPoint}(k')$ for all time points $\tPoint$. Then we show constructively that for any $k$ and $k'$, the set $\{ k'' \mid k \naturalOrder_{\semTimeNI} k' + k''\}$ has a smallest element wrt. $\naturalOrder_{\semTimeNI}$. These two condition are the only requirements for a semiring to have a well-defined monus.
}}

Let $k \monP k'$ denote an operation that returns a temporal $\semK$-element which assigns to each singleton interval $[\tPoint,\tPoint+1)$ the result of the monus for $\semK$: $\timeSlice{\tPoint}(k) \monK \timeSlice{\tPoint}(k')$ \ifnottechreport{.}\iftechreport{(this is $k_{pmin}$ as defined in the proof of Theorem~\ref{theo:bag-interval-norm-monus}, see Appendix~\ref{sec:proofs}).}
In the proof of Theorem~\ref{theo:bag-interval-norm-monus}, we demonstrate that $k \monNI k' = \kCoalesce{\semK}(k \monP k')$. Obviously, computing $k \monP k'$ using singleton intervals is not effective. In our implementation, we use a more efficient way to compute the monus for $\semTimeNI$ that is based on normalizing the input temporal $\semK$-elements $k$ and $k'$ such that annotations are attached to larger time intervals where $k \monP k'$ is guaranteed to be constant.
Importantly, $\tSlice{\tPoint}$ is a homomorphism for monus-semiring $\semTimeNI$.

\begin{theo}\label{theo:bag-interval-norm-monus-homo}
  Mapping $\tSlice{\tPoint}$ is an m-semiring homomorphism. \end{theo}
\proofsketch{\revm{
    Proven by substitution of definitions.
  }}

For example, consider $Q_{skillreq}$ from Example~\ref{ex:bag-difference-running-example} which can be expressed in relational algebra as $\projection_{skill}(assign) - \projection_{skill}(worker)$. The $\semTimeNIN$-relation corresponding to the period relation \texttt{assign} shown in this example annotates each tuple with a singleton temporal $\semN$-element mapping the period of this tuple to $1$, e.g., \texttt{(M1, SP)} is annotated with $\{ [03,12) \mapsto 1 \}$. The annotation of result tuple \texttt{(SP)} is computed as\\[-8mm]
\begin{center}
{\small
\begin{align*}
  &(\{ [03,12) \mapsto 1 \} \addNIN \{ [06,14) \mapsto 1 \})\\
    &\hspace{2mm}\monNIN (\{ [03,10) \mapsto 1 \} \addNIN \{ [08,16) \mapsto 1 \} \addNIN \{ [18,20) \mapsto 1 \}) \\
  = &\{ [03,06) \mapsto 1, [06,12) \mapsto 2, [12,14) \mapsto 1\}\\
      &\hspace{2mm}\monNIN \{ [03,08) \mapsto 1, [08,10)      \mapsto 2, [10,16) \mapsto 1, [18,20) \mapsto 1 \} \\
  = &\{[06,08) \mapsto 1, [10,12) \mapsto 1 \}
\end{align*}
}
\end{center}
As expected, the result is the same as the one from Example~\ref{ex:bag-difference-running-example}.

\subsection{Aggregation}
\label{sec:aggregation}

The K-relational framework has previously been extended to support
aggregation~\cite{AD11d}. This required the introduction of attribute
domains which are symbolic expressions that pair values with semiring
elements to represent aggregated values.  Since the construction used
in this work to derive the mathematical structures representing these
symbolic expressions is applicable to all semirings, it is also
applicable to our period semirings. It was shown that
semiring homomorphisms can be lifted to these more complex annotation
structures and attribute domains. Thus, the timeslice operator, being
a semiring homomorphism, commutes with queries including aggregation,
and it follows that using the approach from~\cite{AD11d}, we can
define a representation system for \SKrels{} under
$\ra$ with aggregation, i.e., $\raAgg$.

One drawback of this definition of aggregation over K-relations with
respect to our use case is that there are multiple ways of encoding
the same \SKrel{} in this model. That is, we would
loose uniqueness of our representation system.  Recall that one of our
major goals is to implement snapshot query semantics on-top of DBMS
using a period multiset encoding of
$\semTimeNIof{\semN}$-relations.  The symbolic expressions
representing aggregation function results are a compact representation
which, in case of our interval-temporal semirings, encode how the
aggregation function results change over time.  However, it is not
clear how to effectively encode the symbolic attribute values and
comparisons of symbolic expression as multiset semantics relations,
and how to efficiently implement our snapshot semantics over this
encoding. Nonetheless, for $\semN$, we can apply a simpler definition
of aggregation that returns a $\semTimeNIof{\semK}$ relation and is
also a representation system.  \BGDel{Furthermore, in
  Appendix~\ref{sec:aggr-using-semim}, we show that for an important
  class of semirings, it is possible to evaluate these symbolic
  expressions~\cite{AD11d} to generate a representation of the result
  relation produced by an aggregation as a $\semTimeNIof{\semK}$
  relation which is equivalent under snapshots (homomorphism
  $\timeSlice{\tPoint}$).}

For simplicity, we define aggregation $\aggregation{G}{f(A)}(R)$ grouping on $G$ to compute a single aggregation function $f$ over the values of an attribute $A$. For convenience, aggregation without group-by, i.e., $\aggregation{}{f(A)}(R)$ is expressed using an empty group-by list.

\begin{defi}[Aggregation]\label{def:aggregation-nin}
 Let $R$ be a $\semTimeNIof{\semN}$ relation.  Operator $\aggregation{G}{f(A)}(R)$ groups the input on a (possibly empty) list of attributes $G = (g_1, \ldots, g_n)$ and computes aggregation function $f$ over the values of attribute $A$.
  This operator is defined as follows:
  \begin{align*}
    &\aggregation{G}{f(A)}(R)(t) = \kCoalesce{\semN}(k_{R,t})\\
    &\quad k_{R,t} (I) =
    \begin{cases}
      1 &\mathtext{if} \exists \tPoint: I = [\tPoint,\tPoint{+}1) \wedge t \in \aggregation{G}{f(A)}(\tSlice{\tPoint}(R))\\
      0 & \mathtext{otherwise}\\
    \end{cases}
  \end{align*}
\end{defi}

In the output of the aggregation operator, each tuple $\tuple$ is annotated with a $\semN$-coalesced temporal $\semN$-element which is constructed from singleton intervals. A singleton interval $I= [\tPoint,\tPoint+1)$ is mapped to $1$ if evaluating the aggregation over the multiset relation corresponding to the snapshot at $\tPoint$ returns tuple $\tuple$. We now demonstrate that $\semTimeNIof{\semN}$
using this definition of aggregation
is a representation system for snapshot $\semN$-relations.

\begin{theo}\label{theo:our-agg-eq-their-agg}
  $\semTimeNIN$-relations
form a representation system for snapshot $\semN$-relations and $\raAgg$ queries using aggregation according to Definition~\ref{def:aggregation-nin}.
\end{theo}
\proofsketch{\revm{
By construction, the result of aggregation is a $\semTimeNIN$ relation (it is coalesced). Also by construction, we have $\tSlice{\tPoint}(\aggregation{G}{f(A)}(R)) = \aggregation{G}{f(A)}(\tSlice{\tPoint}(R))$.
    }}

\section{SQL Period Relation Encoding} \label{sec:bag-semant-repr}

While provably correct, the annotation structure that we have defined is quite complex in nature raising concerns on how to efficiently implement it. We now demonstrate that $\semTimeNIN$-relations (multisets) can be encoded as \SQLrels{} (as shown  on the top of Figure~\ref{fig:overview-approach}).
Recall that \SQLrels{} are multiset relations where the validity time interval (period) of a tuple is stored in an interval-valued attribute (or as two attributes storing interval end points).
Queries over $\semTimeNIN$ are then translated into non-temporal multiset queries over this encoding.
In addition to employing a proven and simple representation of time this enables our approach to run snapshot queries over such relations without requiring any preprocessing and to implement our ideas on top of a
classical DBMS.
For convenience we represent \SQLrels{} using non-temporal $\semN$-relations in the definitions. \SQLrels{} can be obtained based on the well-known correspondence between multiset relations and $\semN$-relations: we duplicate each tuple based on the multiplicity recorded in its annotation.
To encode $\semTimeNIN$-relations as $\semN$-relations we introduce an invertible mapping $\reprN$. We
rewrite queries with $\semTimeNIN$-semantics into non-temporal queries with
$\semN$-semantics over this encoding using a rewriting function
$\reprRewr$. This is illustrated in the commutative diagram below.
  \begin{equation}\label{eq:n-enc-cd}
  \begin{tikzpicture}[baseline=(current  bounding  box.center)]
    \tikzstyle{every node}=[font=\scriptsize]
    \node[draw, minimum width=1.2cm, minimum height=0.5cm,rounded corners=2pt] (R) {$\rel$};     \node[draw, minimum width=1.2cm, minimum height=0.5cm,rounded corners=2pt] (Rp) at ($(R)+(3cm,0)$) {$\rel'$};     \node[draw, minimum width=1.2cm, minimum height=0.5cm,rounded corners=2pt] (Q)  at ($(R)+(0,-1cm)$) {$\query(\rel)$};     \node[draw, minimum width=1.2cm, minimum height=0.5cm,rounded corners=2pt] (Qp) at ($(R)+(3cm,-1cm)$) {$\query'(\rel')$};     \draw[->, line width=1pt] (R) -- node[above]{$\reprN$} (Rp);     \draw[->, line width=1pt] (R) -- node[left]{$\query$} (Q);     \draw[->, line width=1pt] (Rp) -- node[right]{$\query'=\reprRewr(\query)$} (Qp);     \draw[->, line width=1pt] (Qp) -- node[below]{$\reprN^{-1}$} (Q);   \end{tikzpicture}
  \end{equation}

  Our encoding represents a tuple $\tuple$ annotated with a temporal element $\anyTE$ as a set of tuples, one for each interval $\interval$ which is assigned a non-zero value by $\anyTE$. For each such interval, the interval's end points are stored in two attributes $\attrBegin$ and $\attrEnd$, which are appended to the schema of $\tuple$.   Again, we use $\tuple \mapsto k$ to denote that tuple $\tuple$ is annotated with $k$ and $\valDom$ to denote a universal domain of values. We use
  $\schemaOf{R}$ to denote the schema of relation $R$ and
  $\arity(R)$ to denote its arity (the number of attributes in the schema).

\begin{defi}[Encoding as SQL Period Relations]\label{def:N-enc}
 $\reprN$ is a function from $\semTimeNIN$-relations to $\semN$-relations.   Let $\rel$ be a $\semTimeNIN$ relation with schema $\schemaOf{\rel} = \{\attr_1, \ldots, \attr_n\}$. The schema of $\reprN(\rel)$ is $\{\attr_1, \ldots, \attr_n, \attrBegin, \attrEnd\}$.
  Let $\rel'$ be $\reprN(\rel)$ for some $\semTimeNIN$-relation.
  $\reprN$ and its inverse are defined as follows:
  \vspace{-3mm}
  \begin{align*}
    \reprN(R) &= \bigcup_{t \in \valDom^{\arity(R)}} \bigcup_{\interval \in \intervalDom} \{ (t, \iBegin{\interval}, \iEnd{\interval}) \mapsto R(t)(\interval) \}\\
    \reprN^{-1}(R') &= \bigcup_{t \in \valDom^{\arity(R)}} \{ t \mapsto \anyTE_{R',t} \}\\
  \forall \interval \in \intervalDom:  \anyTE_{R',t}(\interval) &= R'(\tuple_{\interval}) \mathtext{for} \tuple_{\interval} = (t, \iBegin{\interval}, \iEnd{\interval})
  \end{align*}
\end{defi}

Before we define the rewriting $\reprRewr$ that reduces a query
$\query$ with $\semTimeNIN$ semantics to a query with $\semN$
semantics, we introduce two operators that we will make use of in the
reduction. The $\semN$-coalesce operator applies $\kCoalesce{\semN}$
to the annotation of each tuple in its input.

\begin{defi}[Coalesce Operator]\label{def:coalesce-op}
Let $\rel$ be $\reprN(\rel')$ for some $\semTimeNIN$-relation $\rel'$.
The coalesce operator $\coalesceOp(\rel)$ is defined as:
  \begin{align*}
    \coalesceOp(R) &= \reprN(R')\\
  \forall \tuple: R'(t) &= \kCoalesce{\semN}(\reprN^{-1}(\rel)(t))
  \end{align*}
\end{defi}

\begin{figure*}[t]
  \centering
    \begin{align*}
\underline{\reprRewr(R)} & = R & \underline{\reprRewr(\selection_\theta(\query))} &= \coalesceOp(\selection_\theta(\reprRewr(\query))) &
\underline{\reprRewr(\projection_A(\query))} &= \coalesceOp(\projection_{A,\attrBegin, \attrEnd}(\reprRewr(\query))) &
    \end{align*}\\[-6mm]
    \begin{align*}
    \underline{\reprRewr(\query_1 \join_\theta \query_2)} &= \coalesceOp(\projection_{\schemaOf{\query_1 \join_{\theta} \query_2}, max(\query_1.\attrBegin, \query_2.\attrBegin),min(\query_1.\attrEnd, \query_2.\attrEnd)} (\reprRewr(\query_1) \join_{\theta \wedge  overlaps(\query_1,\query_2)} \reprRewr(\query_2)))\\
      \underline{\reprRewr(\query_1 - \query_2)} &= \coalesceOp(\normalizeOp_{\schemaOf{\query_1}}(\reprRewr(\query_1), \reprRewr(\query_2)) -  \normalizeOp_{\schemaOf{\query_2}}(\reprRewr(\query_2), \reprRewr(\query_1)))
    \end{align*}\\[-6mm]
\revm{    \begin{align*}
            \underline{\reprRewr(\aggregation{}{f(A)}(\query))} &= \coalesceOp (\aggregation{\attrBegin, \attrEnd}{f(A)}(\normalizeOp_{\emptyset}(\reprRewr(\query) \union \{(null,\tMin, \tMax)\}, \reprRewr(\query))))\\             \underline{\reprRewr(\aggregation{}{count(*)}(\query))} &= \reprRewr(\aggregation{}{count(A)}(\projection_{1 \to A}(\query)))
\end{align*}}\\[-8mm]
\begin{align*}
  \underline{\reprRewr(\aggregation{G}{f(A)}(\query))} &= \coalesceOp (\aggregation{G,\attrBegin, \attrEnd}{f(A)}(\normalizeOp_{G}(\reprRewr(\query), \reprRewr(\query))))
&  \underline{\reprRewr(\query_1 \union \query_2)} &= \coalesceOp(\reprRewr(\query_1) \union \reprRewr(\query_2))
    \end{align*}\\   \caption{Rewriting $\reprRewr$ that reduces queries over $\semTimeNIN$ to queries over a multiset encoding produced by $\reprN$.}
  \label{fig:rewrite}
\end{figure*}

The split operator $\normalizeOp_{G}(R,S)$ splits the intervals in the
temporal elements annotating a tuple $\tuple$ based on the union of
all interval end points from annotations of tuples $\tuple'$ which
agree with $\tuple$ on attributes $G$. Inputs $R$ and $S$ have to be
union compatible. The effect of this operator is that all pairs of
intervals mapped to non-zero elements are either the same or are
disjoint. This operator has been applied
in~\cite{DignosBG12,DignosBGJ16} and in \cite{BowmanT03,T98}. We use
it to implement snapshot-reducible aggregation and difference over
intervals instead of single snapshots as in
Section~\ref{sec:complex-queries}. Recall that in Section~\ref{sec:complex-queries}, the monus (difference) and aggregation were defined in a point-wise manner. The split operator allows us to evaluate these operations over intervals directly by generating tuples with intervals for which the result of these operations is guaranteed to be constant.

\begin{defi}[Split Operator]
  The split operator $\normalizeOp_G(R_1,R_2)$ takes as input two $\semN$-relations $R_1$ and $R_2$ that are encodings of $\semTimeNIN$-relations. For a tuple $\tuple$ in such an encoding let $\interval(\tuple) = [\tuple.\attrBegin, \tuple.\attrEnd)$. The split operator is defined as:
  \begin{align*}
    \normalizeOp_G(R_1,R_2)(\tuple) &= split(\tuple, R_1, EP_G(R_1 \cup R_2,t))\\
    EP_G(R,t) &= \bigcup_{t' \in R: t'.G = t.G \wedge R(t') > 0} \{t'.\attrBegin \} \cup \{t'.\attrEnd\}\\     split(\tuple, R, EP) &= \sum_{\tuple': \interval(\tuple) \subseteq \interval(\tuple') \wedge \interval(\tuple) \in EPI(t, EP)} R(\tuple')\\
    EPI(t, EP) &= \{ [\tPoint_b , \tPoint_e) \mid  \tPoint_b \tLe \tPoint_e \wedge \tPoint_b \in EP \wedge\\ &\mathtab\mathtab(\tPoint_e  \in EP \vee T_e = \tMax) \wedge\\ &\mathtab\mathtab \not\exists \tPoint' \in EP: \tPoint_b \tLe \tPoint' \tLe \tPoint_e \}
  \end{align*}
\end{defi}

Note that the $\reprN$ and $\reprN^{-1}$ mappings are only used in the definitions of the coalesce and split algebra operators for ease of presentation. These  operators can be implemented as SQL queries executed over an $\reprN$-encoded relation.

\begin{defi}[Query Rewriting]\label{def:repr-N-rewriting}
We use $overlaps(\query_1, \query_2)$ as a shortcut for
$\query_1.\attrBegin < \query_2.\attrEnd \land \query_2.\attrBegin < \query_1.\attrEnd$.
  The definition of rewriting $\reprRewr$ is shown in Figure~\ref{fig:rewrite}.
Here $\{t\}$ denotes a constant relation with a single tuple $t$ annotated with $1$.
\end{defi}

\revm{
\begin{exam}\label{ex:rewriting}
  Reconsider query $\qAggEx$ from Example~\ref{ex:running-example-snapshot} and its results for the logical model and period relations (Figure~\ref{fig:overview-approach}).   In relational algebra, the input query is written as $\qAggEx = \aggregation{}{count(*)} (\underbrace{\selection_{skill=SP}(works)}_{\query_1})$. Applying $\reprRewr$ we get:
  \begin{align*}
\reprRewr(\qAggEx) &=
                    \coalesceOp (\aggregation{\attrBegin, \attrEnd}{count(A)}(\normalizeOp_{\emptyset}(\\
    &\hspace{-1mm}\projection_{1 \to A, \attrBegin, \attrEnd}(\reprRewr(\query_1)) \union \{(null,0, 24)\},\\ &\hspace{-1mm}\reprRewr(\query_1))))\\
\reprRewr(\query_1) &= \coalesceOp(\selection_{skill=SP}(works))
  \end{align*}
  Subquery $\reprRewr(\query_1)$  filters out the second tuple from the input (see Figure~\ref{fig:overview-approach}). The split operator is then applied to the union of the result of $\reprRewr(\query_1)$ and a tuple with the neutral element $null$ for the aggregation function and period $[\tMin, \tMax)$, where $\tMin = 0$ and $\tMax = 24$ for this example.
After the split $\normalizeOp$, the aggregation is evaluated grouping the input on $\attrBegin, \attrEnd$. The \lstinline!count! aggregation function then either counts a sequence of $1$s and a single  $null$ value producing the number of facts that overlap over the corresponding period $[\attrBegin, \attrEnd)$, or counts a single $null$ value over a ``gap'' producing $0$.
For instance, for $[08,10)$ there are two facts whose intervals cover this period (Ann and Sam) and, thus, $(2, [08,10))$ is returned by $\reprRewr(\qAggEx)$. While for for $[20,24)$ there are no facts and thus we get $(0, [20,24))$.
\end{exam}
}

\begin{theo}\label{theo:reduction-is-correct}
The commutative diagram in Equation~~\eqref{eq:n-enc-cd} holds.
\end{theo}
\proofsketch{\revm{
Proven by induction over query structure.
    }}

\section{Implementation}
\label{sec:implementation}

We have implemented the encoding and rewriting introduced in the
previous section in a middleware which supports snapshot multiset
semantics through an extension of SQL. To instruct the system to
interpret a subquery using snapshot semantics, the user encloses the
subquery in a \texttt{SEQ VT (...)} block. We assume that the inputs
to a snapshot query are encoded as period multiset relations, i.e.,
each relation has two temporal attributes that store the begin and end
timestamp of their validity interval. For each relation access within
a \texttt{SEQ VT} block, the user has to specify which attributes
store the period of a tuple. \BGDel{We choose to implement this
  middleware as an extension of the GProM system to benefit from the
  system's support for multiple database backends.}{we should not
  mention GProM anyways because of double blind and here it is not
  important that this is build based on GProM}

Our coalescing and split operators can be expressed in SQL. Thus, a
straightforward way of incorporating these operators into the
compilation process is to devise additional rewrites that produce the
relational algebra code for these operators where necessary. However,
preliminary experiments demonstrated that a naive implementation of
these operators is prohibitively expensive. \BGDel{ increasing the query
execution time by several orders of magnitude compared to non-temporal
queries in some cases.}{space}

We address this problem in two ways. First, we observe that it is sufficient to apply coalesce as a last step in a query instead of applying it as part of every operator rewrite. Applying this optimization, the rewritten version of a query will only contain one coalesce operator. Recall from Lemma~\ref{lem:coalesce-push} that coalescing can be redundantly pushed into the addition and multiplication operations of period semirings, e.g., $\kCoalesce{\semK}(k \addP k') = \kCoalesce{\semK}(\kCoalesce{\semK}(k) \addP k')$.
\ifnottechreport{We have proven that this Lemma also holds for monus~\cite{DG18}.}
\iftechreport{We prove that this Lemma also holds for monus in Appendix~\ref{sec:optimizations}.}
Interpreting this equivalence from right to left and applying it repeatedly to a semiring expression $e$, $e$ can be rewritten into an equivalent expression of the form $\kCoalesce{\semK}(e')$, where $e'$ is an expression that only uses operations $\addP$, $\multP$, $\monP$. Since relational algebra over K-relations is defined by applying multiplication, addition, and monus to input annotations, this implies that it is sufficient to apply coalescing only as a final operation in a query. \iftechreport{For an example and additional discussion see Appendix~\ref{sec:optimizations}}\ifnottechreport{For further discussion of this optimization and an example see~\cite[Appendix C]{DG18}.}

We developed an optimized implementation of multiset coalescing using SQL
analytical window functions\reva{, similar to set-based coalescing
in~\cite{ZhouWZ06}}, that counts \reva{for value-equivalent attributes} the
number of open intervals per time point, determines change points based on
differences between these counts, and then only output maximal intervals using
a filter step. \reva{This implementation uses sorting in its window declarations and has time complexity $\mathcal{O}(n \log{n})$ for $n$ tuples.
A native implementation would require only one sorting step. The number of sorting steps required by our SQL implementation depends
on whether the DBMS is capable of sharing window declaration (we observe 2 and 7 sorting steps for the systems used in our experimental evaluation).}

\reva{For aggregation we integrate the split operator into the
aggregation. It turned out to be most effective to pre-aggregate the input
before splitting and then compute the final aggregation results during the
split step by further aggregating the results of the pre-aggregation step.} We apply a similar optimization for difference.

\section{Experiments}
\label{sec:experiments}

In our experimental evaluation we focus on two aspects. First, we
evaluate the cost of our SQL implementation of $\semN$-coalescing
(multiset coalescing).  Then, we evaluate the performance of snapshot
queries with our approach over three DBMSs and compare it against
native implementations of snapshot semantics that are available in
two of these systems (using our implementation of coalescing to
produce a coalesced result).

\subsection{Workloads and Experimental Setup}
\label{sec:data-sets}

\parttitle{Datasets} We use \reva{\iftechreport{three}\ifnottechreport{two} datasets in our experiments}. The \textit{MySQL Employees dataset}
(\url{https://github.com/datacharmer/test_db})
 which contains
$\approx$4~million records and consists of the following six
period tables: table \texttt{employee} stores basic information about
employees; table \texttt{departments} stores department information;
table \texttt{titles} stores the job titles for employees; table
\texttt{salaries} stores employee salaries; table
\texttt{dept\_manager} stores which employee manages which
department; and table \texttt{dept\_emp} stores which employee is
working for which department.
\reva{\textit{TPC-BiH} is the bi-temporal version of the TPC-H benchmark dataset as described in~\cite{DBLP:conf/tpctc/KaufmannFMTK13}. Since our approach supports only one time dimension we only generated the valid time dimension for this dataset. In this configuration a scale factor 1 (SF1) database corresponds to roughly 1GB of data.}
\ifnottechreport{
\reva{In our technical report we also use a real-world \textit{Tourism} dataset (835k records).}
}
\iftechreport{
The \textit{Tourism} dataset (835k records) consists of a single table storing hotel reservations in South Tyrol. Each record corresponds to one reservation. The validity end points of the time period associated with a record is the arrival and departure time.
}

\parttitle{Workloads}
\ifnottechreport{To evaluate the efficiency of snapshot queries, we
  created a workload consisting of the following 10 queries expressed over the employee dataset. \texttt{join-1}:
  salary and department for each employee using a join
  between the department and salary tables. \texttt{join-2}:
  salary and title for each employee using a join between the
  salary and title tables.  \texttt{join-3}:
  department of employees that manage a department and earn more than
  \$70,000 using a join between the manager and salary tables and a
  selection on attribute salary.\texttt{join-4}: full information of each manager using joins
  between the tables managers, salary, and employees.
  \texttt{agg-1}: average salary of employees per
  department using the result of query \texttt{join-1} with a
  subsequent aggregation per department. \texttt{agg-2}:
  average salary of managers using a join between the
  manager and salary tables with a subsequent aggregation without
  grouping. \texttt{agg-3}: number of departments
  with more than 21 employees. This query has no join but two
  aggregations, one to compute the number of employees per department
  and a second one to count the number of relevant departments.
  \texttt{agg-join}: names of employees with the highest
  salary in their department.  It consists of a 4-way join, where one
  of the inputs is the result of a subquery with aggregation.
  \texttt{diff-1}: employees that are not managers using a
  difference operation between two tables.  \texttt{diff-2}:
  salaries of employees that are not managers by computing
  the difference between a table and a subquery with a join.
  \reva{
    For the TPC-BiH dataset we took 9 of the 22 standard queries~\cite{tpc-h} from this benchmark that do not contain nested subqueries or \lstinline!LIMIT! (which are not supported by our or any other approach for snapshot queries we are aware of) and evaluated these queries under snapshot semantics. Note that some of these queries use the \lstinline!ORDER BY! clause that we do not support for snapshot queries. However, we can evaluate such a query without \lstinline!ORDER BY! under snapshot semantics and then sort the result without affecting what rows are returned.
  }
  The number of rows returned by the \reva{Employee and TPC-H}  queries are
  shown in Table~\ref{tab:query-result}. The SQL code and more detailed descriptions of \reva{these} queries are provided in~\cite[Appendix B]{DG18}.
}
\iftechreport{
  We have created a workload consisting of 10 queries to evaluate the
  efficiency of snapshot queries. Queries \texttt{join-1} to \texttt{join-4}
  are join queries, \texttt{agg-1} to \texttt{agg-3} are aggregation-heavy
  queries, \texttt{agg-join} is a join with an aggregation value, and
  \texttt{diff-1} and \texttt{diff-2} use difference. Furthermore, we use one
  query template varying the selectivity to evaluate the performance of
  coalescing. \texttt{C-Sn} denotes the variant of this query that returns
  approximately $nK$ rows, e.g., \texttt{C-S1} returns 1,000 rows.
For the Tourism dataset we use the following queries. \texttt{join}: tourist from same country to same destination using a self join of tourismdata table.  \texttt{agg-0}: number of tourists per destination together with the average number of tourists for all other destinations. This query first computes the number of tourists per destination and do a self unequal join on it. \texttt{agg-1}: number of enquiries and the number of tourists per destination with more than 1000 enquiries using two aggregations on tourismdata table. \texttt{agg-2}: maximum number of tourists per destination using an aggregation on tourismdata table. \texttt{tou-agg-x}: the destination with the most number of tourists. This query has no join but two aggregations, one to compute the number of tourists per destination and a second one to compute the maximum one.
  More
  detailed descriptions of these queries are provided in
  Appendix~\ref{sec:workload-queries}.
    For the TPC-BiH dataset we took 9 of the 22 standard queries~\cite{tpc-h} from this benchmark that do not contain nested subqueries or \lstinline!LIMIT! (which are not supported by our or any other approach for snapshot queries we are aware of) and evaluated these queries under snapshot semantics. Note that some of these queries use the \lstinline!ORDER BY! clause that we do not support for snapshot queries. However, we can evaluate such a query without \lstinline!ORDER BY! under snapshot semantics and then sort the result without affecting what rows are returned.
  The number of rows returned by these
  queries over the  dataset are shown in Table~\ref{tab:query-result}.
}

\begin{table}[thb]
  \caption{Number of query result rows}
  \label{tab:query-result}
\centering
\resizebox{1\columnwidth}{!}{
\setlength{\tabcolsep}{3pt}
\begin{minipage}{1.4\linewidth}
\centering
\begin{tabular}[t]{|r|r|r|r||r|r|r||r||r|r|}
 \hline
 \thead{join-1} &   \thead{join-2} &    \thead{join-3}  & \thead{join-4} & \thead{agg-1}  & \thead{agg-2}  &
\thead{agg-3}  &  \thead{agg-join} &   \thead{diff-1} &  \thead{diff-2}
\\ \hline
  2.8M & 28.3M  & 10 &  177 &  57.4k & 177 & 210 & 260 & 300k  &  2.8M
  \\ \hline
\end{tabular}\\[2mm]
\reva{
\begin{tabular}[t]{|c|c|c|c|c|c|c|c|c|c|c|c|}
  \hline
 \thead{TPC-H}& \thead{Q1} & \thead{Q3} & \thead{Q5} & \thead{Q6} & \thead{Q7} & \thead{Q8} & \thead{Q9} & \thead{Q10} & \thead{Q12} & \thead{Q14} & \thead{Q19} \\ \hline
 \thead{1GB} & 4.3k & 10 & 386 & 529 & 1.6k & 742 & 69.7k & 20 & 785 & 479 & 220 \\ \hline
 \thead{10GB} & 4.3k & 10 & 579 & 532 & 1.7k & 867 & 74.8k & 20 & 786 & 487 & 1.3k \\ \hline
\end{tabular}\\[2mm]
}
\iftechreport{
\reva{
\begin{tabular}[t]{|r||r|r|r||r|}
 \hline
 \thead{tou-join-agg} &   \thead{tou-agg-1} &   \thead{tou-agg-2} & \thead{tou-agg-3}  & \thead{tou-agg-join}
\\ \hline
 64.3k & 954  & 14.5k &  3.2k &  822
  \\ \hline
\end{tabular}
}
}
\end{minipage}
}
\end{table}

\parttitle{Systems}
We ran experiments on three different database management systems: a
version of Postgres (\textit{PG}) with native support for temporal
operators as described in~\cite{DignosBG12,DignosBGJ16}; a commercial
DBMS, \textit{DBX}, with native support for snapshot semantics (only
available as a virtual machine); and a commercial DBMS, \textit{DBY},
without native support for snapshot semantics.
We used our approach to translate snapshot queries into standard SQL
queries and ran the translated queries on all three systems (denoted
as \textit{PG-Seq}, \textit{DBX-Seq}, and \textit{DBY-Seq}).
For PG and DBX, we ran the queries also with the native
solution for snapshot semantics paired with our implementation of
coalescing to produce a coalesced result (referred to as
\textit{PG-Nat} and \textit{DBX-Nat}).
As explained in Section~\ref{sec:related-work}, no system correctly
implements snapshot multiset semantics for difference and
aggregation, and many systems do not support snapshot semantics for
these operators at all. \textit{DBX-Nat} and \textit{PG-Nat}
both support snapshot aggregation, however, their implementations are
not snapshot-reducible. \textit{DBX-Nat} does not support snapshot difference,
whereas \textit{PG-Nat} implements temporal difference with set
semantics.  Despite such differences, the experimental comparison
allows us to understand the performance impact of our provably correct
approach.

All experiments were executed on a machine with 2 AMD Opteron 4238 CPUs, 128GB
RAM, and a hardware RAID with 4 $\times$ 1TB 72.K HDs in RAID 5. \revm{For Postgres we set the buffer pool size to $8$GB. For the other systems we use values recommended by the automated configuration tools of these systems. We execute
queries with warm cache. For short-running queries we show the median runtime
across 100 consecutive runs. For long running queries we computed the
median over 10 runs. In general we observed low variation in runtimes (a few
percent).}

\subsection{Multiset Coalescing}
\label{sssec:multiset-coalescing}

To evaluate the performance of coalescing, we use a selection query
that returns employees that earn more than a specific salary and
materialize the result as a table.  The selectivity varies from 1K to
3M rows.  We then evaluate the query \texttt{\textbf{SELECT} *
  \textbf{FROM} ...}  over the materialized tables under snapshot
semantics in order to measure the cost of coalescing in isolation.
Figure~\ref{fig:coalesce-selectivity} shows the results of this
experiment.  The runtime of coalescing is linear in the input size for
all three systems.  Even though the theoretical worst-case complexity
of the sorting step, which is applied by all systems to evaluate the
analytics functions that we exploit in our SQL-based implementation of
multiset coalescing, is $\mathcal{O}(n \cdot log (n))$, an inspection
of the execution plans revealed that the sorting step only amounts to
5\%-10\% of the execution time (for all selectivities) and, hence, is
not a dominating factor.

\begin{figure}[t]
  \centering
  \includegraphics[width=0.8\linewidth,trim=10pt 30pt 0 130pt, clip]{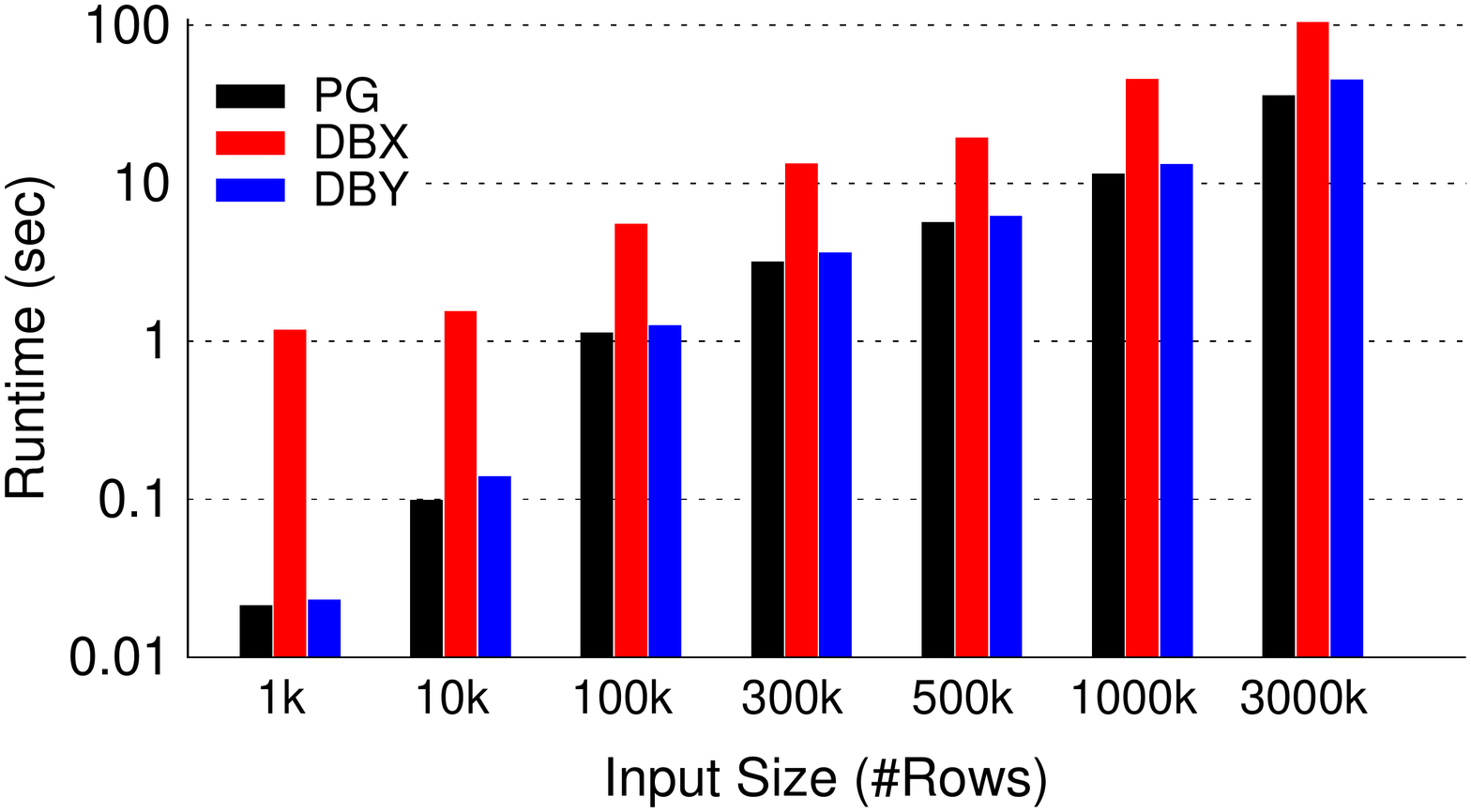}
  \caption{Multiset coalescing for varying input size.}$\,$\\[1mm]
  \label{fig:coalesce-selectivity}
\end{figure}

\subsection{Snapshot Semantics  - Employee}
\label{sssec:sequ-semant-quer}

Table~\ref{tab:runtime-seq-queries} provides an overview of the
performance results for our snapshot query workloads. \revm{For every query we indicate in the rightmost column whether native approaches are subject to the aggregation gap (AG) or bag difference (BD) bugs.}

\begin{table}[t]
  \caption{Runtimes (sec) of snapshot queries: \ONA\, = not supported, \OOTS\, = system ran out of temporary space (2GB), \TimeOut  = timed out (2 hours).}
  \label{tab:runtime-seq-queries}
  \centering
  \scriptsize
  \setlength{\tabcolsep}{3pt}
  \begin{tabular}{|l|rr|rr|r|c|}\hline
      \multicolumn{7}{|c|}{\textbf{Employee dataset}}\\
    \hline
   \thead{Query} & \thead{PG-Seq} & \thead{PG-Nat} & \thead{DBX-Seq} & \thead{DBX-Nat} & \thead{DBY-Seq} & \thead{Bug}\\ \hline
   \texttt{join-1} & 91.97 & 118.01 & 118.95 & 116.03 & 64.00 & \\
    \texttt{join-2} & 1543.81 & 888.13 & 1569.45 & 1200.36 & 763.70 &\\
   \texttt{join-3} & 0.01 & 4.91 & 0.55 & 0.43 & 0.01 & \\
   \texttt{join-4} & 0.52 & 12.85 & 0.83 & 0.60 & 0.22 & \\
   \hline
   \texttt{agg-1} & 7.02 & 5980.85 & 56.47 & \OOTS & 5.24 & \\
   \texttt{agg-2} & 0.06 & 10.31 & 0.82 & 0.82 & 0.01 & AG\\
   \texttt{agg-3} & 1.42 & 0.02 & 0.78 & 0.55 & 0.01 & AG\\
   \hline
   \texttt{agg-join} & 6643.61 & 19195.03 & \OOTS &  \OOTS & 7555.97 & \\
   \hline
   \texttt{diff-1} & 14.18 & 6.88 & 30.15 & \ONA & 10.29 &BD\\
   \texttt{diff-2} & 63.58 & 79.63 & 129.87 & \ONA & 61.90 &BD\\
    \hline \end{tabular}\\[3mm] \iftechreport{
  \reva{
  \begin{tabular}{|l|rr|r|c|}\hline
    \multicolumn{5}{|c|}{\textbf{Tourism}} \\ \hline
    \thead{Query} & \thead{PG-Seq} &  \thead{PG-Nat} &   \thead{DBY-Seq}  &  \thead{Bug}\\
    \hline
  \texttt{tou-join-agg} & 300.28 & 694.88 & 171.09 &  \\
  \hline
        \texttt{tou-agg-1} & 2.41 & 94.58 & 1.61 &  \\
        \texttt{tou-agg-2} & 123.79 & 92.32 & 87.31 &  \\
      \texttt{tou-agg-3} & 6.68 & 98.07 & 7.66 & AG \\
    \hline
        \texttt{tou-agg-join} & 1.06 & 263.61 & 0.94 &  \\
    \hline
  \end{tabular}\\[3mm]
}
}
   \reva{
     \begin{tabular}{|l|rrr|rrr|c|}\hline
     \multicolumn{8}{|c|}{\textbf{TPC-BiH}}\\ \hline
     &\multicolumn{3}{c|}{\textbf{SF1 ($\sim$1 GB)}} & \multicolumn{3}{c|}{\textbf{SF10 ($\sim$10 GB)}} & \\ \cline{2-7}
     \thead{Query} & \thead{PG-Seq} & \thead{PG-Nat} & \thead{DBY-Seq} & \thead{PG-Seq} & \thead{PG-Nat} & \thead{DBY-Seq} & \thead{Bug}\\
     \hline
\texttt{Q1} & 12.02 & 3686.47 &11.80 & 63.85 & \TimeOut & 82.61&\\
\texttt{Q5} & 0.58 & 142.91 &1.14 &    5.85 & 1794.10 & 14.89&\\
\texttt{Q6} & 0.79 & 12.65 &1.14  &     7.70 & 126.91 & 7.28& AG\\
\texttt{Q7} & 1.14 & 285.91 &5.33&     28.70 & 1642.20 & 21.75&\\
\texttt{Q8} & 1.77 & 108.63 &2.20&     21.78 & 1484.61 & 17.33&\\
\texttt{Q9} & 10.12 & \TimeOut &8.09&  129.01 & \TimeOut&  71.37&\\
\texttt{Q12} & 1.10 & 23.85 &1.81&      10.49 & 264.57 & 13.30&\\
\texttt{Q14} & 1.72 & 403.92 &2.75&     26.55 & 3436.30&  23.79&AG\\
\texttt{Q19} & 0.92 & 203.83 &2.55&     9.60 & 2873.13 & 22.35&AG\\
     \hline
   \end{tabular}\\[3mm]
   }
\end{table}

\parttitle{Join Queries}
The performance of our approach for join queries is comparable with
the native implementation in \textit{PG-Nat}. For join queries with larger
intermediate results (\texttt{join-2}), the native implementation
outperforms our approach by $\approx$73\%. Running the queries produced
by our approach in \textit{DBY} is slightly faster than both.
\textit{DBX-Nat} uses merge joins for temporal joins, while both
\textit{PG} and \textit{DBY} use a hash-join on the
non-temporal part of the join condition. The result is that
\textit{DBX-Nat} significantly outperforms the other methods for
temporal join operations. However, the larger cost for the SQL-based
coalescing implementation in this system often outweighs this
effect.  This demonstrates the potential for improving our approach by making
use of native implementations of temporal operators in our rewrites
for operators that are compatible with our semantics (note that joins
are compatible).

\parttitle{Aggregation Queries}
Our approach outperforms the native implementations of snapshot semantics on all systems by several orders of magnitude for aggregation queries as long as the aggregation input exceeds a certain size (\texttt{agg-1} and \texttt{agg-2}).
Our approach as well as the native approaches split the aggregation input which requires sorting and then apply a standard aggregation operator to compute the temporal aggregation result.  The main reason for the large performance difference is that the SQL code we generate for a snapshot aggregation includes several levels of pre-aggregation that are intertwined with the split operator. Thus, for our approach the sorting step for split is applied to a typically much smaller pre-aggregated dataset.
This turned out to be quite effective. The only exception is if the aggregation input is very small (\texttt{agg-3}) in which case an efficient implementation of split (as in \textit{PG-Nat}) outweighs the benefits of pre-aggregation.
\revm{Query \texttt{agg-1} did not finish on \textit{DBX-Nat} as it exceeded
the $2$GB temporary space restriction (memory allocated for intermediate results) of the freely available version of this DBMS.}

\parttitle{Mixed Aggregation and Join}
Query \texttt{agg-join} applies an aggregation over the result of several joins. Our
approach is more effective, in particular for the aggregation part of this
query, compared to \textit{PG-Nat}.
This query did not finish on \textit{DBX}  due to the $2$GB temporary space restriction per query imposed by the DBMS.

\parttitle{Difference Queries}
For difference queries we could only compare our approach against \textit{PG-Nat}, since \textit{DBX-Nat} does not support difference in snapshot queries. Note that, \textit{PG-Nat} applies set difference while our approach supports multiset difference. While our approach is less effective for \texttt{diff-1} which contains a single difference operator, we outperform \textit{PG-Nat} on \texttt{diff-2}.

\reva{
\subsection{Snapshot Semantics  - TPC-BiH}
\label{sec:sequ-semant-quer}

The runtimes for TPC-H queries interpreted under snapshot semantics (9 queries are currently supported by the approaches) over the 1GB and 10GB valid time versions of TPC-BiH is also shown in Table~\ref{tab:runtime-seq-queries}. For this experiment we skip \textit{DBX} since the limitation to 2GB of temporary space of the free version we were using made it impossible to run most of these queries. Overall we observe that our approach scales roughly linearly from 1GB to 10GB for these queries. We significantly outperform \textit{PG-Nat} because all of these queries use aggregation. Additionally, some of these queries use up to 7 joins. For these queries the fact that \textit{PG-Nat} aligns both inputs with respect to each other~\cite{DignosBG12} introduces unnecessary overhead and limits join reordering. The combined effect of these two drawbacks is quite severe. Our approach is 1 to 3 orders of magnitude faster than
\textit{PG-Nat}. For some queries this is a lower bound on the overhead of \textit{PG-Nat} since the system timed out for these queries (we stopped queries that did not finish within 2 hours).
}

\iftechreport{
\reva{
\subsection{Snapshot Semantics - Tourism}
\label{sec:sequ-semant-quer-1}

The results for the queries over the Tourism database are shown in the middle of Table~\ref{tab:runtime-seq-queries}. We only report our approach for Postgres and DBY, and the native implementation in Postgres. With the exception of query \textit{tou-agg-2} our approach outperforms \textit{PG-Nat} quite significantly since all these queries contain aggregation. Since query \textit{tou-agg-2} does use \lstinline!max! we do not apply our sweeping technique (see Appendix~\ref{sec:combining-split-with}). \textit{PG-Nat}'s native implementation of the split operator results in 30\% better performance for this query.
Query \textit{tou-join-agg} applies an inequality self-join over an aggregation result ($\approx$ 100k rows under snapshot semantics) and then applies a final aggregation to the join. The large size of this join result is the main reason
}
}

\subsection{Summary}
\label{sec:summary}

Our experiments demonstrate that an SQL-based implementation of multiset coalescing is feasible -- exhibiting runtimes linear in the size of the input, albeit with a relatively large constant factor. We expect that it would be possible to significantly reduce this factor by introducing a native implementation of this operator. Using pre-aggregation during splitting, our approach significantly outperforms native implementations for aggregation queries. DBX uses merge joins for temporal joins (interval overlap joins) which is significantly more efficient than hash joins which are employed by Postgres and DBY. This shows the potential of integrating such specialized operators with our approach in the future. For example, we could compile snapshot queries into SQL queries that selectively employ the temporal extensions of a system like DBX.

\section{Conclusions and Future Work}
\label{sec:concl-future-work}

We present the first provably correct interval-based representation system for snapshot semantics over multiset relations and its implementation in a database middleware.
We achieve this goal by addressing a more general problem: snapshot-reducibility for temporal $\semK$-relations.
Our solution is a uniform framework for evaluation of queries under snapshot semantics over an interval-based encoding of temporal $\semK$-relations for any semiring $\semK$.
That is, in addition to sets and multisets, the framework supports snapshot temporal extensions of probabilistic databases, databases annotated with provenance, and many more.
In future work, we will study how to extend our approach for updates over annotated relations, will study its applicability for combining probabilistic and temporal query processing, investigate implementations of split and $\semK$-coalescing inside a database kernel, \revm{and study extensions for bi-temporal data}.

\ifnottechreport{\clearpage}

\appendix
\section{Proofs}
\label{sec:proofs}

\begin{proof}[Proof of Lemma \ref{lem:coalesce-properties}]
\myproofpar{Equivalence preservation}
Proven by contradiction. Assume that $\exists \tPoint: \tSlice{\tPoint}(\anyTE) \neq \tSlice{\tPoint}(\kCoalesce{\semK}(\anyTE))$. We have to distinguish two cases. If $\tPoint \in \CPs{\anyTE}$, then by definition of $\semK$-coalesce we have the contradiction: $\tSlice{\tPoint}(\kCoalesce{\semK}(\anyTE)) = \tSlice{\tPoint}(\anyTE)$. If $\tPoint \not\in \CPs{\anyTE}$, then let $\tPoint'$ be the largest change point that is smaller than $\tPoint$ (this point has to exist). From the definition of change points follows that $\tSlice{\tPoint} = \tSlice{\tPoint'}$. By construction there has to exists exactly one interval overlapping $\tPoint$ that is assigned a non-zero value in $\kCoalesce{\semK}(\anyTE)$ and this interval starts in $\tPoint'$. Hence, we have the contradiction.

\myproofpar{Uniqueness}
Note that change points are defined using $\tSlice{\tPoint}$ only and $(\forall \tPoint \in \timeDomain: \tSlice{\tPoint}(\anyTE_1) = \tSlice{\tPoint}(\anyTE_2)) \Leftrightarrow \anyTE_1 \intervalEq \anyTE_2$. Since the result of coalescing is uniquely determined by the change points of a temporal element it follows that $ \anyTE_1 \intervalEq \anyTE_2 \Leftrightarrow \kCoalesce{\semK}(\anyTE_1) = \kCoalesce{\semK}(\anyTE_2)$

\myproofpar{Idempotence}
Idempotence follows from the other two properties. If we substitute $\anyTE$ and $\kCoalesce{\semK}(\anyTE)$ for $\anyTE_1$ and $\anyTE_2$ in the uniqueness condition, we get idempotence: $\kCoalesce{\semK}(\kCoalesce{\semK}(\anyTE)) = \kCoalesce{\semK}(\anyTE)$.
\end{proof}

\begin{proof}[Proof of Lemma \ref{lem:coalesce-push}]
\myproofpar{Push Through Addition}
We prove this part by proving that for any $k''$ if $k \intervalEq k'$ then $(k \addP k'') \intervalEq (k' \addP k'')$.
We have to show that for all $\tPoint \in \timeDomain$ we have $\tSlice{\tPoint}(k + k'') = \tSlice{\tPoint}(k' + k'')$. Substituting definitions we get:
\begin{align*}
  \sum_{\tPoint \in \interval} (k(\interval) +_\semK k''(\interval)) =   (\sum_{\tPoint \in \interval}k(\interval)) +_\semK (\sum_{\tPoint \in \interval}k''(\interval))
\end{align*}
Using $k \intervalEq k'$ and substituting the definition of $\intervalEq$, i.e., $\sum_{\tPoint \in \interval}k(\interval) = \sum_{\tPoint \in \interval}k'(\interval)$, we get:
\begin{align*}
  = (\sum_{\tPoint \in \interval}k'(\interval)) +_\semK (\sum_{\tPoint \in \interval}k''(\interval))
  =   \sum_{\tPoint \in \interval} (k'(\interval) +_\semK k''(\interval))
\end{align*}

\myproofpar{Push Through Multiplication}
Analog to the proof for addition, we prove this part by showing that snapshot equivalence of inputs implies snapshot equivalence of outputs for multiplication.
\begin{align*}
\tSlice{\tPoint}(k \multP k'') =
 \sum_{\forall \interval', \interval'': \interval = \interval' \cap \interval'' \wedge \tPoint \in \interval} k(\interval') \cdot_\semK k''(\interval'')
\end{align*}
Based on the fact that timeslice is a homomorphism $\semTimeNI \to \semK$ which we will prove in Theorem~\ref{theo:hib-is-homomorphism}, time slice commutes with multiplication and addition:
\begin{align*}
  = &(\sum_{\forall \interval \wedge \tPoint \in \interval} k(\interval)) \cdot_\semK (\sum_{\forall \interval \wedge \tPoint \in \interval} k''(\interval))\\
  = &(\sum_{\forall \interval \wedge \tPoint \in \interval} k'(\interval)) \cdot_\semK (\sum_{\forall \interval \wedge \tPoint \in \interval} k''(\interval))\\
=  &\tSlice{\tPoint}(k' \multP k'') \tag*{\qedhere}
\end{align*}
\end{proof}

\begin{proof}[Proof of Theorem \ref{theo:timeib-is-semiring}]
  We have to show that the structure we have defined obeys the laws of commutative semirings.   Since the elements of $\semTimeNI$ are functions, it suffices to show $k(\interval) = k'(\interval)$ for every $\interval \in \intervalDom$ to prove that $k = k'$. For all $k, k' \in \semTimeNI$ and $\interval \in \intervalDom$:

\myproofpar{Addition is commutative}\\[-4mm]
  \begin{gather*}
    (k \addP k')(\interval) = k(\interval) +_\semK k'(\interval) = k'(\interval) +_\semK k(\interval) = (k \addP k')(\interval)\\
    k \addNI k' = \kCoalesce{\semK}(k \addP k')
= \kCoalesce{\semK}(k' \addP k) =   k' \addNI k
\end{gather*}
\myproofpar{Addition is associative}\\[-4mm]
\begin{gather*}
  ((k \addP k') \addP k'')(\interval) = (k(\interval) +_\semK k'(\interval)) +_\semK k''(\interval)\\ = k(\interval) +_\semK (k'(\interval) +_\semK k''(\interval)) =   (k \addP (k' \addP k''))(\interval)\\[3mm]
      (k \addNI k') = \kCoalesce{\semK}(k \addP k')
= \kCoalesce{\semK}(k' \addP k) =   (k' \addNI k)
\end{gather*}
\myproofpar{Zero is neutral element of addition}\\[-4mm]
\begin{gather*}
  (k \addP \zeroNI)(\interval) = k(\interval) +_\semK \zeroNI(\interval) = k(\interval) +_\semK 0_\semK = k(\interval) \\
  k  \addNI \zeroNI = \kCoalesce{\semK}(k \addP \zeroNI) = \kCoalesce{\semK}(k) = k
\end{gather*}
\myproofpar{Multiplication is commutative}\\[-4mm]
  \begin{gather*}
  (k \multP k')(\interval) =
 \sum_{\forall \interval', \interval'': \interval = \interval' \cap \interval''} k(\interval') \cdot_\semK k'(\interval'') \\
 =  \sum_{\forall \interval'', \interval': \interval = \interval'' \cap \interval'} k(\interval'') \cdot_\semK k'(\interval')\\
 =  \sum_{\forall \interval', \interval'': \interval = \interval' \cap \interval''} k'(\interval') \cdot_\semK k(\interval'')
 = (k' \multP k)(\interval)\\
  k \multNI k' =
\kCoalesce{\semK}(k \multP k')
= \kCoalesce{\semK}(k' \multP k) =   k' \multNI k
\end{gather*}
\myproofpar{Multiplication is associative}
  \begin{gather*}
 ((k \multP k') \multP k'')(\interval)\\
= \sum_{\forall \interval_1, \interval_2: \interval = \interval_1 \cap \interval_2} ( \sum_{\forall \interval_3, \interval_4: \interval_1 = \interval_3 \cap \interval_4} k(\interval_3) \cdot_\semK k'(\interval_4)) \cdot_\semK k''(\interval_2)\\
 = \sum_{\forall \interval_1, \interval_2: \interval = \interval_1 \cap \interval_2} \sum_{\forall \interval_3, \interval_4: \interval_1 = \interval_3 \cap \interval_4} (k(\interval_3) \cdot_\semK k'(\interval_4) \cdot_\semK k''(\interval_2))\\
 = \sum_{\forall \interval_1, \interval_2, \interval_3: \interval = \interval_1 \cap \interval_2 \cap \interval_3} k(\interval_1) \cdot_\semK k'(\interval_2) \cdot_\semK k''(\interval_3)\\
 = \sum_{\forall \interval_1, \interval_2: \interval = \interval_1 \cap \interval_2}  k(\interval_1) \cdot (\sum_{\forall \interval_3, \interval_4: \interval_2 = \interval_3 \cap \interval_4} k'(\interval_3) \cdot_\semK k''(\interval_4))\\
 = (k \multP (k' \multP k''))(\interval)\\[3mm]
 (k \multNI k') \multNI k'\\
 = \kCoalesce{\semK}(\kCoalesce{\semK}(k \multP k') \multP k'')\\
 = \kCoalesce{\semK}(k \multP k' \multP k'')\\
 = \kCoalesce{\semK}(k \multP \kCoalesce{\semK}(k' \multP k''))\\
= k \multNI (k' \multNI k'')
 \end{gather*}
\myproofpar{One is neutral element of multiplication}
 \begin{gather*}
 (k  \multP \oneNI)(\interval) =  \sum_{\forall \interval', \interval'': \interval = \interval' \cap \interval''}  k(\interval') \cdot_\semK \oneNI(\interval'')\\
 =  k(\interval) \multP \oneNI([t_{min},t_{max})) = k(\interval) \cdot_\semK 1_\semK \\
 = k(\interval) \\
 k  \multNI 1_{\semTimeNI} \\
 = \kCoalesce{\semK}(k  \multP 1_{\semTimeNI})
= \kCoalesce{\semK}(k) = k
 \end{gather*}
\myproofpar{Distributivity}
 \begin{gather*}
 (k  \multP (k' \addP k''))(\interval)\\
 = \sum_{\forall \interval', \interval'': \interval = \interval' \cap \interval''} k(\interval') \cdot_\semK (k'(\interval'') +_\semK k''(\interval''))\\
  = \sum_{\forall \interval', \interval'': \interval = \interval' \cap \interval''} (k(\interval') \cdot_\semK k'(\interval'')) +_\semK (k(\interval') \cdot_\semK k''(\interval''))\\
  = \sum_{\forall \interval', \interval'': \interval = \interval' \cap \interval''} (k(\interval') \cdot_\semK k'(\interval'')) \\ \hspace{2cm}+ \sum_{\forall \interval', \interval'': \interval = \interval' \cap \interval''} (k(\interval') \cdot_\semK k''(\interval''))\\
  =   ((k  \multP k')  \addP (k  \multP k''))(\interval)\\[3mm]
   k  \multNI (k' \addNI k'')\\
 = \kCoalesce{\semK}(k  \multP \kCoalesce{\semK}(k' \addP k''))\\
 = \kCoalesce{\semK}(k  \multP (k' \addP k''))\\
 = \kCoalesce{\semK}((k  \multP k') \addP (k  \multP k'')))\\
 = \kCoalesce{\semK}(\kCoalesce{\semK}(k  \multP k') \addP \kCoalesce{\semK}(k  \multP k'')))\\
  =   (k  \multNI k')  \addNI (k  \multNI k'') \tag*{\qedhere}
\end{gather*}
\end{proof}

\begin{proof}[Proof of Theorem \ref{theo:hib-is-homomorphism}]
  Proven by substitution of definitions:

\myproofpar{Preserves neutral elements}
 $$\kCoalesce{\semK}(\tSlice{\tPoint}(\zeroNI)) =  \sum_{\interval \in \intervalDom: \tPoint \in \interval} \zeroNI(\interval) = \sum_{\interval \in \intervalDom: \tPoint \in \interval} 0_\semK =  0_\semK$$\\
 $$\tSlice{\tPoint}(\oneNI) = \sum_{\interval \in \intervalDom: \tPoint \in \interval} \oneNI(\tPoint)$$
 Since $\tPoint \in [\tMin,\tMax)$
 for any $\tPoint \in \timeDomain$ and $\oneNI(\interval) = 0_\semK$ for any interval $\interval$ except for $[\tMin,\tMax)$ where $\oneNI([\tMin,\tMax)) = 1_\semK$ we get
 $\sum_{\interval \in \intervalDom: \tPoint \in \interval} \oneNI(\tPoint) = 1_\semK$\\
\myproofpar{Commutes with addition}
 $$\tSlice{\tPoint}(k +_{\anyTE} k') = \sum_{\interval \in \intervalDom: \tPoint \in \interval} (k +_{\anyTE} k')(\interval)
 = \sum_{\interval \in \intervalDom: \tPoint \in \interval} k(\interval) +_\semK k'(\interval)$$ $$=
 \sum_{\interval \in \intervalDom: \tPoint \in \interval} k(\interval) +_\semK \sum_{\interval \in \intervalDom: \tPoint \in \interval} k'(\interval)
= \tSlice{\tPoint}(k) + \tSlice{\tPoint}(k')$$\\
\myproofpar{Commutes with multiplication}
\begin{align*}
  &\tSlice{\tPoint}(k \cdot_{\anyTE} k') = \sum_{\interval \in \intervalDom: \tPoint \in \interval} (k \cdot_{\anyTE} k')(\interval) \\
  = &\sum_{\interval \in \intervalDom: \tPoint \in \interval} \sum_{\forall \interval', \interval'': \interval = \interval' \cap \interval''} k(\interval') \cdot_\semK k'(\interval'')\\
= &\sum_{\forall \interval', \interval'': \tPoint \in \interval' \wedge \tPoint \in \interval''} k(\interval') \cdot_\semK k'(\interval'')
\end{align*}

 Let $n_1, \ldots, n_l$ denote the elements $k(\interval)$ for all intervals from the set of intervals with $\tPoint \in \interval$ and $k(\interval) \neq 0$. Analog, let $m_1, \ldots, m_o$ bet the set of elements with the same property for $k'$. Then the sum can be rewritten as:
 $$ = \sum_{i=1}^{l} \sum_{j=1}^{o} n_i \cdot_\semK m_j = \sum_{i=1}^{l}  n_i \cdot_\semK (\sum_{j=1}^{o} m_j) = (\sum_{i=1}^{l}  n_i) \cdot_\semK (\sum_{j=1}^{o} m_j)$$
 replacing this again with the interval notation we get:
 \begin{gather*}
 = (\sum_{\forall \interval: \tPoint \in \interval} k(\interval)) \cdot_\semK (\sum_{\forall \interval: \tPoint \in \interval}k'(\interval))
   = \tSlice{\tPoint}(k) \cdot_\semK \tSlice{\tPoint}(k')
   \tag*{\qedhere}
 \end{gather*}

\end{proof}

\begin{proof}[Proof of Lemma \ref{lem:encoding-is-bijective}]
\myproofpar{injective}
We have to show that for any two snapshot $\semK$-relations $\rel$ and $\rel'$, $\repr_{\semK}(\rel) = \repr_{\semK}(\rel') \Rightarrow \rel = \rel'$. Since, $\kCoalesce{\semK}$ preserves snapshot equivalence and is a unique representation of any temporal $\semK$-element $\anyTE$, it is sufficient to show that for all $t$, we have $\anyTE_{R,\tuple} = \anyTE_{R',\tuple}$ instead. For sake of contradiction, assume that there exists a tuple $\tuple$ such that $\anyTE_{R,\tuple} \neq \anyTE_{R',\tuple}$. Then there has to exist  $\tPoint \in \timeDomain$ such that $\anyTE_{R,\tuple}([\tPoint,\tPoint+1)) \neq \anyTE_{R',\tuple}([\tPoint,\tPoint+1))$. However, based on the definition of $\anyTE_{R,\tuple}$ this implies that $\rel(\tPoint)(\tuple) \neq \rel'(\tPoint)(\tuple)$ which contradicts the assumption.

\myproofpar{surjective} Given a $\semTimeNI$-relation $\rel$, we construct a snapshot $\semK$-relation $\rel'$ such that $\repr_{\semK}(\rel') = \rel$: $\rel'(\tPoint)(\tuple) = \sum_{\tPoint \in \interval}\rel(\tuple)(\interval)$.
\end{proof}

\begin{proof}[Proof of Lemma \ref{lem:encoding-timeslice-is-compatible}]
  By virtue of snapshot equivalence between  $\kCoalesce{\semK}(\anyTE)$ and $\anyTE$ and based on the singleton interval definition of $\anyTE_{R,\tuple}$ in $\repr_{\semK}$, we have for any tuple $\tuple$:
  \begin{gather*}
  \tSlice{\tPoint}(\repr_{\semK}(\rel))(\tuple) = \tSlice{\tPoint}(\anyTE_{R,\tuple}) =
\anyTE_{R,\tuple}([\tPoint,\tPoint+1)) =
  \rel(\tPoint)(\tuple)\\ = \tSlice{\tPoint}(\rel)(\tuple)
    \tag*{\qedhere}
  \end{gather*}
\end{proof}

\begin{proof}[Proof of Theorem \ref{theo:repr-system}]
We have to show that $(\domTimeRel{\semK}, \reprInv{\semK}, \tSlice{})$ fulfills conditions (1), (2), and (3)  of Definition~\ref{def:repr-system} to prove that this triple is a representation system for $\semK$-relations. Conditions (1) and (2) have been proven in Lemmas~\ref{lem:encoding-is-bijective} and~\ref{lem:encoding-timeslice-is-compatible}, respectively. Condition (3) follows from the fact that $\tSlice{\tPoint}$ is a homomorphism (Theorem~\ref{theo:hib-is-homomorphism}) and that semiring homomorphisms commute with $\raPlus$-queries (\cite{GK07}, Proposition 3.5).
\end{proof}

\begin{proof}[Proof of Theorem \ref{theo:bag-interval-norm-monus}]
To prove that $\semTimeNI$ has a well-defined monus, we have to show $\semTimeNI$ is naturally ordered and that for any $k$ and $k'$, the set $\{ k'' \mid k \naturalOrder_{\semTimeNI} k' + k''\}$ has a unique smallest element according to $\naturalOrder_{\semTimeNI}$. A semiring is a naturally ordered if $\naturalOrder_{\semTimeNI}$ is a partial order (reflexive, antisymmetric, and transitive).
$k \naturalOrder_{\semTimeNI} k' \Leftrightarrow \exists k'': k \addNI k'' = k'$.
Substituting the definition of addition, we get $\exists k'': \kCoalesce{\semK}(k \addP k'') = k'$.
Since $\kCoalesce{\semK}(k') = k'$ and coalesce preserves snapshot equivalence, we have $\kCoalesce{\semK}(k \addP k'') = k' \Leftrightarrow \forall \tPoint \in \timeDomain: \tSlice{\tPoint}(k) +_{\semK} \tSlice{\tPoint}(k'') = \tSlice{\tPoint}(k')$.
From $\kCoalesce{\semK}(k \addP k'') = k' \Leftrightarrow \forall \tPoint \in \timeDomain: \tSlice{\tPoint}(k) +_{\semK} \tSlice{\tPoint}(k'') = \tSlice{\tPoint}(k')$ follows that $k \naturalOrder_{\semTimeNI} k' \Leftrightarrow \forall \tPoint \in \timeDomain: \tSlice{\tPoint}(k) \naturalOrder_{\semK} \tSlice{\tPoint}(k')$.

Note that we only have to  prove that $\naturalOrder_{\semTimeNI}$ is antisymmetric, since   reflexivity and transitivity of the natural order follows from the semiring axioms and, thus, holds for all semirings.

\myproofpar{Antisymmetric} We have to show that $\forall k,k' \in \semTimeNI: k \naturalOrder_{\semTimeNI} k' \wedge k' \naturalOrder_{\semTimeNI} k \rightarrow k = k'$. This holds because, $k \naturalOrder_{\semTimeNI} k'$ and $k' \naturalOrder_{\semTimeNI} k$ iff for all $\tPoint \in \timeDomain$ we have $\tSlice{\tPoint}(k) \naturalOrder_{\semK} \tSlice{\tPoint}(k')$ and $\tSlice{\tPoint}(k') \naturalOrder_{\semK} \tSlice{\tPoint}(k)$ which implies $\tSlice{\tPoint}(k) = \tSlice{\tPoint}(k')$ for all $\tPoint \in \timeDomain$ which can only be the case if $k \intervalEq k'$. Since $k$ and $k'$ are coalesced it follows that $k = k'$.

\myproofpar{Unique Smallest Element Exists}
  It remains to be shown that $\{ k'' \mid k \naturalOrder_{\semTimeNI} k' + k''\}$ has a smallest member for all $k, k' \in \semTimeNI$. We give a constructive proof by constructing the smallest such element $k_{min}$. $k_{min}$ is defined by coalescing an element $k_{pmin}$ that consists of singleton intervals ($[\tPoint,\tPoint+1)$) as follows:
  \begin{align*}
    k_{min} &= \kCoalesce{\semK}(k_{pmin})\\
 \forall \interval \in \intervalDom:   k_{pmin}(\interval) &=
              \begin{cases}
                \tSlice{\tPoint}(k) -_{\semK} \tSlice{\tPoint}(k') &\mathtext{if} I = [\tPoint,\tPoint+1)\\
                0_\semK & \mathtext{else}\\
              \end{cases}
  \end{align*}

First we have to demonstrate that indeed $k \naturalOrder_{\semTimeNI} k' +_\semK k_{min}$. Recall that $\kCoalesce{\semK}(k' \addP k_{min}) = k \Leftrightarrow \forall \tPoint \in \timeDomain: \tSlice{\tPoint}(k') + \tSlice{\tPoint}(k_{min})) = \tSlice{\tPoint}(k')$. Substituting the definition of $k_{min}$ and using the fact that $\tSlice{\tPoint}$ commutes with addition, for every time point $\tPoint$ we distinguish two cases. Either $\tSlice{\tPoint}(k') \naturalgeq_\semK \tSlice{\tPoint}(k)$ in which case $\tSlice{\tPoint}(k) -_{\semK} \tSlice{\tPoint}(k') = 0_\semK$ and we have: $\tSlice{\tPoint}(k') +_{\semK} (\tSlice{\tPoint}(k) -_{\semK} \tSlice{\tPoint}(k')) = \tSlice{\tPoint}(k') + 0_\semK = \tSlice{\tPoint}(k')$. Thus, $\tSlice{\tPoint}(k') +_{\semK} (\tSlice{\tPoint}(k) -_{\semK} \tSlice{\tPoint}(k'))  \naturalgeq_{\semK} k \tSlice{\tPoint}(k')$ reduces to $\tSlice{\tPoint}(k') \naturalgeq_{\semK} \tSlice{\tPoint}(k)$ which was assumed to hold.

Otherwise for $\tSlice{\tPoint}(k') \naturalOrder_{\semK} \tSlice{\tPoint}(k)$ we have:
$\tSlice{\tPoint}(k') +_{\semK} (\tSlice{\tPoint}(k) -_{\semK} \tSlice{\tPoint}(k'))$.
Let $k'' = (\tSlice{\tPoint}(k) -_{\semK} \tSlice{\tPoint}(k'))$. Substituting the definition of $-_{\semK}$, we get
$k'' = min_{k'''} \tSlice{\tPoint}(k') + k''' \naturalgeq_{\semK} \tSlice{\tPoint}(k)$. Thus,
$$\tSlice{\tPoint}(k') +_{\semK} k'' \naturalgeq_{\semK} k.$$

It remains to be shown that $k_{min}$ is minimal. For contradiction assume that there exists a smaller such member $k_{alt}$. Then there has to exist at least one time point $\tPoint$ such that $\tSlice{\tPoint}(k_{alt}) \naturalStrictOrder_{\semK} \tSlice{\tPoint}(k_{min})$. We have to distinguish two cases. If $\tSlice{\tPoint}(k) \naturalOrder_\semK \tSlice{\tPoint}(k')$, then $\tSlice{\tPoint}(k_{min}) = 0_\semK$. However, since $0_\semK \leq k$ for any $k \in \semK$ this leads to a contradiction. Otherwise $\tSlice{\tPoint}(k') + \tSlice{\tPoint}(k_{alt}) \naturalStrictOrder_{\semK} \tSlice{\tPoint}(k') +_\semK \tSlice{\tPoint}(k_{min}) = \tSlice{\tPoint}(k)$ contradicting the assumption that $k \naturalOrder_\semK k' +_\semK k_{alt}$. \end{proof}

\begin{proof}[Proof of Theorem \ref{theo:bag-interval-norm-monus-homo}]
  We have to prove that $\tSlice{\tPoint}(k \monNI k') = \tSlice{\tPoint}(k) -_{\semK} \tSlice{\tPoint}(k')$. We start with $\tSlice{\tPoint}(k \monNI k') = \tSlice{\tPoint}(\kCoalesce{\semK}(k \monP k'))$.
Since $\kCoalesce{\semK}$ preserves $\intervalEq$ and $\tSlice{\tPoint}(k) = \tSlice{\tPoint}(k')$ if $k \intervalEq k'$, we get:
\begin{align*}
  =   &\sum_{\tPoint \in \interval} (k \monP k')(\interval)
  = &\tSlice{\tPoint}(k) -_{\semK} \tSlice{\tPoint}(k') \tag*{\qedhere} \end{align*}
\end{proof}

\begin{proof}[Proof of Theorem~\ref{theo:our-agg-eq-their-agg}]
By construction, the result of aggregation is a $\semTimeNIN$ relation (it is coalesced). Also by construction, we have $\tSlice{\tPoint}(\aggregation{G}{f(A)}(R)) = \aggregation{G}{f(A)}(\tSlice{\tPoint}(R))$.
\end{proof}
\proofsketch{\revm{
Proven by induction over the structure of an algebra expression.
    }}

\begin{proof}[Proof of Theorem~\ref{theo:reduction-is-correct}]
To prove the relationships in the commutative diagram of Equation~~\eqref{eq:n-enc-cd}, we have to prove that $\reprN^{-1}(\reprN(R)) = R$ and that queries commute with $\reprN$ if rewritten using $\reprRewr$, i.e., $\reprN(Q(R)) = \reprRewr(Q)(\reprN(R))$.

\myproofpar{$\reprN^{-1}(\reprN(R)) = R$}
Let $R$ be a $\semTimeNIN$-relation and $R'$ denote $\reprN(R)$. Consider an arbitrary tuple $\tuple$ and let $\anyTE$ denote the temporal element associated with $\tuple$, i.e., $R(\tuple) = \anyTE$. Consider any interval $\interval \in \intervalDom$ and let $n_\interval = \anyTE(\interval)$ (the multiplicity assigned by $\anyTE$ to $\interval$).
According to Definition~\ref{def:N-enc}, this implies that tuple $\tuple_\interval = (t, \iBegin{\interval}, \iEnd{\interval})$ is annotated with $n_{\interval}$. Let $\anyTE_t$ denote the temporal element assigned by $\reprN^{-1}$ to $t$. By construction $\anyTE_t(\interval) = n_{\interval} = \anyTE(\interval)$.

\myproofpar{$\reprN(Q(R)) = \reprRewr(Q)(\reprN(R))$} We prove this part by induction over the structure of a query. Let $R' = \reprN(R)$.

\myproofpar{Base case} Assume that $\query = R$ for some relation $R$. The claim follows immediately from $\reprRewr(R) = R$.

\myproofpar{Induction Step} Assume the claim holds for queries with up to $n$ operators. We have to prove the claim for any query $\query$ with $n+1$ operators. For unary operators, WLOG let $\query = op(\query_n)$ for an operator $op$ and query $\query_n$ with $n$ operators and let $\query' = \reprRewr(\query)$.

\myproofpar{Selection: $op = \selection_{\theta}$} A selection is rewritten as $\query' = \coalesceOp(\selection_\theta(\reprRewr(\query)))$. Consider an input tuple $\tuple$ from $R$. The temporal $\semK$-element $\anyTE$ annotating tuple $\tuple$ is represented as a set of tuples of the form $(t, \iBegin{\interval}, \iEnd{\interval})$ for some interval $\interval$. If $\tuple$ fulfills the selection, then $\tuple$ is annotated with $\anyTE$ in the result. In $R'$, all of these tuples are in the result of $\query'$ if $t \models \theta$ and applying $\reprN^{-1}$ we get $\anyTE$ as the annotation of $\tuple$. If $\tuple$ does not fulfill the condition then $\tuple$ is annotated with $0$ in both encodings.

\myproofpar{Projection: $op = \projection_{A}$} A projection is rewritten by adding the attributes encoding the interval associated to a tuple to the projection expressions. There will be one tuple $(t, \iBegin{\interval}, \iEnd{\interval})$ in the result for each
 interval $\interval$ assigned a non-zero annotation in $R(u)$ for any tuple  $u$ projected on tuple $\tuple$. Function $\reprN^{-1}$ creates the annotation of an output as a temporal element that maps each interval mapping to a non-zero annotation in $\reprN(R)$ to that annotation. This corresponds to addition of singleton temporal elements and based on the fact that addition is associative this implies that the annotation of $\tuple$ in the output will be the sum of temporal elements $R(u)$ for each $u$ projected onto $\tuple$. Thus, the claim holds.

\myproofpar{Aggregation: $op = \aggregation{}{f(A)}$} The rewrite for aggregation without group-by utilizes the split operation $\normalize$ we have defined. Note that $\normalize_{\emptyset}$ returns a $\semTimeNIN$-relation $S$ where for any pair of tuples $t$ and $t'$ and any pair of intervals $\interval_1$ and $\interval_2$ we have $\interval_1 \neq \interval_2 \wedge \interval_1 \cap \interval_2 \neq \emptyset \Rightarrow S(\tuple')(\interval_1) = 0 \vee S(\tuple')(\interval_2) = 0$. That is, all intervals with non-zero annotations from any pair of temporal elements do not overlap or are the same.  From that follows that for any two time points $\tPoint_1, \tPoint_2 \in \interval$ for an interval $\interval$ that is mapped to $n \neq 0$ in the annotation of  at least  one tuple $S$, the value of the result of aggregation is the same for the snapshots at $\tPoint_1$ and $\tPoint_2$. Thus, grouping by the interval boundaries yields the expected result with the exception of an empty snapshot. However, since a tuple $(0_f, \tMin, \tMax)$ is added to the input, the aggregation will produce $0$ (count) or \lstinline!NULL! (other aggregation functions) for intervals containing only empty snapshots. This does not effect the result of the aggregation for non-empty snapshots, because $0_f$ is the neutral element of the aggregation function $f$.

\myproofpar{Aggregation: $op = \aggregation{G}{f(A)}$} For aggregation with group-by, split is applied grouping on $G$ and no additional tuple $(0_f, \tMin, \tMax)$ is added to the input. Since the tuples within one group are split, the argument we have used above for aggregation without group-by applies also to aggregation with group-by.

For binary operators WLOG let $\query = op(\query_l, \query_r)$ where the total number of operators in $\query_l$ and $\query_r$ is $n$.

\myproofpar{Join: $op = \query_l \join_\theta \query_r$} Consider a tuple $\tuple$ that is the result of joining tuples $u$ and $v$. Let $\anyTE_u$ and $\anyTE_v$ be the temporal elements annotating $u$ and $v$ in the input, respectively. Based on the definition of the rewriting, in the result of the rewritten join there will be a tuple $\tuple, \iBegin{\interval}, \iEnd{\interval}$ annotated with $\sum_{\interval_u, \interval_v} \query_l(u) \cdot \query_r(v)$ for all intervals $\interval_u$ and $\interval_v$ such that $\interval = \interval_u \cap \interval_c$. This corresponds to the definition of multiplication (join) in $\semTimeNIN$.

\myproofpar{Union: $op = \query_l \union \query_r$} Union is rewritten as a union of the rewritten inputs. For any tuple $\tuple$, let $\anyTE_l = \query_l(\tuple)$ and $\anyTE_r = \query_r(\tuple)$. In the result of the union applied by $\query'$ a tuple
 $(\tuple, \iBegin{\interval}, \iEnd{\interval})$ for each interval $\interval$ will be annotated with $\anyTE_l(\interval) + \anyTE_r(\interval)$. The result of the union is then coalesced. Applying $\reprN^{-1}$ the  annotation computed for $\tuple$ is equivalent to $\kCoalesce{\semN}(\anyTE_l \addP \anyTE_r)$.

\myproofpar{Difference: $op = \query_l - \query_r$} A difference is rewritten by applying difference to the pairwise normalized inputs. Recall that the monus operator of $\semTimeNIN$ associates the result of the monus for $\semN$ to each snapshot of a temporal $\semN$-element. Since the split operator adjusts intervals such that there is no overlap, the claim holds.
\end{proof}

\section{Query Descriptions}
\label{sec:workload-queries}

\subsection{MySQL Employee Dataset}
\label{sec:mysql-empl-datas-workload}

\parttitle{join-1}
Return the salary and department for every employee.
\lstset{basicstyle=\scriptsize\upshape\ttfamily}
\begin{lstlisting}
SELECT a.emp_no, dept_no, salary
FROM dept_emp a JOIN salaries b ON (a.emp_no = b.emp_no)
\end{lstlisting}

\parttitle{join-2}
Return the department, salary, and title for every employee.

\begin{lstlisting}
SELECT title, salary, dept_no
FROM dept_emp a JOIN salaries b ON (a.emp_no = b.emp_no)
                JOIN titles c ON (a.emp_no = c.emp_no)
\end{lstlisting}

\parttitle{join-3}
Return employees that manage a particular department and earn more then \$70,000.

\begin{lstlisting}
SELECT a.emp_no, dept_no
FROM dept_manager a
     JOIN salaries b ON (a.emp_no = b.emp_no)
WHERE salary > 70000
\end{lstlisting}

\parttitle{join-4}
Returns information about the manager of each department.

\begin{lstlisting}
SELECT a.emp_no, a.dept_no, b.salary, first_name,
       last_name
FROM dept_manager a, salaries b, employees e
WHERE a.emp_no = b.emp_no and a.emp_no = e.emp_no
\end{lstlisting}

\parttitle{agg-1}
Returns the average salary of employees per department.

\begin{lstlisting}
SELECT dept_no, avg(salary) as avg_salary
FROM dept_emp a
     JOIN salaries b ON (a.emp_no = b.emp_no)
GROUP BY dept_no
\end{lstlisting}

\parttitle{agg-2}
Returns the average salary of managers.

\begin{lstlisting}
SELECT avg(salary) as avg_salary
FROM dept_manager a
     JOIN salaries b ON (a.emp_no = b.emp_no)
\end{lstlisting}

\parttitle{agg-3}
Returns the number of departments with more than 21 employees.

\begin{lstlisting}
SELECT count(1)
FROM (SELECT count(*) AS c, dept_no
	  FROM dept_emp WHERE emp_no < 10282
	  GROUP BY dept_no HAVING count(*) > 21) s
\end{lstlisting}

\parttitle{agg-join}
Returns the names of employees with the highest salary in their department. It contains a 4-way join where one of the join inputs is the result of a subquery with aggregation.

\begin{lstlisting}
SELECT d.emp_no, e.first_name, e.last_name,
       maxS.max_salary, d.dept_no
FROM (SELECT max(salary) as max_salary,dept_no
      FROM dept_emp a
           JOIN salaries b ON (a.emp_no = b.emp_no)
      GROUP BY dept_no) maxS,
      salaries s, dept_emp d, employees e
WHERE e.emp_no = s.emp_no
      AND s.salary = maxS.max_salary
	  AND d.dept_no = maxS.dept_no
      AND d.emp_no = e.emp_no
\end{lstlisting}

\parttitle{diff-1}
Returns employees that are not managers of any department.

\begin{lstlisting}
SELECT emp_no FROM dept_emp
EXCEPT ALL
SELECT emp_no FROM dept_manager
\end{lstlisting}

\parttitle{diff-2}
Returns salaries of employees that are not managers.

\begin{lstlisting}
SELECT a.emp_no, salary
FROM (SELECT emp_no FROM dept_emp
      EXCEPT ALL
      SELECT emp_no FROM dept_manager) a
      JOIN salaries b ON (a.emp_no = b.emp_no)
\end{lstlisting}

\parttitle{C-Sn}
To evaluate the performance of coalescing we use the following query template
varying the selection condition on salary to control the size of the output.
The query returns employee salaries. We materialize the result of this query
for each selectivity and use this as the input to coalescing.

\begin{lstlisting}
SELECT a.EMP_NO, salary
FROM employees a
     JOIN salaries b ON (a.emp_no = b.emp_no)
WHERE salary > ?
\end{lstlisting}

\subsection{Tourism Dataset}
\label{sec:tourism-dataset}

Recall that this dataset stores travel booking inquiries in the South Tyrol area in Italy.

\parttitle{tou-join-agg}
This query returns for each destination the sum of the number of persons of all bookings for  this destination paired with the average number of this sum for all other destinations.

\begin{lstlisting}
WITH numPerDest AS (
    SELECT sum(adults + children) AS numT,
           destination AS dest
    FROM tourismdata
    GROUP BY destination
)
SELECT n.numT, n.dest, avg(o.numT) AS otherAvg
FROM numPerDest n, numPerDest o
WHERE n.dest <> o.dest
GROUP BY n.numT, n.dest
\end{lstlisting}

\parttitle{tou-agg-1}
For destinations with more than 1000 inquiries return the number of inquiries for this destination and the total number of persons for which inquiries were made.

\begin{lstlisting}
SELECT destination,
       count(*) AS numEnquiry,
       sum(adults + children) AS numTourists
FROM tourismdata
GROUP BY destination
HAVING count(*) > 1000;
\end{lstlisting}

\parttitle{agg-2}
For each destination return the maximum number of persons per inquiry.

\begin{lstlisting}
SELECT destination,
       max(adults + children) AS maxTourists
FROM tourismdata
GROUP BY destination;

\end{lstlisting}

\parttitle{tou-agg-3}
Find the maximum number of inquiries from the total number of inquiries per destination.

\begin{lstlisting}
SELECT max(cnt) AS maxInq
FROM (SELECT count(*) AS cnt
      FROM tourismdata
      GROUP BY destination)
\end{lstlisting}

\parttitle{agg-join}
This query returns the total number of inquiries per continent.

\begin{lstlisting}
SELECT continent, count(*) AS numEnquiries
FROM tourismdata t, country c
WHERE t.countrycode = c.countrycode
GROUP BY continent
\end{lstlisting}

\section{Pulling-up Coalescing}
\label{sec:optimizations}

The main overhead of our approach for snapshot temporal queries
compared to non-temporal query processing is the extensive use of
coalescing, which can be expensive if naively implemented in
SQL. Furthermore, the application of coalescing after each operation
may prevent the database optimizer from applying standard
optimizations such as join reordering.  To address this issue, we now
investigate how to reduce the number of coalescing steps. In fact, we
demonstrate that it is sufficient to apply coalescing as a last step
in query processing instead of applying it to intermediate
results. Similar optimizations have been proposed by Bowman et
al.~\cite{BowmanT03} for their multiset temporal normalization
operator and by B\"ohlen et al.~\cite{DBLP:conf/vldb/BohlenSS96} for
set-coalescing.

Consider how a $\raPlus$ query $\query$ is evaluated over an $\semTimeNI$-database. $\raPlus$ over K-relations computes the annotation of a tuple in the result of a query using the addition and multiplication operations of the semiring. That is, the annotation of any result tuple is computed using an arithmetic expression over the annotations of tuples from the input of the query. In the case of a semiring $\semTimeNI$, addition and multiplication are defined as coalescing a temporal element that is computed based on point-wise application of the addition (multiplication) operations of semiring $\semK$ (denoted as $\addP$ and $\multP$).
Recall from Lemma~\ref{lem:coalesce-push} that coalescing can be redundantly pushed into the addition and multiplication operations of interval-temporal semirings, e.g., $\kCoalesce{\semK}(k \addP k') = \kCoalesce{\semK}(\kCoalesce{\semK}(k) \addP k')$. Interpreting this equivalence from right to left and applying it repeatedly to an arithmetic expression $e$ using $\addNI$ and $\multNI$, the expression can be rewritten into an equivalent expression of the form $\kCoalesce{\semK}(e')$, where $e'$ is an expression that only uses operations $\addP$ and $\multP$. Now consider expressions that also include applications of the monus operator $\monNI$. This operator is defined as $\kCoalesce{\semK}(k \monP k')$. The $\monP$ operator computes the timeslice of the inputs at every point in time and then applies $\monK$ to each timeslice. According to Lemma~\ref{lem:coalesce-properties}, $\tSlice{\tPoint}(k) \intervalEq \tSlice{\tPoint}(\kCoalesce{\semK}(k'))$. Thus, the result of $\monP$ is independent of whether the input is coalesced or not.

\begin{lem}\label{lem:single-coalesce-enough}
Any arithmetic expression $e$ using operations and elements from an period m-semiring $\semTimeNI$ is equivalent to an expression of the form $\kCoalesce{\semK}(e')$, where $e'$ only contains operations $\addP$, $\multP$, and $\monP$.
\end{lem}

Lemma~\ref{lem:single-coalesce-enough} implies that it is sufficient to apply coalescing as a last step in a rewritten query $\reprRewr(\query)$ instead of after each operator.

\begin{coll}[Coalesce Pullup]\label{coll:coalesce-pullup}
  For any $\ra$ query $\query$, $\reprRewr(\query)$ is equivalent to a query $\query'$ which is derived from $\reprRewr(\query)$ by removing all but the outermost coalescing operator.
\end{coll}
\begin{proof}  Operations $\addNI$, $\multNI$, and $\monNI$ are defined as applying $\kCoalesce{\semK}$ to the result of operations $\addP$, $\multP$, and $\monP$, respectively. Thus, expression $e$ is equivalent to an expression that interleaves the $\kCoalesce{\semK}$ as well as $\addP$, $\monP$, and $\multP$ operations. To prove this, we first prove that the following equivalence holds: $\kCoalesce{\semK}(k \monP k') \Leftrightarrow \kCoalesce{\semK}(\kCoalesce{\semK}(k) \monP k') \Leftrightarrow \kCoalesce{\semK}(k \monP \kCoalesce{\semK}(k'))$. Consider the definition of $\monP$. Every interval $\interval = [\tPoint, \tPoint +1)$ is assigned the annotation $\tSlice{\tPoint}(k) \monK \tSlice{\tPoint}(k')$. Applying  Lemma~\ref{lem:coalesce-properties} we get $\tSlice{\tPoint}(k) = \tSlice{\tPoint}(\kCoalesce{\semK}(k))$ and $\tSlice{\tPoint}(k') = \tSlice{\tPoint}(\kCoalesce{\semK}(k'))$. Thus, the equivalence holds.
By repeatedly applying this equivalence and the equivalences proven in Lemma~\ref{lem:coalesce-push}, all except the outermost K-coalesce operations can be removed resulting in an expression of the form $\kCoalesce{\semK}(e')$ where $e'$ does not contain any coalesce operations.
\end{proof}

\begin{exam}
  Consider the following query $\query = S - \projection_{sal}(\selection_{sal < sal'}(S \times \rename_{sal' \gets sal}(S))$  that returns the largest salary from relation $S$ as shown in Figure~\ref{fig:example_temp_coalesce} (consider the corresponding $\semTimeNIN$-relation using the annotation shown on the right in this figure coalesced as shown in Example~\ref{ex:k-coalesce}). Consider how the annotation of tuple $r = (50k)$ in the result of $\query$ is computed. Applying the definitions of difference, projection, and join over K-relations and denoting the database instance of $S$ as $D$, we obtain:
  \begin{align*}
    \query(D)(t) &= \kCoalesce{\semN}(S(t) \monP \kCoalesce{\semN}(\sum_{u = (v,w): u.sal = t}\\ &\kCoalesce{\semN}(\kCoalesce{\semN}((S(v) \multP S(w))) \multP (sal < sal')(u))))
  \end{align*}
  Pulling up coalesce we get:
  \begin{align*}
    \query(D)(t) &= \kCoalesce{\semN}(S(t) \monP{} \\ &\sum_{u = (v,w): u.sal = t} (S(v) \multP S(w)) \multP (sal < sal')(u))
  \end{align*}
\end{exam}

\section{Interaction of Our Approach with Query Optimization}
\label{sec:inter-our-appr}

In this section we briefly discuss the impact of our rewrite-based approach for
implementing snapshot semantics on query optimization. Importantly, the
combination of uniqueness and snapshot reducibility guarantees that queries are
equivalent wrt. our logical model precisely when they are equivalent under
regular $\semK$-relational semantics. As a special case of this result, queries
over $\semTimeNIN$ relations are equivalent iff they are equivalent under bag
semantics ($\semN$-relations). That is, any query equivalence that is
applied by classical database optimizers, e.g., join reordering,
can be applied to optimize snapshot queries.

That being said, we pass a rewritten query to the DBMS optimizer which is not
aware of the fact that this query implements snapshot semantics. The
preservation of bag semantics query equivalences does not necessary imply that
these rewritten queries can be successfully optimized by a general purpose query
optimizer. However, as we will explain in the following, our approach is
designed to aid the database optimizer in finding a successful plan. First off,
note that our rewrites essentially keep the structure of the input query intact
with the exception of the introduction of split before aggregation and
difference, and coalescing which is applied as a final step for every snapshot
query. Every other operator is preserved in the rewritten query, e.g., joins,
are rewritten into joins.

\begin{exam}\label{ex:spj-rewriting}
  Consider the following query $Q = \projection_{name, city} (person \join_{name=pName} livesAt \join_{address=aId} address)$ over relations
  \begin{center}{\ttfamily\upshape
    person(name, age, A\textsubscript{Begin}, A\textsubscript{End})\\
    livesAt(pName, address, A\textsubscript{Begin}, A\textsubscript{End})\\
    address(aId, city, zip, street, A\textsubscript{Begin}, A\textsubscript{End})
    }
  \end{center}

 This query returns for each person the city(ies) they live in. Applying $\reprRewr$ we get the query shown in Figure~\ref{fig:rewr-spj-example}. Note how the structure of the input query was preserved. The exception are the coalescing operator  at the end and the introduction of new projections. However, typically database optimizers will at least consider a transformation called subquery pull-up (called view merging in Oracle) which would pull-up and merge these projections. Thus, these projections do not hinder join reordering.
\end{exam}

\begin{figure*}[t]
  \centering

  \begin{tikzpicture}
[op/.style={anchor=south},
pro/.style={red,font=\footnotesize},
conn/.style={-,line width=1pt}]

\node[op] (c) at (0,3) {$\coalesceOp$};
\node[op] (p) at (0,2) {$\projection_{name, city, \attrBegin, \attrEnd}$};

\node[op] (p1) at (0,1) {$\projection_{name, age, pName, address, max(person.\attrBegin, livesIn.\attrBegin), min(person.\attrEnd, livesIn.\attrEnd)}$};
\node[op] (j1) at (0,0) {$\join_{name = pName \wedge person.\attrBegin < livesIn.\attrEnd \wedge livesIn.\attrBegin < person.\attrEnd}$};

\node[op] (p2) at (-4,-1) {$\projection_{name, age, pName, address, aId, city, zip, street}$};
\node[op] (j2) at (-4,-2) {$\join_{address=aId \wedge person.\attrBegin < livesIn.\attrEnd \wedge livesIn.\attrBegin < person.\attrEnd}$};

\node[op] (r1) at (-5,-3) {$person$};
\node[op] (r2) at (-3,-3) {$livesAt$};
\node[op] (r3) at (3,-3) {$address$};

\draw[conn] (p) to (c);
\draw[conn] (p) to (p1);

\draw[conn] (p1) to (j1);

\draw[conn] (j1) to (p1);
\draw[conn] (p2) to (j1);
\draw[conn] (r3) to (j1);

\draw[conn] (r1) to (j2);
\draw[conn] (r2) to (j2);
\draw[conn] (j2) to (p2);

\end{tikzpicture}
\\[3mm]

  \caption{Rewriting $\reprRewr(Q)$ for SPJ query $Q$ from Example~\ref{ex:spj-rewriting}}
  \label{fig:rewr-spj-example}
\end{figure*}
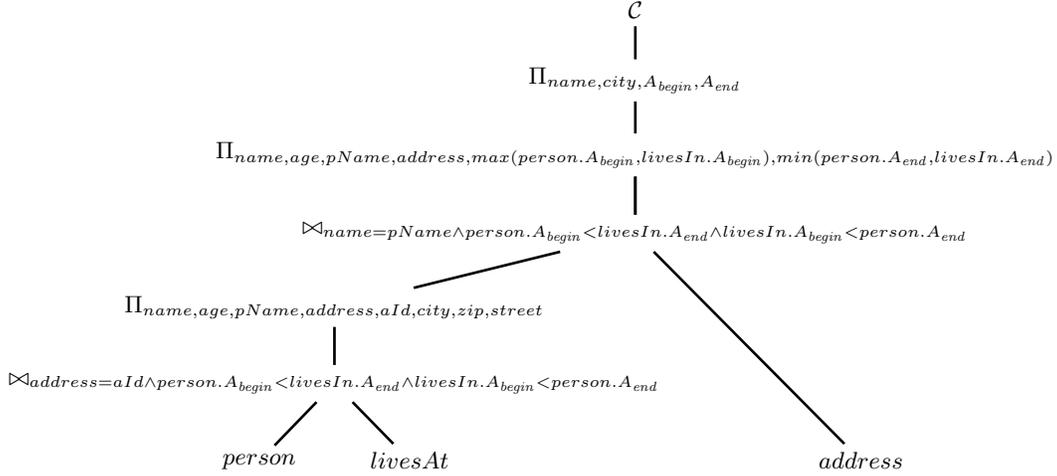

\section{SQL Implementations of Bag Coalescing and Split}
\label{sec:optimizations}

In the following, we explain our implementation of bag coalescing in SQL
using a step by step example. Afterwards, we present the implementation of
the split operator integrated with aggregation and (bag) difference.

\subsection{Bag Coalesce}
\label{sec:bag-coalesce}

Figure~\ref{fig:-bag-coalesce} shows the SQL code for computing bag coalescing
for a table recording the activity of production machines.
Figure~\ref{fig:intermeidate-result-bag-col} shows an example instance of this table and the intermediate and final results produced by the query for this instance. Here we assume that periods are stored as two timestamp attributes $\bm{t_{start}}$ and $\bm{t_{end}}$ recording the start and the end of the period. The input table \textbf{active}  with the schema (\textbf{mach}, $\bm{t_{start}}$, $\bm{t_{end}}$) is shown on the top-left of Figure~\ref{fig:intermeidate-result-bag-col}. Each row in the table records a time interval (from $\bm{t_{start}}$ to $\bm{t_{end}}$) during which a  machine (\textbf{mach}) is running. For convenience we show a timeline with the intervals encoded by this table.

Before explaining the steps of the SQL implementation, we review bag
coalescing. To coalesce an input we have to determine for each tuple $t$ its
annotation change points, i.e., the end points of maximum intervals during which
the multiplicity of the tuple does not change. Then for each adjacent pair of
change points we output a number of duplicates of tuple $t$ that is equal to the
number of duplicates of tuple $t$ in the input whose intervals cover the two
change points. This could be implemented as a native operator which splits each
tuples associated with a period into two tuples with the intervals end points
where each generated tuple is marked to indicate whether it represents an
interval start or end point. Then any aggregation algorithm can be applied to
calculate the number of intervals associated with a tuple that open and close at
a particular time point.  The output of this step is then sorted on the
non-temporal attributes and secondary on the timestamp attribute. The final
result is produced by scanning through the sorted output once outputting for
each tuple and adjacent pair of change points a number of duplicates determined
based on the number of intervals covering these change points which is
determined based on the counts of opening and closing intervals.  We leave a
native implementation and further optimizations (e.g., we could partition the
input on the non-temporal attributes and then process multiple such partitions
in parallel) to future work and now explain how our SQL implementation realizes
the computational steps outlined above.

\parttitle{Determine the Number of Opening and Closing Intervals Per Change Point}
In lines 3 - 24 of Figure~\ref{fig:intermeidate-result-bag-col} we compute the annotation change points for each tuple and the number of intervals that are opening and closing for each such change point. This is done by counting for each tuple and one of its change points the number of opening and closing intervals separately and then for each such pair merge the number of opening and closing counts into a single output tuple. Note that strictly speaking not all of the time points returned by this query are guaranteed to be annotation change points. The actual change points are computed in one of the following steps as explained below. For the example instance there is only one pair $(M1, 5)$ where time point $\tPoint$ (attribute \texttt{t}) is both the start and end point of an interval associated with tuple $(M1)$. As another example consider time point $6$ which is the end point of two intervals associated with tuple $(M2)$ corresponding to the last tuple in the result of subquery \texttt{change\_points}. The pre-aggregation before the union is merely a performance tweak. It turns a single aggregation over $2 \cdot \card{active}$ tuples into two aggregations over $\card{active}$ tuples.

\parttitle{Counting Open Intervals}
Line 26 - 35 of the query create the common table expression \texttt{num\_intervals} which returns the number of open intervals for a tuple per time point $\tPoint$. This is achieved by subtracting the number of intervals for this tuple with an end point that is less than or equal to $\tPoint$ (attribute \texttt{t}) from the number of intervals with a start point that is less than or equal to $\tPoint$. Intuitively, the number of open intervals for a tuple is the number of duplicates of the tuple that exist in the time interval between $\tPoint$ and the adjacent following change point. We compute these running sums using SQL's window functions partitioning the input on the non-temporal attributes (\texttt{mach} in this example) and within each partition order the tuple based on the timestamp \texttt{t} computing the aggregate over a window including all tuples with a timestamp less than or equal to \texttt{t}. For example, consider the first tuples in the instance of \texttt{num\_intervals} as shown in Figure~\ref{fig:intermeidate-result-bag-col}. This tuple records that there two duplicates of tuple $(M1)$ exist at time point $1$.

\parttitle{Removing Spurious Change Points}
Recall that bag coalescing determines maximal intervals during which the annotation (multiplicity in the case of bag semantics) of a tuple is constant. As shown in the example, \texttt{num\_intervals} may contain adjacent time points with the same number of open intervals which, according to the definition of $\semK$-coalescing, are not annotation change points. Subquery \texttt{diff\_previous} (Lines 38-47) computes the difference between the number of open intervals at a time point and the previous time point. Subquery  \texttt{changed\_intervals} (Lines 49-53) removes tuples where this difference is zero (the number of duplicates has not changed).

\parttitle{Reconstructing Intervals}
At this point in the computation we have calculated the set of annotation changepoints for each tuple and the number of duplicates of the tuple that exist during the time interval between each two adjacent annotation change points.
Subquery \texttt{pair\_points} (Lines 55-65) computes pairs of adjacent annotation change points.
For instance, the first tuple in the result for the example records that there exists two duplicates of tuple \texttt{(M1)} during time interval $[1,7]$. The subquery of \texttt{pair\_points} also returns tuples with $t_{start}$ equal to the last change point of teach tuple. These tuples are filtered in the \lstinline!WHERE! clause of the outer query. For example, the tuples marked in red in the result of the subquery as shown in Figure~\ref{fig:intermeidate-result-bag-col} are such tuples.

\parttitle{Generating Duplicates}
In the last step, we generate duplicates of  tuples based on the counts stored in attribute \texttt{$\#_{open}$}. One way to realize this would be to use a set-returning function that takes as input a tuple $t$ and a count $c$, and returns $c$ duplicates of $t$. While perfectly viable, to avoid the overhead of calling a user-defined function for every distinct output tuple, we use subquery \texttt{max\_seq} (Lines 67-72) to generate a table storing a sequence of numbers $\{1, \ldots, m\}$ where $m$ is the maximum number of duplicates of any tuple in the query result. We then join this table with \texttt{pair\_points} (lines 74-76). For the example database the maximum number of duplicates for any tuple and time point is $2$. Hence, subquery \texttt{max\_seq} returns $\{(1), (2)\}$. The final join with \texttt{pair\_points} then returns the appropriate number of duplicates for each tuple using the counts stored in $\#_{open}$, e.g., there are 2 duplicates of tuple $(M2)$ during time interval $[3,6]$.

\subsection{Split Operator Implementation}
\label{sec:split-oper-impl}

We first introduce our implementation of the split operator and then afterwards discuss the optimized versions of the aggregation and difference which incorporate split. We show the SQL code for $\normalizeOp_{mach}(active,active)$, i.e., splitting the intervals of relation \lstinline!active! based on its own interval boundaries for attribute \lstinline!mach! (the only attribute of this relation). Figure~\ref{fig:-normalization} shows the SQL code generated by our system and Figure~\ref{fig:intermeidate-result-normalize} shows an example database the intermediate results of produced by the SQL implementation for this example. In lines 2-9 we assign aliases to the left and right input. For this particular example, both inputs are table \lstinline!active!.

\parttitle{Computing Interval End Points}
The first of the computation (lines 10-23) generates the set of all interval end points for both inputs. Note that for this example where a relation is split wrt. itself the four-way union is not necessary and can be replaced with a two-way union. In our implementation we apply this optimization, but for sake for the example we show the four-way union to illustrate how the approach would work when a relation is split wrt. to another relation.

\parttitle{Creating Unique Identifiers for Intervals}
Next we assign a unique identifier to each tuple from the left input (lines 25-32). For instance, there are three such tuples in the example shown in Figure~\ref{fig:intermeidate-result-normalize}.

\parttitle{Pair Intervals with End Points}
We now join the left input with all endpoints we have computed beforehand (lines 34-46) such that each interval from the left input is paired with all end points it contains with the exception of the maximum point in the interval. Intuitively the purpose of this step is to creating sufficiently many duplicates of each input tuples to be able to generate the split versions of the interval for this tuple. Furthermore, the end points we have paired with an interval will be the starting points of the split intervals. In Figure~\ref{fig:intermeidate-result-normalize} we highlight tuples with colors to indicate which tuples correspond to the same input interval. For example, the first two tuples in the result of \texttt{split\_points} correspond to the tuple with id $1$ and the  starting points of the two intervals this interval will be split into (end point $4$ is contained in the interval $[1,7]$).

\parttitle{Generating Split Intervals}
Finally, we adjust the start ($t_{start}$) and end points ($t_{end}$) of each interval produced in the previous step (lines 48-52). The start point is set to the time point $t$ (the time point from the set of interval end points we have paired with the interval) and the end point is the next larger time point associated with the same interval identifier (or the end point of the interval is no such time point exists). For example, for the first tuple from the result of subquery \texttt{split\_points} we output tuple \texttt{(M1,1,4)}. As can be seen in the timeline representation of the result shown on the bottom right of Figure~\ref{fig:intermeidate-result-normalize} in the result of split any two intervals associated with the same values of the non-temporal attributes are either equal or disjoint.

\subsection{Combining Split with Temporal Aggregation}
\label{sec:combining-split-with}

There is synergy in combining the split operator with temporal aggregation. The resulting implementation is similar to temporal aggregation algorithms which utilize end point indexes (e.g., aggregation over a timeline index~\cite{DBLP:conf/sigmod/KaufmannMVFKFM13}). These approaches calculate the result of an aggregation function over time using ``sweeping'' by sorting the endpoints of intervals on time and then scan over the data in sort order adding the values of tuples whose intervals start at the current  point in time to the current aggregation result and subtract the values of tuples whose intervals end at this point in time. Note that this only can be applied to aggregation functions like sum and count where it is possible to retract a value (the underlying function, e.g., addition in the case of sum, has an inverse). For aggregation functions min and max it is necessary to maintain a list of previously seen values (although it is not necessary to keep all previous values~\cite{DBLP:conf/ssd/PiatovH17}). We do not use the sweeping technique for min and max, but still apply the pre-aggregation optimization described below.
We explain how to combine split with aggregation using the example query shown in Figure~\ref{fig:-query-normalization_agg} which computes the average consumption (\texttt{consum}) of machines.

\parttitle{Pre-aggregation}
For aggregation functions like sum and count that are commutative, associative, and where the underlying operation has an inverse, we can compute pre-aggregate the input data before computing split points. For that we group on the input query's group-by attribute plus the attributes $t_{start}$ and $t_{end}$ which store the end points of a tuple's period. The pre-aggregation step return partial aggregation results for each list of group-by attribute values and period that occurs with this group. During split these periods may be further subdivided and the final aggregation results will be computed by accumulating results for these subdivisions. For aggregation functions like average that do not fulfill the conditions required for pre-aggregation, but which can be computed by evaluating an arithmetic  expression over the result of other aggregation functions that do, we can still apply this trick to calculate the other aggregation functions and delay the computation of the aggregation we are actually interested in until the end. For example, the query shown in Figure~\ref{fig:-query-normalization_agg} computes an average that can be computed as $sum/count$. Thus, as shown in lines 2-12 of Figure~\ref{fig:-normalization_agg} we compute two aggregation functions grouping on \texttt{mach}, $t_{start}$, and $t_{end}$. The example instance of table \lstinline!active! contains two tuples belonging to the same group which also have the same period: \texttt{(M1,10,1,5)} and \texttt{(M1,20,1,5)}. Based on these two tuples we compute the pre-aggregated result \texttt{(M1,30,1,5)}. Note that no matter what aggregation function we are computing, we always will also compute count since it is needed later in the implementation to determine intervals without results for aggregation with group-by.

\parttitle{Calculate Increase and Decrease of Aggregation Values}
Our approach for computing aggregation functions sum and count uses a sweeping technique which scans over the set of all interval end points paired with in time order. We keep a partial aggregation result and for each time point adds the values of the aggregation input attribute for tuples with intervals that open at this time point and ``retracts'' the values of aggregation input attributes for tuples with intervals that close at this time point. For this purpose, we aggregate to total increase (opening intervals) and retraction (closing intervals) for each time point and group. Consider lines 14-38 in Figure~\ref{fig:-normalization_agg}. Since we are computing aggregation functions sum and count, we store for each time point the increase/decrease for both functions. For that, we use attributes \texttt{add\_c} and \texttt{dec\_c} (count) and \texttt{add\_s} and \texttt{dec\_s} (sum). For interval start points we set attributes recording decrease to $0$ while for end points we points we set the \texttt{add\_\textasteriskcentered} attributes to $0$. Afterwards, we compute the total increase and decrease per time point using aggregation. For instance, consider time point $5$ in the example shown in Figure~\ref{fig:intermeidate-result-normalization-agg}. Two intervals with a total consumption of $30$ close at this time point and one new interval opens with a consumption of $40$. This is encoded in the third tuple \texttt{(M1,1,40,2,30,5)} in the result of subquery \texttt{increase\_decrease}.

\parttitle{Compute Accumulative Totals}
We then calculate the aggregation function result for each group and each point in time where at least one interval for this group starts or ends as the sum of the increases up to and including this point in time and subtract from that the sum of decreases. For example, the third tuple in the result of subquery \texttt{accumulation} shows that at time 5 there are 2 open intervals with their \texttt{consum} values summing up to 80.

\parttitle{Generate Output Intervals}
Finally, we pair each split point and its count and sum with the following split point to produce output intervals and compute the average as the sum divided by the count. This is realized by the inner query of the subquery shown in lines 56-69 in Figure~\ref{fig:-normalization_agg}. Note that it may be the case that no periods start at a given split point. In this case the count would be $0$ (no intervals open during between this time point and the next split point). This is dealt with by the \lstinline!WHERE! clause of the outer query which filters out tuples where the count is 0.

\parttitle{Aggregation Without Group-by}
Recall that for aggregation without group-by we have to return results for time periods where the relation is empty. This is easily achieved in our implementation by adding a dummy interval $[\tMin, \tMax]$ associated with the neutral value of the aggregation function to the result of subquery \texttt{pre\_agg} ($0$ for count and $null$ otherwise). For time periods where the input relation is empty the split operator creates an interval covering the ``gap'' and will return the value we did associate with the dummy interval which is chosen to correspond to the result of an aggregation over an input relation as defined in the SQL standard. For periods where the input is non-empty the result is not affected since the dummy interval is associated with the neutral value of an aggregation function. An additional change that is required is that the final \lstinline!WHERE! clause (Figure~\ref{fig:-normalization_agg}, line 68) has to be changed to \lstinline!$t_{end}$ IS NOT NULL! to (i) return results for gaps (where the count is $0$) and  not return a tuple where $t_{start}$ is the last split point (equal to $\tMax$ for the case of aggregation without group-by).

\subsection{Combining Split with Difference}
\label{sec:combining-split-with-1}

To explain the combined implementation of split with bag difference we evaluate the example query shown in Figure~\ref{fig:-query-normalization_sefdiff} under snapshot semantics. The query returns all machines and their consumption  removing consumptions of machines which have been incorrectly recorded (table \lstinline!faulty!).
The SQL implementation for the snapshot version of this query which uses combined split and difference is shown in Figure~\ref{fig:-normalization_sefdiff}. We show an example instance and intermediate results for the query in Figure~\ref{fig:intermeidate-result-normalization-set-diff}. We combine the split operator with bag difference by reducing bag difference to the problem of count aggregation. Consider a snapshot at time $\tPoint$ and tuple $t$ and assume that $t$ appears in the left input with multiplicity $n$ and in the right input with multiplicity $m$ at $\tPoint$. Then we have to return $max(0, n-m)$ duplicates of tuple $t$ for this snapshot. This can be achieved by computing counts for each interval end point in the left and in the right input and then subtracting the counts of the right hand side from the counts of the left hand side. The combination of counting and split essentially uses the approach described in Appendix~\ref{sec:combining-split-with}.

\parttitle{Computing Changes in Multiplicities}
Subquery \texttt{end\_point\_counts} (Figure~\ref{fig:-normalization_sefdiff}, lines 15-36) computes the number of opening and closing intervals for both inputs. We count the end points from the right input negatively. For instance, in the example the second tuple in the result of this subquery records that there are two opening intervals for tuple \texttt{(M1, 40)} at time $1$.

\parttitle{Aggregate Multiplicities}
Next, we use subquery \texttt{acc\_counts} aggregate the multiplicities to get a single count of opening and closing intervals per time point (lines 49-63).
Note that in the result of this subquery both $\#_{open}$ and $\#_{close}$ may be negative. This has to be interpreted as that there is a larger number of opening/closing intervals from the right input than the left input.

\parttitle{Generating Intervals}
We now pair adjacent time points (lines 49-63) to create intervals and compute the final multiplicity for each tuple.

\parttitle{Final Result}
To compute the final result of the difference operator we have to create the right amount of duplicates for each tuple. The method we apply here is exactly the same as the one applied for aggregation: we join the result of subquery \texttt{intervals} with a table contain numbers $1$ to $n$ where $n$ is the maximum multiplicity across all tuples and time points.

\subsection{Coalesce after Split}
\label{sec:coalesce-normalization}

We can also apply coalesce (introduced in Section~\ref{sec:bag-coalesce}) after split (introduced in Section~\ref{sec:split-oper-impl}), for example, we apply split the table \textbf{active} with the schema (\textbf{mach}, $\bm{t_{start}}$, $\bm{t_{end}}$) and apply coalesce afterwards, Figure~\ref{fig:coalesce-normalization} shows this workflow.

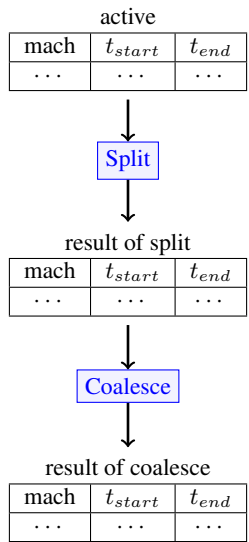
\begin{figure}[t]
\centering
\begin{tikzpicture}
[
      draw,
      block/.style={draw, rectangle, color=blue, fill=shadeblue},
      tuple/.style={draw, ellipse, color=red, fill=shadered},
      downaw/.style={->,fill=grey,line width=0.35mm},
]
      \node[] (active1) at (0,6)
       {
          \begin{tabular}{ | c | c | c | }
                \multicolumn{3}{ c }{active} \\
                \hline
                mach & $t_{start}$ & $t_{end}$ \\
                \hline
                $\cdots$ & $\cdots$ & $\cdots$ \\
                \hline
            \end{tabular}
      };

      \node[block] (normalization) at (0,4.5) { Split };

      \node[] (active2) at (0,3)
      {
          \begin{tabular}{ | c | c | c | }
                \multicolumn{3}{ c }{result of split} \\
                \hline
                mach & $t_{start}$ & $t_{end}$ \\
                \hline
                $\cdots$ & $\cdots$ & $\cdots$ \\
                \hline
            \end{tabular}
      };

      \node[block] (coalesce) at (0,1.5) {Coalesce};

      \node[] (active3) at (0,0)
        {
          \begin{tabular}{ | c | c | c | }
                \multicolumn{3}{ c }{result of coalesce} \\
                \hline
                mach & $t_{start}$ & $t_{end}$ \\
                \hline
                $\cdots$ & $\cdots$ & $\cdots$ \\
                \hline
            \end{tabular}
      };

       \draw[downaw] (active1) to (normalization);
       \draw[downaw] (normalization) to (active2);
       \draw[downaw] (active2) to (coalesce);
       \draw[downaw] (coalesce) to (active3);

\end{tikzpicture}

  \caption{Coalesce after split}
  \label{fig:coalesce-normalization}
\end{figure}

\clearpage

\begin{figure}[t]
  \centering
\begin{minipage}{1\linewidth}
\lstset{style=psqlcolor,basicstyle=\upshape\ttfamily\scriptsize,numbers=left}
  \begin{lstlisting}
-- Count opening/closing intervals per change point
WITH
change_points (mach, $\#_{start}$, $\#_{end}$, t) AS
(
    SELECT mach,
           sum($\#_{start}$) AS $\#_{start}$,
           sum($\#_{end}$) AS $\#_{end}$,
           t
    FROM (
          SELECT mach,
                 count(*) AS $\#_{start}$,
                 0 AS $\#_{end}$,
                 $t_{start}$ AS t
          FROM active
          GROUP BY $t_{start}$, mach
          UNION ALL
          SELECT mach,
                 0 AS $\#_{start}$,
                 count(*) AS $\#_{end}$,
                 $t_{end}$ AS t
          FROM active
          GROUP BY $t_{end}$, mach)
    GROUP BY t, mach
),
-- Count the open intervals per tuple and time point
num_intervals (mach, $\#_{open}$, t) AS
(
    SELECT DISTINCT mach,
           sum($\#_{start}$) OVER w
           - sum($\#_{end}$) OVER w AS $\#_{open}$,
           t
    FROM change_points
    WINDOW w AS (PARTITION BY mach ORDER BY t
                 RANGE UNBOUNDED PRECEDING)
),
-- Compute the difference between the number of open
-- intervals at t and at the previous change point
diff_previous (mach, $\#_{open}$, diffPrevious, t) AS
(
    SELECT mach,
           $\#_{open}$,
           COALESCE($\#_{open}$ - (lag($\#_{open}$,1) OVER w,
                    -1) AS diffPrevious,
           t
    FROM num_intervals
    WINDOW w AS (PARTITION BY mach ORDER BY t)
),
-- Remove unchanged intervals
changed_intervals (mach, t, $\#_{open}$, diffPrevious) AS
(
    SELECT * FROM diff_previous
    WHERE diffPrevious != 0
),
-- Pair each change point with the following change point
pair_points (mach, $\#_{open}$, $t_{start}$, $t_{end}$) AS
(
    SELECT  mach, $\#_{open}$, $t_{start}$, $t_{end}$
    FROM (SELECT mach, $\#_{open}$, t AS $t_{start}$,
                 last_value(t) OVER w AS $t_{end}$
          FROM changed_intervals)
    WHERE $t_{end}$ IS NOT NULL
    WINDOW w AS (PARTITION BY mach
                 ORDER BY t
                 ROWS BETWEEN 1 FOLLOWING AND 1 FOLLOWING)
),
-- Create a sequence (1, ..., max(#open))
max_seq (n) AS
(
    SELECT n
    FROM (SELECT max($\#_{open}$) AS mopen FROM pair_points) x,
         generate_sequence(1,mopen) AS y(n)
),
-- Create the right number of duplicates for each tuple
SELECT mach, $t_{start}$, $t_{end}$
FROM pair_points p, max_seq s
WHERE p.$\#_{open}$ >= s.n
\end{lstlisting}
\end{minipage}
  \caption{Applying the SQL implementation of bag coalescing to the example table \lstinline!active!.}
  \label{fig:-bag-coalesce}
\end{figure}

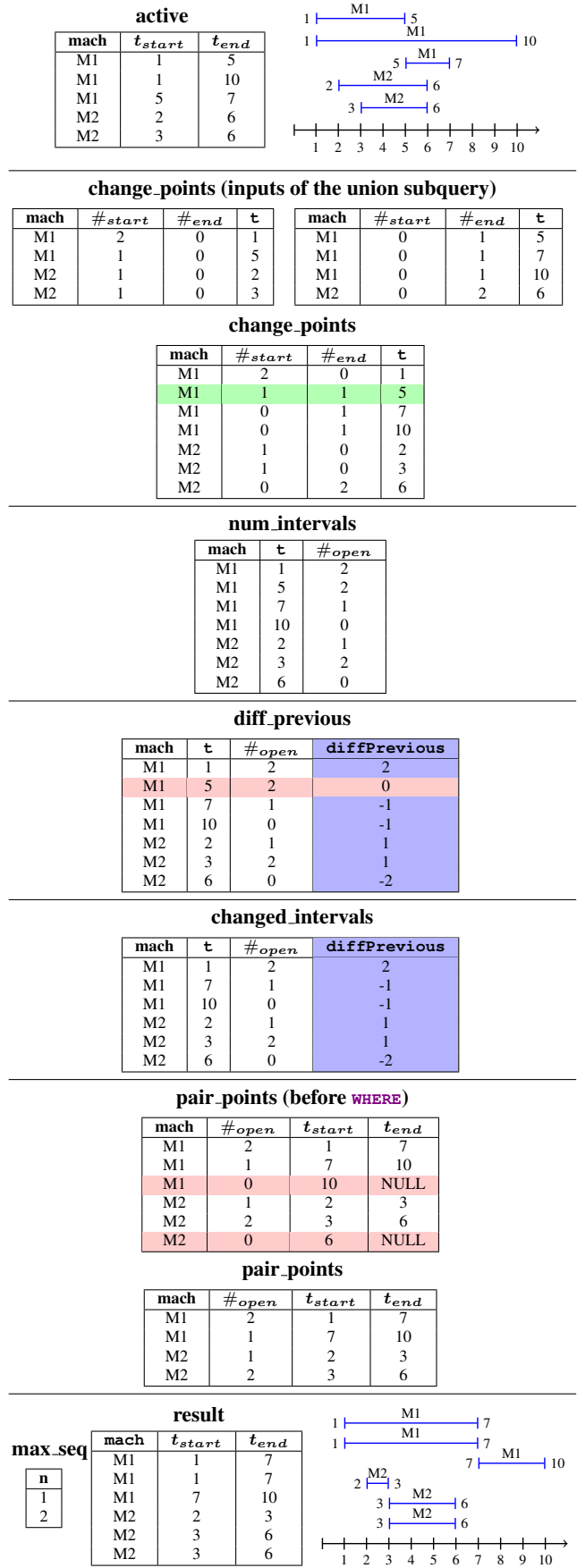
\begin{figure}[t]
 \centering
\begin{minipage}{1.0\linewidth}
\centering
\begin{minipage}{0.43\linewidth}
\centering
\textbf{active}\\[1mm]
{\scriptsize
  \begin{tabular}{|c|c|c|}
    \hline
\thead{mach} & \tthead{{$\bm{t_{start}}$}}& \tthead{{$\bm{t_{end}}$}}  \\ \hline
 M1 & 1 & 5  \\
 M1 & 1 & 10 \\
    M1 & 5 & 7 \\
  M2 & 2 & 6 \\
  M2 & 3 & 6 \\
 \hline
\end{tabular}
}
\end{minipage}
\vspace*{1mm}
\begin{minipage}{0.45\linewidth}
\centering
\begin{center}
  \resizebox{1\linewidth}{!}{
       \centering
    \begin{tikzpicture} [
    ]

    \def\xscaler{0.5}
    \def\yscaler{0.5}
    \def\axispos{-6 * \yscaler}

\draw[|->, thick] (-0,\axispos) -- (11 * \xscaler,\axispos);
\foreach \x in {1,...,10}
    \draw[thick] (\x * \xscaler,0.3 * \yscaler + \axispos) -- (\x * \xscaler,-0.3 * \yscaler + \axispos) node[below] {\x};

\foreach \b/\s/\pos/\tup in {1/5/1/M1,1/10/2/M1,5/7/3/M1,2/6/4/M2,3/6/5/M2}
    \draw[thick,blue,|-|]   (\b * \xscaler,-1 * \yscaler * \pos)  node[left,black]{\b} -- node[above,black]{\tup} (\s * \xscaler,-1 * \yscaler * \pos) node[right,black]{\s}
;

    \end{tikzpicture}
  }
\end{center}

\end{minipage} \\[1mm]
\hrule$ $\\[1mm]
\begin{minipage}{1\linewidth}
\centering
\textbf{change\_points (inputs of the union subquery)}\\[1mm]
{\scriptsize
\begin{minipage}{1.0\linewidth}
    \centering
    \begin{minipage}{0.49\linewidth}
    \begin{tabular}{|c|c|c|c|}
      \hline
      \thead{mach} & \tthead{{$\bm{\#_{start}}$}} & \tthead{{$\bm{\#_{end}}$}} & \tthead{{t}} \\
      \hline
      M1  & 2 & 0  & 1 \\
       M1  & 1 & 0 & 5 \\
 M2  & 1 & 0 & 2 \\
      M2  & 1 & 0 & 3 \\
 \hline
\end{tabular}
\end{minipage}
\begin{minipage}{0.49\linewidth}
\begin{tabular}{|c|c|c|c|}
      \hline
      \thead{mach} & \tthead{{$\bm{\#_{start}}$}} & \tthead{{$\bm{\#_{end}}$}} & \tthead{{t}} \\
      \hline
  M1  & 0 & 1 & 5  \\
   M1  & 0 & 1 & 7  \\
 M1  & 0 & 1 & 10 \\
 M2  & 0 & 2 & 6 \\
 \hline
\end{tabular}
\end{minipage}
\end{minipage}
}
\end{minipage} \\[1mm]
\begin{minipage}{1\linewidth}
\centering
\textbf{change\_points}\\[1mm]
{\scriptsize
    \begin{tabular}{|c|c|c|c|}
      \hline
      \thead{mach} & \tthead{{$\bm{\#_{start}}$}} & \tthead{{$\bm{\#_{end}}$}} & \tthead{{t}} \\
      \hline
 M1 & 2 & 0  & 1 \\
 \rowcolor{green!30} M1 & 1 & 1 & 5 \\
      M1 & 0 & 1 & 7  \\
      M1 & 0 & 1 & 10 \\
 M2 & 1 & 0 & 2 \\
      M2 & 1 & 0 & 3 \\
 M2 & 0 & 2 & 6 \\
 \hline
\end{tabular}
}
\end{minipage} \\[1mm]
\hrule$ $\\[1mm]
\begin{minipage}{1\linewidth}
\centering
\textbf{num\_intervals}\\[1mm]
{\scriptsize
    \begin{tabular}{|c|c|c|}
      \hline
      \thead{mach} & \tthead{t} & \tthead{$\bm{\#_{open}}$} \\
      \hline
 M1 & 1 & 2  \\
 M1 & 5 & 2  \\
 M1 & 7 & 1  \\
 M1 & 10 & 0   \\
 M2 & 2 & 1  \\
 M2 & 3 & 2 \\
 M2 & 6 & 0  \\
 \hline
\end{tabular}
}
\end{minipage} \\[1mm]
\hrule$ $\\[1mm]
\begin{minipage}{1\linewidth}
\centering
\textbf{diff\_previous}\\[1mm]
{\scriptsize
    \begin{tabular}{|c|c|c|>{\columncolor{blue!30}}c|}
      \hline
      \thead{mach} & \tthead{{t}} & \tthead{{$\bm{\#_{open}}$}} & \tthead{{diffPrevious}} \\
      \hline
 M1 & 1 & 2  & 2 \\
    \rowcolor{LightRed}
 M1 & 5 & 2 & 0 \\
 M1 & 7 & 1 & -1 \\
 M1 & 10 & 0 & -1  \\
 M2 & 2 & 1 & 1 \\
 M2 & 3 & 2 & 1 \\
 M2 & 6 & 0 & -2 \\
 \hline
\end{tabular}
}
\end{minipage} \\[1mm]
\hrule$ $\\[1mm]
\begin{minipage}{1\linewidth}
\centering
\textbf{changed\_intervals}\\[1mm]
{\scriptsize
    \begin{tabular}{|c|c|c|>{\columncolor{blue!30}}c|}
      \hline
      \thead{mach} & \tthead{{t}} & \tthead{{$\bm{\#_{open}}$}} & \tthead{{diffPrevious}} \\
      \hline
 M1 & 1 & 2  & 2 \\
 M1 & 7 & 1 & -1 \\
 M1 & 10 & 0 & -1  \\
 M2 & 2 & 1 & 1 \\
 M2 & 3 & 2 & 1 \\
 M2 & 6 & 0 & -2 \\
 \hline
\end{tabular}
}
\end{minipage} \\[1mm]
\hrule$ $\\[1mm]
\begin{minipage}{1\linewidth}
\centering
\textbf{pair\_points (before \lstinline!WHERE!)}\\[1mm]
{\scriptsize
    \begin{tabular}{|c|c|c|c|c|}
      \hline
      \thead{mach}& \tthead{{$\bm{\#_{open}}$}} & \tthead{{$\bm{t_{start}}$}} & \tthead{{$\bm{t_{end}}$}}  \\
      \hline
 M1 & 2 & 1 & 7 \\
 M1 & 1 & 7 & 10 \\
    \rowcolor{LightRed}
 M1 & 0 & 10 & NULL \\
 M2 & 1 & 2 & 3 \\
 M2 & 2 & 3 & 6 \\
   \rowcolor{LightRed}
 M2 & 0 & 6 & NULL \\
 \hline
\end{tabular}
}
\end{minipage} \\[1mm]
\begin{minipage}{1\linewidth}
\centering
\textbf{pair\_points}\\[1mm]
{\scriptsize
    \begin{tabular}{|c|c|c|c|c|}
      \hline
\thead{mach}& \tthead{{$\bm{\#_{open}}$}} & \tthead{{$\bm{t_{start}}$}} & \tthead{{$\bm{t_{end}}$}}  \\
      \hline
 M1 & 2 & 1 & 7 \\
 M1 & 1 & 7 & 10 \\
 M2 & 1 & 2 & 3 \\
 M2 & 2 & 3 & 6 \\
 \hline
\end{tabular}
}
\end{minipage} \\[1mm]
\hrule$ $\\[1mm]
\begin{minipage}{0.12\linewidth}
\centering
\textbf{max\_seq}\\[1mm]
{\scriptsize
    \begin{tabular}{|c|}
      \hline
      \thead{n} \\
      \hline
 1   \\
 2  \\
 \hline
\end{tabular}
}
\end{minipage}
\begin{minipage}{0.4\linewidth}
\centering
\textbf{result}\\[1mm]
{\scriptsize
    \begin{tabular}{|c|c|c|}
      \hline
       \tthead{{mach}}& \tthead{{$\bm{t_{start}}$}} & \tthead{{$\bm{t_{end}}$}} \\
      \hline
 M1 & 1 & 7  \\
 M1 & 1 & 7 \\
 M1 & 7 & 10 \\
 M2 & 2 & 3 \\
 M2 & 3 & 6 \\
 M2 & 3 & 6 \\
 \hline
\end{tabular}
}
\end{minipage}
\begin{minipage}{0.45\linewidth}
\centering
\begin{center}
  \resizebox{1\linewidth}{!}{
       \centering
    \begin{tikzpicture} [
    ]

    \def\xscaler{0.5}
    \def\yscaler{0.45}
    \def\axispos{-7 * \yscaler}

\draw[|->, thick] (-0,\axispos) -- (11 * \xscaler,\axispos);
\foreach \x in {1,...,10}
    \draw[thick] (\x * \xscaler,0.3 * \yscaler + \axispos) -- (\x * \xscaler,-0.3 * \yscaler + \axispos) node[below] {\x};

\foreach \b/\s/\pos/\tup in {1/7/1/M1,1/7/2/M1,7/10/3/M1,2/3/4/M2,3/6/5/M2,3/6/6/M2}
    \draw[thick,blue,|-|]   (\b * \xscaler,-1 * \yscaler * \pos)  node[left,black]{\b} -- node[above,black]{\tup} (\s * \xscaler,-1 * \yscaler * \pos) node[right,black]{\s}
;

    \end{tikzpicture}
  }
\end{center}
\end{minipage}
\end{minipage} \\[2mm]
\caption{Example database and intermediate results of the query implementing bag coalescing for table \lstinline!active!.}
\label{fig:intermeidate-result-bag-col}
\end{figure}

\begin{figure}
  \centering
\begin{minipage}{1\linewidth}
\lstset{style=psqlcolor,basicstyle=\scriptsize\upshape\ttfamily,numbers=left,tabsize=4}
  \begin{lstlisting}
-- name left and right inputs
WITH
left AS (
    SELECT * FROM active
),
right AS (
    SELECT * FROM active
),
-- Gather change points
end_points AS
(
	SELECT mach, $t_{start}$ AS t
	FROM left
	UNION
	SELECT mach, $t_{end}$ AS t
	FROM left
	UNION
	SELECT mach, $t_{start}$ AS t
	FROM right
	UNION
	SELECT mach, $t_{end}$ AS t
	FROM right
),
-- Gather intervals of LEFTY with a unique ID
interval_id AS
(
	SELECT row_number() OVER (ORDER BY 1) AS id,
           mach,
           $t_{start}$,
           $t_{end}$
	FROM left
),
-- Join intervals with change points
split_points AS
(
    SELECT l.id,
           l.mach,
           l.$t_{start}$,
           l.$t_{end}$,
           c.t
	FROM interval_id l,
	     end_points c
	WHERE c.mach = l.mach
	      AND c.T >= l.$t_{start}$
	      AND c.T < l.$t_{end}$
)
-- Produce output by input on change points
SELECT mach,
       t AS $t_{start}$,
       COALESCE(lead(t) OVER w, $t_{end}$) AS $t_{end}$
FROM split_points
WINDOW w AS (PARTITION BY id ORDER BY t) ;
\end{lstlisting}
\end{minipage}
  \caption{SQL implementation of the split operator applied to example table \lstinline!active!.}
  \label{fig:-normalization}
\end{figure}

\begin{figure}[t]
 \centering
\begin{minipage}{1.0\linewidth}
\centering
\begin{minipage}{1\linewidth}
\centering
\begin{minipage}{0.49\linewidth}
\centering
  \textbf{active}\\[1mm]
{\scriptsize
\begin{tabular}{|c|c|c|}
 \hline
\thead{mach} & \tthead{$\bm{t_{start}}$}& \tthead{$\bm{t_{end}}$}  \\
 \hline
 M1 & 1 & 7  \\
 M1 & 4 & 9 \\
 M2 & 2 & 8 \\
 \hline
\end{tabular}
}
\end{minipage}
\begin{minipage}{0.49\linewidth}
\begin{center}
  \resizebox{1\linewidth}{!}{
       \centering
    \begin{tikzpicture} [
    ]

    \def\xscaler{0.5}
    \def\yscaler{0.5}
    \def\axispos{-4 * \yscaler}

\draw[|->, thick] (-0,\axispos) -- (11 * \xscaler,\axispos);
\foreach \x in {1,...,10}
    \draw[thick] (\x * \xscaler,0.3 * \yscaler + \axispos) -- (\x * \xscaler,-0.3 * \yscaler + \axispos) node[below] {\x};

\foreach \b/\s/\pos/\tup in {1/7/1/M1,4/9/2/M1,2/8/3/M2}
    \draw[thick,blue,|-)]   (\b * \xscaler,-1 * \yscaler * \pos)  node[left,black]{\b} -- node[above,black]{\tup} (\s * \xscaler,-1 * \yscaler * \pos) node[right,black]{\s}
;

    \end{tikzpicture}
  }
\end{center}
\end{minipage}

\end{minipage}\\[1mm]
\hrule$ $\\[1mm]
\begin{minipage}{1\linewidth}
\centering
\textbf{change\_points}\\[1mm]
{\scriptsize
\begin{tabular}{|c|c|}
 \hline
\thead{mach} & \tthead{{t}} \\
 \hline
M1   & 1\\
 M1   & 4\\
 M1   & 7\\
 M1   & 9\\
 M2   & 2\\
 M2   & 8\\
  \hline
\end{tabular}
}
\end{minipage}\\[1mm]
\hrule$ $\\[1mm]
\begin{minipage}{1\linewidth}
\centering
\textbf{interval\_id}\\[1mm]
{\scriptsize
\begin{tabular}{|c|c|c|c|c|}
 \hline
\thead{id} & \tthead{{mach}}& \tthead{$\bm{t_{start}}$} & \tthead{$\bm{t_{end}}$} \\
 \hline
 1 & M1 & 1 & 7  \\
 2 & M1 & 4 & 9 \\
 3 & M2 & 2 & 8 \\
 \hline
\end{tabular}
}
\end{minipage}\\[1mm]
\hrule$ $\\[1mm]
\begin{minipage}{1\linewidth}
\centering
\textbf{time\_points}\\[1mm]
{\scriptsize
\begin{tabular}{|c|c|c|c|c|}
 \hline
\thead{id} & \tthead{{mach}}& \tthead{$\bm{t_{start}}$}& \tthead{$\bm{t_{end}}$}& \tthead{{t}} \\
  \hline
\rowcolor{green!30}  1 & M1   &      1 &    7 & 1\\
\rowcolor{green!30}  1 & M1   &      1 &    7 & 4\\
 \rowcolor{LightRed} 2 & M1   &      4 &    9 & 4\\
\rowcolor{LightRed}  2 & M1   &      4 &    9 & 7\\
\rowcolor{yellow}  3 & M2   &      2 &    8 & 2\\
 \hline
\end{tabular}
}
\end{minipage} \\[1mm]
\hrule$ $\\[1mm]
\begin{minipage}{1\linewidth}
\centering
\begin{minipage}{0.49\linewidth}
  \centering
\textbf{result}\\[1mm]
{\scriptsize
\begin{tabular}{|c|c|c|}
 \hline
  \thead{mach} & \tthead{$\bm{t_{start}}$}& \tthead{$\bm{t_{end}}$}\\
  \hline
  \rowcolor{green!30}
  M1 & 1 & 4\\
  \rowcolor{green!30}
  M1 & 4 & 7 \\
  \rowcolor{LightRed}
  M1 & 4 & 7  \\
  \rowcolor{LightRed}
  M1 & 7 & 9 \\
 \rowcolor{yellow}
  M2 & 2 & 8 \\
 \hline
\end{tabular}
}
\end{minipage}
\begin{minipage}{0.49\linewidth}
\begin{center}
  \resizebox{1\linewidth}{!}{
       \centering
    \begin{tikzpicture} [
    ]

    \def\xscaler{0.5}
    \def\yscaler{0.5}
    \def\axispos{-6 * \yscaler}

\draw[|->, thick] (-0,\axispos) -- (11 * \xscaler,\axispos);
\foreach \x in {1,...,10}
    \draw[thick] (\x * \xscaler,0.3 * \yscaler + \axispos) -- (\x * \xscaler,-0.3 * \yscaler + \axispos) node[below] {\x};

\foreach \b/\s/\pos/\tup in {1/4/1/M1,4/7/2/M1,4/7/3/M1,7/9/4/M1,2/8/5/M2}
    \draw[thick,blue,|-)]   (\b * \xscaler,-1 * \yscaler * \pos)  node[left,black]{\b} -- node[above,black]{\tup} (\s * \xscaler,-1 * \yscaler * \pos) node[right,black]{\s}
;

    \end{tikzpicture}
  }
\end{center}
\end{minipage}

\end{minipage} \\[1mm]
\hrule$ $\\[1mm]
\end{minipage} \\[2mm]
\caption{Example table \lstinline!active! and the (intermediate) results of the SQL query implementing the split operator.}
\label{fig:intermeidate-result-normalize}
\end{figure}
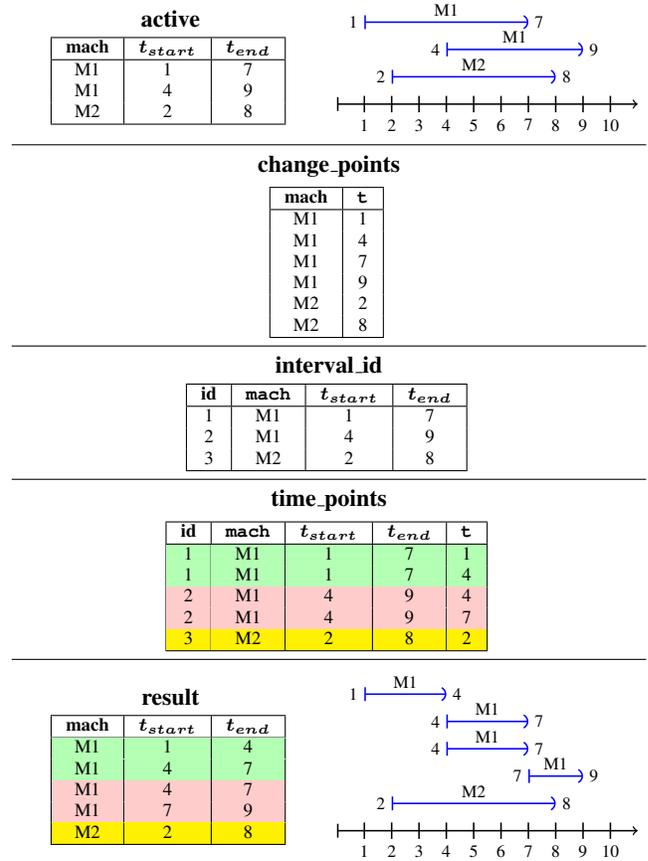

\clearpage

\begin{figure}
  \centering
\begin{minipage}{1\linewidth}
\lstset{tabsize=4,style=psqlcolor,basicstyle=\scriptsize\upshape\ttfamily,,numbers=left}
  \begin{lstlisting}
-- Pre-aggregate before splitting
WITH
pre_agg (mach, c, s, $t_{start}$, $t_{end}$) AS
(
	SELECT mach,
				 count(*) AS c,
				 sum(consum) AS s,
				 $t_{start}$,
				 $t_{end}$
	FROM active
	GROUP BY mach, $t_{start}$, $t_{end}$
),
-- Compute amount of increase/decrease at each time point
increase_decrease (mach, add_c, add_s, dec_c, dec_s, t) AS
(
	SELECT mach,
				 sum(add_c) AS add_c,
				 sum(add_s) AS add_s,
				 sum(dec_c) AS dec_c,
				 sum(dec_s) AS dec_s,
				 t
	FROM (SELECT mach,
							 c AS add_c,
				 			 s AS add_s,
				 			 0 AS dec_c,
				 			 0 AS dec_s,
				 			 $t_{start}$ AS t
				FROM pre_agg
				UNION ALL
				SELECT mach,
							 0 AS add_c,
				 			 0 AS add_s,
				 			 c AS dec_c,
				 			 s AS dec_s,
				 			 $t_{end}$ AS t
				FROM pre_agg)
	GROUP BY mach, t
),
-- Calculate accumulative total for interval start
-- points up to and including time point t and
-- subtract the total for "closing" intervals
accumulation (mach, c, s, t) AS
(
    SELECT mach
				   sum(add_c) OVER w
				   - sum(dec_c) OVER w AS c,
				   sum(add_s) OVER w
				   - sum(dec_s) OVER w AS s,
				   t
    FROM increase_decrease
    WINDOW w AS (PARTITION BY mach
                 ORDER BY t
                 RANGE UNBOUNDED PRECEDING)
),
-- output results for adjacent "split" points
SELECT mach, avg_con, $t_{start}$, $t_{end}$
FROM (SELECT mach,
             c
             CASE WHEN (c = 0) THEN NULL
                  ELSE s / c END AS avg_con,
             t AS $t_{start}$,
             last_value(t) OVER w AS $t_{end}$
      FROM accumulation
      WINDOW w AS (PARTITION BY mach
                   ORDER BY t
                   ROWS BETWEEN 1 FOLLOWING AND 1 FOLLOWING)
)
WHERE c > 0
\end{lstlisting}
\end{minipage}
  \caption{Example SQL implementation of split + aggregation (average) applied to example table \lstinline!active!}
  \label{fig:-normalization_agg}
\end{figure}

\begin{figure}
  \centering
\begin{minipage}{1\linewidth}
\lstset{basicstyle=\scriptsize\upshape\ttfamily, commentstyle=\color{blue},numbers=left}
  \begin{lstlisting}
SELECT mach, avg(consum) as avg_con
FROM active
GROUP BY mach;
\end{lstlisting}
\end{minipage}
  \caption{Example aggregation query}
  \label{fig:-query-normalization_agg}
\end{figure}

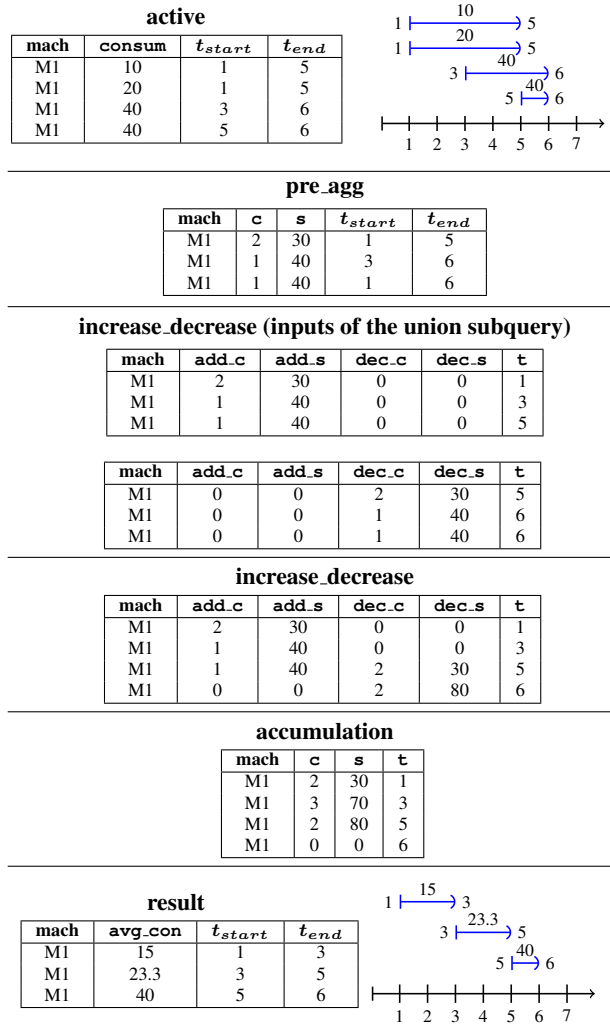
\begin{figure}[t]
 \centering
\begin{minipage}{1.0\linewidth}
\centering
\begin{minipage}{1\linewidth}
\centering
\begin{minipage}{0.52\linewidth}
  \centering
\textbf{active}\\[1mm]
{\scriptsize
  \begin{tabular}{|c|c|c|c|c|}
    \hline
\thead{mach} & \tthead{consum} & \tthead{$\bm{t_{start}}$} & \tthead{$\bm{t_{end}}$}  \\ \hline
 M1 & 10 & 1 & 5  \\
 M1 & 20 & 1 & 5  \\
 M1 & 40 & 3 & 6 \\
 M1 & 40 & 5 & 6 \\
 \hline
\end{tabular}
}
\end{minipage}
\begin{minipage}{0.46\linewidth}
\begin{center}
  \resizebox{0.8\linewidth}{!}{
       \centering
    \begin{tikzpicture} [
    ]

    \def\xscaler{0.5}
    \def\yscaler{0.45}
    \def\axispos{-5 * \yscaler}

\draw[|->, thick] (-0,\axispos) -- (8 * \xscaler,\axispos);
\foreach \x in {1,...,7}
    \draw[thick] (\x * \xscaler,0.3 * \yscaler + \axispos) -- (\x * \xscaler,-0.3 * \yscaler + \axispos) node[below] {\x};

\foreach \b/\s/\pos/\tup in {1/5/1/10,1/5/2/20,3/6/3/40,5/6/4/40}
    \draw[thick,blue,|-)]   (\b * \xscaler,-1 * \yscaler * \pos)  node[left,black]{\b} -- node[above,black]{\tup} (\s * \xscaler,-1 * \yscaler * \pos) node[right,black]{\s}
;

    \end{tikzpicture}
  }
\end{center}
\end{minipage}

\end{minipage} \\[2mm]
\hrule$ $\\[1mm]
\begin{minipage}{1\linewidth}
\centering
\textbf{pre\_agg}\\[1mm]
{\scriptsize
\begin{minipage}{1.0\linewidth}
    \centering
    \begin{tabular}{|c|c|c|c|c|}
      \hline
      \thead{mach} & \tthead{c} & \tthead{s} & \tthead{$\bm{t_{start}}$} & \tthead{$\bm{t_{end}}$} \\
      \hline
M1  & 2 & 30 & 1  & 5 \\
M1  & 1 & 40 & 3  & 6 \\
M1  & 1 & 40 & 1  & 6 \\
 \hline
\end{tabular}\\[3mm]
\end{minipage}
}
\end{minipage} \\[1mm]
\hrule$ $\\[1mm]
\begin{minipage}{1\linewidth}
\centering
\textbf{increase\_decrease (inputs of the union subquery)}\\[1mm]
{\scriptsize
    \begin{tabular}{|c|c|c|c|c|c|}
      \hline
      \thead{mach} & \tthead{add\_c} & \tthead{add\_s} & \tthead{dec\_c} & \tthead{dec\_s} & \tthead{t} \\
      \hline
 M1 & 2 & 30 & 0  & 0 & 1 \\
 M1 & 1 & 40 & 0  & 0 & 3\\
 M1 & 1 & 40 & 0  & 0 & 5\\
 \hline
\end{tabular} \\[3mm]
\begin{tabular}{|c|c|c|c|c|c|}
      \hline
      \thead{mach} & \tthead{add\_c} & \tthead{add\_s} & \tthead{dec\_c} & \tthead{dec\_s} & \tthead{t} \\
      \hline
 M1 & 0 & 0 & 2  & 30 & 5 \\
 M1 & 0 & 0 & 1  & 40 & 6\\
 M1 & 0 & 0 & 1  & 40 & 6\\
 \hline
\end{tabular}
}
\end{minipage} \\[1mm]
\hrule$ $\\[1mm]
\begin{minipage}{1\linewidth}
\centering
\textbf{increase\_decrease}\\[1mm]
{\scriptsize
\begin{tabular}{|c|c|c|c|c|c|}
      \hline
      \thead{mach} & \tthead{add\_c} & \tthead{add\_s} & \tthead{dec\_c} & \tthead{dec\_s} & \tthead{t} \\
      \hline
 M1 & 2 & 30 & 0  & 0 & 1 \\
 M1 & 1 & 40 & 0  & 0 & 3\\
 M1 & 1 & 40 & 2  & 30 & 5\\
 M1 & 0 & 0 & 2  & 80 & 6\\
 \hline
\end{tabular}
}
\end{minipage} \\[1mm]
\hrule$ $\\[1mm]
\begin{minipage}{1\linewidth}
\centering
\textbf{accumulation}\\[1mm]
{\scriptsize
\begin{tabular}{|c|c|c|c|}
      \hline
      \thead{mach} & \tthead{c} & \tthead{s} & \tthead{t} \\
      \hline
 M1 & 2 & 30 & 1   \\
 M1 & 3 & 70 & 3  \\
 M1 & 2 & 80 & 5  \\
 M1 & 0 & 0 & 6  \\
 \hline
\end{tabular}
}
\end{minipage} \\[1mm]
\hrule$ $\\[1mm]
\begin{minipage}{0.49\linewidth}
\centering
\textbf{result}\\[1mm]
{\scriptsize
\begin{tabular}{|c|c|c|c|}
      \hline
      \thead{mach} & \tthead{avg\_con} & \tthead{$\bm{t_{start}}$} & \tthead{$\bm{t_{end}}$} \\
      \hline
 M1 & 15 & 1 & 3   \\
 M1 & 23.3 & 3 & 5  \\
 M1 & 40 & 5 & 6  \\
 \hline
\end{tabular}
}
\end{minipage}
\begin{minipage}{0.46\linewidth}
\begin{center}
  \resizebox{0.8\linewidth}{!}{
       \centering
    \begin{tikzpicture} [
    ]

    \def\xscaler{0.5}
    \def\yscaler{0.55}
    \def\axispos{-4 * \yscaler}

\draw[|->, thick] (-0,\axispos) -- (8 * \xscaler,\axispos);
\foreach \x in {1,...,7}
    \draw[thick] (\x * \xscaler,0.3 * \yscaler + \axispos) -- (\x * \xscaler,-0.3 * \yscaler + \axispos) node[below] {\x};

\foreach \b/\s/\pos/\tup in {1/3/1/15,3/5/2/23.3,5/6/3/40}
    \draw[thick,blue,|-)]   (\b * \xscaler,-1 * \yscaler * \pos)  node[left,black]{\b} -- node[above,black]{\tup} (\s * \xscaler,-1 * \yscaler * \pos) node[right,black]{\s}
;

    \end{tikzpicture}
  }
\end{center}
\end{minipage}
\\[1mm]
\end{minipage} \\[2mm]
\caption{Intermediate results of the query implementing split + aggregation for the query from Figure~\ref{fig:-query-normalization_agg}}
\label{fig:intermeidate-result-normalization-agg}
\end{figure}

\clearpage

\begin{figure}
  \centering
\begin{minipage}{1\linewidth}
\lstset{style=psqlcolor,basicstyle=\scriptsize\upshape\ttfamily,numbers=left}
  \begin{lstlisting}
WITH
left (mach, consum, $t_{start}$, $t_{end}$) AS
(
	SELECT mach, consum, $t_{start}$, $t_{end}$
	FROM active
),
right (mach, consum, $t_{start}$, $t_{end}$) AS
(
	SELECT mach, consum, $t_{start}$, $t_{end}$
	FROM faulty
),
-- Count opening and closing intervals for each interval
-- end point.  Intervals from the right input are
-- counted negatively.
end_point_counts (mach, consume, $\#_{open}$, $\#_{close}$, t) AS
(
	SELECT mach, consum, $t_{start}$ AS t,
				 count(*) AS $\#_{open}$, 0 AS $\#_{close}$
	FROM left
	GROUP BY $t_{start}$, mach, consum
	UNION ALL
	SELECT mach, consum, $t_{end}$ AS t,
				 0 AS $\#_{open}$, count(*) AS $\#_{close}$
	FROM left
	GROUP BY $t_{end}$, mach, consum
	UNION ALL
	SELECT mach, consum, $t_{start}$ AS t,
				 - count(*) AS $\#_{open}$, 0 AS $\#_{close}$
	FROM right
	GROUP BY $t_{start}$, mach, consum
	UNION ALL
	SELECT mach, consum, $t_{end}$ AS t,
				 0 AS $\#_{open}$, - count(*) AS $\#_{close}$
	FROM right
	GROUP BY $t_{end}$, mach, consum
),
-- Accumulate counts to get multiplicities
acc_counts (mach, consum, t, $\#_{open}$, $\#_{close}$) AS
(
	SELECT mach,
				 consum,
				 sum($\#_{open}$) AS $\#_{open}$,
				 sum($\#_{close}$) AS $\#_{close}$,
				 t
	FROM end_point_counts
	GROUP BY t, mach, consum
),
-- Produce intervals with the corresponding multiplicities
intervals (mach, consume, $t_{start}$, $t_{end}$, multiplicity) AS
(
    SELECT mach,
           consum,
           t AS $t_{start}$,
           lead(t) OVER w1 AS $t_{end}$,
           sum($\#_{open}$) OVER w2
           - sum($\#_{close}$) OVER w2 AS multiplicity
    FROM acc_counts
    WINDOW w1 AS (PARTITION BY mach, consum
                  ORDER BY t),
           w2 AS (PARTITION BY mach, consum
                  ORDER BY t
                  RANGE UNBOUNDED PRECEDING)
),
-- Compute max multiplicity
max_seq (n) AS
(
    SELECT n
    FROM (SELECT max(numOpen) AS max_open FROM intervals) x,
          generate_sequence(1,max_open) AS y(n)
)
-- Produce duplicates based on multiplicities
SELECT mach, consum, $t_{start}$, $t_{end}$
FROM intervals i, max_seq m
WHERE multiplicity > 0 AND i.multiplicity >= m.n;
\end{lstlisting}
\end{minipage}
  \caption{SQL implementation of split + bag difference applied to example table \lstinline!active!.}
  \label{fig:-normalization_sefdiff}
\end{figure}

\begin{figure}
  \centering
\begin{minipage}{1\linewidth}
\lstset{basicstyle=\scriptsize\upshape\ttfamily, commentstyle=\color{blue},numbers=left}
  \begin{lstlisting}
SELECT mach, consum, ${t_{start}}$, $t_{end}$
FROM active
EXCEPT ALL
SELECT mach, consum, $t_{start}$, $t_{end}$
FROM faulty;
\end{lstlisting}
\end{minipage}
  \caption{Example query using bag difference.}
  \label{fig:-query-normalization_sefdiff}
\end{figure}

\begin{figure}[t]
 \centering
\begin{minipage}{1.0\linewidth}
\centering
\begin{minipage}{1\linewidth}
\centering
\begin{minipage}{0.49\linewidth}
  \centering
\textbf{active (left input)}\\[1mm]
{\scriptsize
  \begin{tabular}{|c|c|c|c|c|}
    \hline
\thead{mach} & \tthead{consum} & \tthead{$t_{start}$} & \tthead{$t_{end}$}  \\ \hline
 M1 & 20 & 1 & 5  \\
 M1 & 40 & 1 & 7 \\
 M1 & 40 & 1 & 9 \\
 \hline
\end{tabular}
}
\end{minipage}
\begin{minipage}{0.49\linewidth}
\begin{center}
  \resizebox{0.8\linewidth}{!}{
       \centering
    \begin{tikzpicture} [
    ]

    \def\xscaler{0.5}
    \def\yscaler{0.45}
    \def\axispos{-4 * \yscaler}

\draw[|->, thick] (-0,\axispos) -- (10 * \xscaler,\axispos);
\foreach \x in {1,...,9}
    \draw[thick] (\x * \xscaler,0.3 * \yscaler + \axispos) -- (\x * \xscaler,-0.3 * \yscaler + \axispos) node[below] {\x};

\foreach \b/\s/\pos/\tup in {1/5/1/20,1/7/2/40,1/9/3/40}
    \draw[thick,blue,|-)]   (\b * \xscaler,-1 * \yscaler * \pos)  node[left,black]{\b} -- node[above,black]{\tup} (\s * \xscaler,-1 * \yscaler * \pos) node[right,black]{\s}
;

    \end{tikzpicture}
  }
\end{center}
\end{minipage}
\end{minipage} \\[2mm]
\begin{minipage}{1\linewidth}
\centering
\begin{minipage}{0.49\linewidth}
  \centering
\textbf{faulty (right)}\\[1mm]
{\scriptsize
  \begin{tabular}{|c|c|c|c|c|}
    \hline
\thead{mach} & \tthead{consum} & \tthead{$\bm{t_{start}}$} & \tthead{$\bm{t_{end}}$}  \\ \hline
 M1 & 20 & 2 & 6  \\
M1 & 40 & 3 & 5  \\
    \hline
\end{tabular}
}
\end{minipage}
\begin{minipage}{0.49\linewidth}
\begin{center}
  \resizebox{0.8\linewidth}{!}{
       \centering
    \begin{tikzpicture} [
    ]

    \def\xscaler{0.5}
    \def\yscaler{0.55}
    \def\axispos{-3 * \yscaler}

\draw[|->, thick] (-0,\axispos) -- (10 * \xscaler,\axispos);
\foreach \x in {1,...,9}
    \draw[thick] (\x * \xscaler,0.3 * \yscaler + \axispos) -- (\x * \xscaler,-0.3 * \yscaler + \axispos) node[below] {\x};

\foreach \b/\s/\pos/\tup in {2/6/1/20,3/9/2/40}
    \draw[thick,blue,|-)]   (\b * \xscaler,-1 * \yscaler * \pos)  node[left,black]{\b} -- node[above,black]{\tup} (\s * \xscaler,-1 * \yscaler * \pos) node[right,black]{\s}
;

    \end{tikzpicture}
  }
\end{center}
\end{minipage}
\end{minipage} \\[1mm]
\hrule$ $\\[1mm]
\begin{minipage}{1\linewidth}
\centering
\textbf{end\_point\_counts}\\[1mm]
{\scriptsize
\begin{minipage}{1.0\linewidth}
    \centering
    \begin{tabular}{|c|c|c|c|c|}
      \hline
      \thead{mach} & \tthead{consum} & \tthead{t} & \tthead{$\bm{\#_{open}}$} & \tthead{$\bm{\#_{close}}$} \\
      \hline
M1  & 20 & 1 & 1  & 0 \\
M1  & 40 & 1 & 2  & 0 \\
M1  & 20 & 5 & 0  & 1 \\
M1  & 40 & 7 & 0  & 1 \\
M1  & 40 & 9 & 0  & 1 \\
 \hline
M1  & 20 & 2 & -1  & 0 \\
M1  & 40 & 3 & -1  & 0 \\
M1  & 20 & 6 & 0  & -1 \\
M1  & 40 & 5 & 0  & -1 \\
      \hline
\end{tabular}\\[3mm]
\end{minipage}
}
\end{minipage} \\[1mm]
\hrule$ $\\[1mm]
\begin{minipage}{1\linewidth}
\centering
\textbf{acc\_counts}\\[1mm]
{\scriptsize
    \begin{tabular}{|c|c|c|c|c|}
      \hline
      \thead{mach} & \tthead{consum} & \tthead{t} & \tthead{$\bm{\#_{open}}$} & \tthead{$\bm{\#_{close}}$} \\
      \hline
 M1 & 20 & 1 & 1  & 0 \\
 M1 & 20 & 2 & -1  & 0 \\
 M1 & 20 & 5 & 0  & 1 \\
 M1 & 20 & 6 & 0  & -1 \\
 M1 & 40 & 1 & 2  & 0 \\
 M1 & 40 & 3 & -1  & 0 \\
 M1 & 40 & 7 & 0  & 1 \\
 M1 & 40 & 9 & 0  & 0 \\
      \hline
\end{tabular}
}
\end{minipage} \\[1mm]
\hrule$ $\\[1mm]
\begin{minipage}{1\linewidth}
\centering
\textbf{intervals}\\[1mm]
{\scriptsize
    \begin{tabular}{|c|c|c|c|c|}
      \hline
      \thead{mach} & \tthead{consum} & \tthead{$\bm{t_{start}}$} & \tthead{$\bm{t_{end}}$} & \tthead{multiplicity}\\
      \hline
 M1 & 20 & 1 & 2  & 1 \\
 M1 & 20 & 2 & 5  & 0 \\
 M1 & 20 & 5 & 6  & -1 \\
 M1 & 20 & 6 & NULL  & 0 \\
 M1 & 40 & 1 & 3  & 2 \\
 M1 & 40 & 3 & 7  & 1 \\
 M1 & 40 & 7 & 9  & 0 \\
 M1 & 40 & 9 & NULL  & 0 \\
      \hline
\end{tabular}
}
\end{minipage} \\[1mm]
\hrule$ $\\[1mm]
\begin{minipage}{0.13\linewidth}
\centering
\textbf{max\_seq}\\[1mm]
{\scriptsize
    \begin{tabular}{|c|}
      \hline
      \thead{n} \\
      \hline
 1   \\
 2  \\
 \hline
\end{tabular}
}
\end{minipage}
\begin{minipage}{0.48\linewidth}
\centering
\textbf{result}\\[1mm]
{\scriptsize
    \begin{tabular}{|c|c|c|c|}
      \hline
       \tthead{mach}& \tthead{consum} & \tthead{$\bm{t_{start}}$} & \tthead{$\bm{t_{end}}$}\\
      \hline
 M1 & 20 & 1 & 2  \\
 M1 & 40 & 1 & 3  \\
 M1 & 40 & 1 & 3  \\
 M1 & 40 & 3 & 7  \\
 \hline
\end{tabular}
}
\end{minipage}
\begin{minipage}{0.35\linewidth}
\begin{center}
  \resizebox{0.8\linewidth}{!}{
       \centering
    \begin{tikzpicture} [
    ]

    \def\xscaler{0.5}
    \def\yscaler{0.55}
    \def\axispos{-5 * \yscaler}

\draw[|->, thick] (-0,\axispos) -- (8 * \xscaler,\axispos);
\foreach \x in {1,...,7}
    \draw[thick] (\x * \xscaler,0.3 * \yscaler + \axispos) -- (\x * \xscaler,-0.3 * \yscaler + \axispos) node[below] {\x};

\foreach \b/\s/\pos/\tup in {1/2/1/20,1/3/2/40,1/3/3/40,3/7/4/40}
    \draw[thick,blue,|-)]   (\b * \xscaler,-1 * \yscaler * \pos)  node[left,black]{\b} -- node[above,black]{\tup} (\s * \xscaler,-1 * \yscaler * \pos) node[right,black]{\s}
;

    \end{tikzpicture}
  }
\end{center}
\end{minipage}
\end{minipage} \\[2mm]
\caption{Example instance of table \lstinline!active! and intermediate results of the query implementing split + bag difference for the query from Figure~\ref{fig:-query-normalization_sefdiff}}
\label{fig:intermeidate-result-normalization-set-diff}
\end{figure}
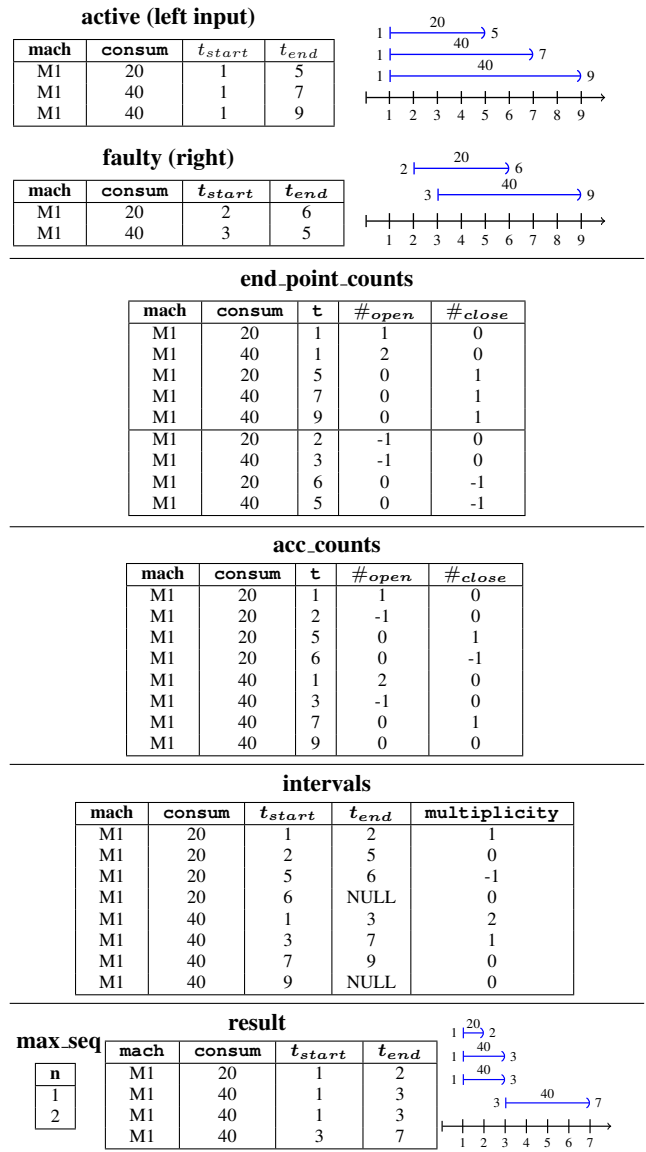

\end{document}